%% file: main.tex
\documentclass[graybox, envcountchap]{svmult}

\usepackage{makeidx}         
\usepackage{graphicx}        
\usepackage{multicol}        
\usepackage[bottom]{footmisc}

\usepackage{newtxtext}       %
\usepackage[varvw]{newtxmath}       

\makeindex             

\def\bC{{\mathbb C}}
\def\bQ{{\mathbb Q}}
\def\bR{{\mathbb R}}
\def\bN{{\mathbb N}}
\def\bZ{{\mathbb Z}}

\spnewtheorem{thm}{Theorem}{\itshape}{\rmfamily}
\spnewtheorem{ex}[thm]{Example}{\itshape}{\rmfamily}
\spnewtheorem{defi}[thm]{Definition}{\itshape}{\rmfamily}
\spnewtheorem{lem}[thm]{Lemma}{\itshape}{\rmfamily}
\spnewtheorem{pro}[thm]{Problem}{\itshape}{\rmfamily}
\spnewtheorem{cor}[thm]{Corollary}{\itshape}{\rmfamily}
\spnewtheorem{rem}[thm]{Remark}{\itshape}{\rmfamily}
\spnewtheorem{alg}[thm]{Algorithm}{\itshape}{\rmfamily}

\DeclareMathOperator{\resultant}{resultant}
\DeclareMathOperator{\ord}{ord}

\begin{document}

\frontmatter

\tableofcontents

\mainmatter

\include{chapter}

\backmatter
\appendix

\printindex


\end{document}

%% file: chapter.tex
%
%
%

%
%
%
%
%
%
%

\let\set\mathbb
\def\<#1>{\langle#1\rangle}

\title{Creative Telescoping}
\author{Shaoshi Chen\orcidID{0000-0001-8756-3006},\\
  Manuel Kauers\orcidID{0000-0001-8641-6661}, and\\
  Christoph Koutschan\orcidID{0000-0003-1135-3082}}
\institute{Shaoshi Chen \at KLMM, Academy of Mathematics and Systems Science, Chinese Academy of Sciences, Beijing, China, \email{schen@amss.ac.cn}
  \and Manuel Kauers \at Johannes Kepler University, Linz, Austria, \email{manuel.kauers@jku.at}
  \and Christoph Koutschan \at Johann Radon Institute for Computational and Applied Mathematics, Linz, Austria, \email{christoph.koutschan@ricam.oeaw.ac.at}}

\thanks{S.\ Chen was partially supported by the National Key R\&D Program of China (No.\ 2023YFA1009401), 
the NSFC grant (No.\ 12271511), and the CAS Funds of the Youth Innovation Promotion Association (No.\ Y2022001).
M.\ Kauers and C.\ Koutschan were supported by the Austrian FWF grant 10.55776/I6130. M.\ Kauers was also supported by the Austrian FWF grants 10.55776/PAT8258123 and 10.55776/PAT9952223).  All authors also were supported by the International Partnership Program of Chinese Academy of Sciences (Grant No.\ 167GJHZ2023001FN).}

%
%
\maketitle

\abstract{These notes on creative telescoping is based on a series of lectures at the Institut Henri Poincar\'e in November and December 2023.}

\section{Introduction}\label{sec:intro}

This chapter contains lecture notes of a series of lectures that have
been held at the Institut Henri Poincar\'e in November and December 2023.
Our goal was to teach the basics of creative telescoping, review some
historical developments, and discuss recent advancements and future trends
of the topic. This chapter does not contain any original research
results that are not published elsewhere.

The field of mathematics that has most profited from the invention
of the creative telescoping method~\cite{zeilberger91}
is probably combinatorics. 
Let us recall some basic mathematical quantities that appear frequently
in the enumeration of combinatorial objects and that we will use in our
examples throughout this chapter:
\begin{itemize}
\item The \emph{exponential function} $2^n$ counts the
  number of $\{0,1\}$-vectors of length~$n$.
\item The \emph{factorial} $n! := 1\cdot2\cdots n$ counts the number
  of permutations of $n$ elements (often, the more general Gamma function~$\Gamma(n)$
  is used, with $n!=\Gamma(n+1)$ whenever $n\in\set N$).
\item The \emph{Pochhammer symbol}
  $(a)_n := a\cdot(a+1)\cdots(a+n-1)$
  (also called ``rising factorial'') can be defined in terms of
  the Gamma function: $(a)_n = \Gamma(a+n)/\Gamma(a)$.
\item The \emph{binomial coefficient}
  $\binom{n}{k} := \frac{n!}{k!\cdot(n-k)!} = \frac{(n-k+1)_k}{(1)_k}$
  counts the number of ways to choose $k$ elements from a set of $n$ elements.
\item The \emph{Catalan numbers}
  $C_n := \frac{1}{n+1}\binom{2n}{n}$
  count the number of Dyck paths of length~$2n$, or alternatively
  binary trees with $n$ internal nodes, and many other things.
\item The \emph{Stirling numbers} (of the second kind) $S(n,k)$
  (or $\genfrac{\{}{\}}{0pt}{}{m}{i}$)
  count the number of partitions of an $n$-set into $k$ non-empty subsets.
\end{itemize}
Most of these quantities have the nice property of being \emph{hypergeometric}
(exercise: find out which one is not!). We say that an expression~$f(n)$ is
hypergeometric if its shift quotient is a rational function in~$n$, in symbols:
\[
  \frac{f(n+1)}{f(n)}\in\set Q(n).
\]

Combinatorial identities that originate from counting problems often take the
form of hypergeometric summations, such as:
\begin{align*}
  \sum_{k=0}^n \binom{n}{k} &= 2^n, \\
  \sum_{k=0}^{n-1} C_k\cdot C_{n-k-1} &= C_n, \\
  \sum_{k=0}^n\binom{n}{k}^{\!\!2}\binom{k+n}{k}^{\!2} &=
  \sum_{k=0}^n\binom{n}{k}\binom{k+n}{k}\sum_{j=0}^k\binom{k}{j}^{\!3}.
\end{align*}
Such identities can nowadays be proven
in an automatic and mechanical way~\cite{Wilf1992},
and it is one of the purposes of this chapter to explain how this is done;
see in particular Sections~\ref{sec:Celine}, \ref{sec:Gosper}, and~\ref{sec:Zeilberger}.

A meanwhile classical example of a hypergeometric sum
is the following 
\[
  b_n := \sum_{k=0}^n \binom{n}{k}^2 \binom{n+k}{k}^2
\]
that played a crucial role in Roger Ap\'{e}ry's proof~\cite{vanDerPoorten79}
of the irrationality of $\zeta(3)$. The task was to show
that the quantity~$b_n$ satisfies the second-order recurrence
\[
  (n+2)^3 b_{n+2} = (2 n+3) (17 n^2+51 n+39) b_{n+1} - (n+1)^3 b_n.
\]
Alf van der Poorten wrote: ``Neither Cohen nor I had been able to prove [\dots]
in the intervening 2 months.  After a few days of fruitless effort the
specific problem was mentioned to Don Zagier (Bonn), and with irritating speed
he showed that indeed the sequence $(b_n)$ satisfies the recurrence''.
With the algorithms that will be presented in this chapter, this
formerly demanding task is nowadays completely routine.


However, hypergeometric summations are just the starting point: the same
algorithmic ideas work in the much more general setting of D-finite
functions (see Sections~\ref{sec:DfiniteUniv} and~\ref{sec:DfiniteMult}),
which allows us to prove large classes of special function identities
with the help of the computer.

Generally speaking, special functions
are functions that (1) arise in real-world phenomena (e.g., physics)
and as such are typically solutions to certain differential equations,
but (2) cannot be expressed in terms of the usual elementary functions
($\sqrt{\phantom{x}}$, $\exp$, $\log$, $\sin$, $\cos$, \dots).
If such a function seems important enough it will
receive its own name and be considered a ``special function''.
Examples for special functions are: the Airy function that describes
the intensity of light in the neighborhood of a caustic, the Bessel function
that represents the modes of vibration of a thin circular acoustic membrane,
or the Coulomb wave function that describes the behavior of charged
particles in a Coulomb potential.

The \emph{holonomic systems approach} that was proposed by Zeilberger in his
1990 seminal paper~\cite{Zeilberger1990} can deal with large classes of special
functions (so-called holonomic functions).
Apart from having many applications in mathematics, physics, and
elsewhere, it created a large research area within symbolic computation
(computer algebra). A large portion of special function identities that
were formerly tabulated in voluminous books 
\cite{AbramowitzStegun1964,GradshteynRyzhik14,Olver10,PBM1986}
can nowadays be proven in an algorithmic way.
In Sections~\ref{sec:AdvClosureP} and~\ref{sec:Chyzak} we explain how this works.

The successor of the classical Handbook of Mathematical
Functions~\cite{AbramowitzStegun1964} is the
Digital Library of Mathematical Functions (DLMF)~\cite{DLMF}. During its
development, the following happened: on May 18, 2005, Frank Olver,
the mathematics editor of DLMF, sent an e-mail to Peter Paule, asking
whether computer algebra methods could provide automatic verifications
of some identities that were listed in DLMF but whose proofs had been
lost since the author of the corresponding chapter, Henry Antosiewiecz,
had passed away. These identities involve (spherical) Bessel functions,
sine and cosine integrals, Legendre polynomials, and other special
functions. In the following we display only a few of them (in total,
there were about a dozen of such identities):
\begin{align*}
  \frac1z\sin\sqrt{z^2+2zt}
  &=\sum_{n=0}^\infty\frac{(-t)^n}{n!}y_{n-1}(z),
  \\
  \left[\frac{\partial}{\partial \nu} j_\nu(z)\right]_{\nu=0}
  &=\frac1z\bigl(\operatorname{Ci}(2z)\sin z-\operatorname{Si}(2z)\cos z\bigr),
  \\
  J_0(z\sin\theta)
  &=\sum_{n=0}^\infty(4n+1)\frac{(2n)!}{2^{2n}n!^2}j_{2n}(z)P_{2n}(\cos\theta).
\end{align*}
Within two weeks, all identities were proven with computer algebra, by 
the members of Paule's algorithmic combinatorics group at RISC~\cite{GKKPSZ13},
using the algorithms that we are going to present below.

Special functions do not only appear in physics, but are also
frequently used in mathematical analysis, thanks to their orthogonality
properties. A particular application where computer algebra methods
turned out to be very
useful~\cite{KoutschanLehrenfeldSchoeberl12,SchoeberlKoutschanPaule15}
was in the context of simulating
the propagation of electromagnetic waves according to Maxwell's equations
using finite element methods. For efficient solving, the basis
functions~$\varphi_{i,j}(x,y)$
are defined in terms of orthogonal polynomials, which in this case are
Legendre and Jacobi polynomials:
\[
  \varphi_{i,j}(x,y) :=
  (1-x)^i P_j^{(2 i+1,0)}(2 x-1) P_i\big(\textstyle\frac{2 y}{1-x}-1\big).
\]
In the implementation one needs formulate how the partial derivatives
of $\varphi_{i,j}(x,y)$ can be represented in the basis itself, i.e., as
linear combinations of shifts of the $\varphi_{i,j}(x,y)$. In
Section~\ref{sec:Examples} we demonstrate how this was achieved
using algorithms for D-finite functions.

One of the big advantages of the symbolic computation methods presented
in this chapter is that they cannot only deal with named special functions,
but also with no-name D-finite functions that appear as solutions to
(potentially complicated) differential equations or recurrences.
One situation where such functions appear is Zeilberger's
holonomic ansatz~\cite{Zeilberger07} for symbolic determinant evaluations, which
will be elucidated in more detail in Section~\ref{sec:Examples}. For
the moment, let us just mention that it allows us to prove difficult
determinant evaluations such as
\begin{align*}
  & \det_{1\leq i, j \leq 2m+1}
    \left[ \binom{\mu+i+j+2r}{j+2r-2} - \delta_{i,j+2r} \right] \\
  &\qquad=\frac{(-1)^{m-r+1} \, (\mu+3) \, (m+r+1)_{m-r}}{
    2^{2m-2r+1} \, \bigl(\frac{\mu}{2}+r+\frac32\bigr)_{m-r+1}} \cdot
  \prod_{i=1}^{2m} \frac{(\mu+i+3)_{2r}}{(i)_{2r}} \\
  &\qquad\quad\times \prod_{i=1}^{m-r} \frac{\bigl(\mu+2i+6r+3\bigr)_i^2 \,
    \bigl(\frac{\mu}{2}+2i+3r+2\bigr)_{i-1}^2}{\bigl(i\bigr)_i^2 \,
    \bigl(\frac{\mu}{2}+i+3r+2\bigr)_{i-1}^2}.
\end{align*}
Maybe the greatest success of the holonomic ansatz was the proof
of the $q$-TSPP Conjecture~\cite{KKZ2011},
which was famously formulated by David P.~Robbins and George Andrews
in 1983~\cite{Stanley86}.
It states that the orbit-counting generating function for totally
symmetric plane partitions is given by
\[
  \sum_{\pi\in T(n)} q^{|\pi/{S_3}|} = 
  \prod_{1 \leq i \leq j \leq k \leq n}\frac{1-q^{i+j+k-1}}{1-q^{i+j+k-2}}.
\]
Soichi Okada~\cite{Okada89} reformulated this problem as a certain
determinant evaluation, which by Zeilberger's holonomic ansatz
was translated into several $q$-holonomic summation identities,
which finally could be proven by creative telescoping (based on
massive computer calculations).

Although this chapter is mainly about symbolic computation, we shall
mention that the symbolic methods can often be used to support numerical
computations. One approach in this spirit is the holonomic gradient
method~\cite{NakayamaEtAl11,Koyama14}
that can be used for evaluating or optimizing holonomic expressions.
For an input holonomic function $f(x_1,\dots,x_s)$ and a point
$(a_1,\dots,a_s)\in\set R^s$, it outputs an approximation of the
evaluation $f(a_1,\dots,a_s)$, using the following steps:
\begin{enumerate}
\item Determine a holonomic system (set of differential equations) to which
  $f$ is a solution, and let $r$ be its holonomic rank.
\item Fix a suitable ``basis'' of derivatives
  $\mathbf{f}=\bigl(f^{(m_1)},\dots,f^{(m_r)}\bigr)$ of~$f(x_1,\dots,x_s)$.
\item Convert the holonomic system into a set of Pfaffian systems, i.e.,
  $\frac{\mathrm{d}}{\mathrm{d}x_i}\mathbf{f}=A_i\mathbf{f}$ for each~$x_i$.
\item Compute $f^{(m_1)},\dots,f^{(m_r)}$ at a suitably chosen point
  $(b_1,\dots,b_s)\in\set R^s$, for which this is easy to achieve.
\item Use a numerical integration procedure (e.g., Euler,
  Runge-Kutta) to obtain $\mathbf{f}(a_1,\dots,a_s)$.
\end{enumerate}

We close this introductory section by listing diverse applications of
creative telescoping in different areas of mathematics and sciences,
without claiming that this list is complete in any sense.
\begin{itemize}

\item In particle physics, the evaluation of Feynman integrals, an
  easy instance being the integral
  \[
    \int_0^1\int_0^1\frac{w^{-1-\varepsilon /2}(1-z)^{\varepsilon /2}
      z^{-\varepsilon /2}}{(z+w-wz)^{1-\varepsilon}}
    \left(1-w^{n+1}-(1-w)^{n+1}\right)\,dw\,dz,
  \]
  is a central problem. Schneider and collaborators~\cite{schneider07,
    Schneider2016,schneider13b,schneider16} transform these
  integrals into complicated multi-sums that can be treated with
  symbolic summation software (such as the Sigma package~\cite{SchneiderSigma07}).

\item Fast converging series for the efficient computation of mathematical
  constants, such as
  \[
    \frac{16\pi^2}{3} = \sum_{k=0}^\infty \left(\frac{27}{64}\right)^{\!k}
    \frac{\bigl(k!\bigr)^{\!3} \, \bigl(\frac56\bigr)_{\!k} \, \bigl(\frac76\bigr)_{\!k}}{%
      \bigl(\frac32\bigr)_{\!k}^{\!5}} \, (74k^2+101k+35)
  \]
  can be found and/or proven by creative
  telescoping~\cite{Guillera06,Guillera08,Campbell23}.

\item Creative telescoping is used in algebraic statistics to be able to
  evaluate normalizing constants and other quantities of interest, by
  means of the holonomic gradient method. For example, this was carried out
  in the context of MIMO wireless communication systems to evaluate the
  signal-to-noise-ratio (SNR) probability density
  function~\cite{SiriteanuEtAl15,SiriteanuEtAl16},
  which faces accuracy problems when done with standard floating-point
  arithmetic.

\item Certain knot invariants in quantum knot theory, such as the colored
  Jones function, have been computed with the help of creative
  telescoping~\cite{GaroufalidisKoutschan13}. This invariant is a
  $q$-holonomic sequence of
  Laurent polynomials, whose $q$-recurrence can be computed by symbolic
  summation or by guessing. For example, the colored Jones function of
  a double twist knot~$K_{p,p'}$ is given by 
  \[
    J_{K_{p,p'},n}(q) = \sum_{k=0}^{n-1} (-1)^k c_{p,k}(q) c_{p',k}(q)
    q^{-k n -\frac{k(k+3)}{2}} \bigl(q^{n-1};q^{-1}\bigr)_k\,\bigl(q^{n+1};q\bigr)_k
  \]
  where the sequence $c_{p,n}(q)$ is defined by
  \[
    c_{p,n}(q) = \sum_{k=0}^n (-1)^{k + n}
    q^{-\frac{k}{2}+\frac{k^2}{2}+\frac{3 n}{2}+\frac{n^2}{2}+k p+k^2 p}
    \frac{\bigl(1 - q^{2k+1}\bigr) (q;q)_n}{(q;q)_{n-k}(q;q)_{n+k+1}}.
  \]

\item Hypergeometric expressions for generating functions of walks with
  small steps in the quarter plane have been found as solutions to
  differential equations, which had been computed before by creative
  telescoping~\cite{bostan16b}.
  
\item The uniqueness of the solution to Canham's problem which predicts the shape
  of biomembranes was established using creative telescoping, namely by showing
  that the reduced volume $\operatorname{Iso}(z)$ of any
  stereographic projection of the Clifford torus to~$\set R^3$ is
  bijective~\cite{BostanYurkevich22,BostanYuYurkevich24}.
  
\item Creative telescoping can be applied for computing efficiently the
  $n$-dimensional volume of a compact
  semi-algebraic set, i.e., the solution set of multivariate polynomial
  inequalities, up to a prescribed precision $2^{-p}$~\cite{Lairez19}.
  
\item An accurate, reliable and efficient method to compute a certified orbital
  collision probability between two spherical space objects involved in a
  short-term encounter under Gaussian-distributed uncertainty was developed
  in~\cite{Serra16}. The computational method is based on an analytic
  expression for the integral, derived by using the Laplace transform and
  D-finite functions.
  
\item The study of integrals and diagonals related to some topics in theoretical
  physics such as the Ising model or the lattice Green's function has been
  significantly promoted by creative telescoping~\cite{
    BostanBoukraaChristolHassaniMaillard13,
    KoutschanFCC,
    HassaniKoutschanMaillardZenine16,
    AbdelazizBoukraaKoutschanMaillard18b,
    AbdelazizKoutschanMaillard20,
    AbdelazizBoukraaKoutschanMaillard20}.
  
\item Some irrationality measures of mathematical constants such as~$\pi$,
  $\zeta$-values, or elliptic $L$-values were computed with the help of
  creative telescoping~\cite{ZeilbergerZudilin20,KoutschanZudilin22}.
  
\end{itemize}

\section{Rational Integration: Ostrogradsky--Hermite reduction}\label{sec:hermite}


We know from college calculus that any rational function $f(x)\in \bC(x)$ has an indefinite integral of the form
\[\int f(x) \, dx = g_0(x) + \sum_{i=1}^{n} c_i \log(g_i(x)),\] 
where $g_0, g_1, \ldots, g_n$ are rational functions and $c_1, \dots, c_n$ are constants. The rational function $g_0$ 
and the sum $\sum_{i=1}^{n} c_i \log(g_i(x))$ are called
the \emph{rational part} and the \emph{logarithmic part} of the integral $\int f(x) \, dx$, respectively. 
Ostrogradsky~\cite{Ostrogradsky1845} in 1845 and Hermite~\cite{Hermite1872} in 1872 gave a reduction method for finding the rational part 
by only using GCD calculations without any algebraic extensions. We will recall the Ostrogradsky--Hermite reduction 
for rational integration from~\cite{BronsteinBook} and then use it in Section~\ref{sec:ct} for designing a creative-telescoping algorithm for bivariate rational functions.

Let $C$ be a field of characteristic zero, such as $\bQ$, $\bR$ and $\bC$. Let $C[x]$ be the ring of polynomials 
in $x$ over $C$ and $C(x)$ be the field of rational functions in $x$ over $C$. The usual derivation $d/dx$ on $C(x)$
is denoted by $'$, which satisfies the three properties: (i) $x'=1$; (ii) $(f+g)' = f' + g'$ for all $f, g\in C(x)$;
(iii) $(fg)' = f' g + fg'$ for all $f, g\in C(x)$. Two main problems in rational integration are as follows:

\medskip 
\textbf{Integrability Problem.}~~Given a rational function $f\in C(x)$, decide whether there exists another rational function 
$g\in C(x)$ such that $f = g'$. If such a $g$ exists, we say that $f$ is \emph{integrable} in $C(x)$.

\medskip
\textbf{Decomposition Problem.}~~Given a rational function $f\in C(x)$, compute  
$g, r\in C(x)$ such that $f = g' + r$, where $r=a/b$ satisfies some \lq\lq minimal\rq\rq\ conditions: (i) $\gcd(a, b)=1$;
(ii) $\deg(a)<\deg(b)$; (iii) $b$ is squarefree, i.e., $\gcd(b, b')=1$.

\medskip
We first recall two types of polynomial factorizations and related partial fraction decompositions for rational functions.
For a polynomial $P\in C[x]$, a \textit{squarefree factorization} of $P\in C[x]$ is of the form
\[P=P_1P_2^2\cdots P_m^m,\]
where $m\in \bN$ and $P_i\in C[x]$ satisfy that $\gcd(P_i,P_j)=1$ for all $i, j$ with $1\leq i<j\leq m$ and 
each $P_i$ is~\textit{squarefree}.  An~\textit{irreducible factorization} of $P$ is of the form
\[P=P_1^{d_1}P_2^{d_2}\cdots P_m^{d_m},\]
where $d_i \in \bN$ and $P_i\in C[x]$ satisfy that $\gcd(P_i,P_j)=1$ for all $i, j$ with $1\leq i<j\leq m$ and each $P_i$ is irreducible.

Let $f=P/Q\in C(x)$ be such that $\gcd(P, Q) =1$ and $\deg(P)<\deg(Q)$. Given any factorization $Q = Q_1 Q_2 \cdots Q_m$ such that $\gcd(Q_i,Q_j)=1$ for all $i, j$ with $1\leq i<j\leq m$, we have 
a partial fraction decomposition for $f$
\[ f=\frac{P}{Q}=\sum_{i=1}^m\frac{P_i}{Q_i}.\]
Corresponding to the previous two types of polynomial factorizations,  the \textit{squarefree partial fraction decomposition} for $f$ is of the form
\[ f = \sum_{i=1}^m \frac{P_i}{Q_i^i},\]
where $Q=Q_1Q_2^2\cdots Q_m^m$ is a squarefree factorization of $Q$. The \textit{irreducible partial fraction decomposition} for $f$ is of the form 
\[f=\sum_{i=1}^m \frac{P_i}{Q_i^{d_i}},\]
where $Q=Q_1^{d_1}Q_2^{d_2}\cdots Q_m^{d_m}$ is an irreducible factorization of $Q$.
 
The following lemma provides a criterion for testing integrability in $C(x)$, which also 
shows that \textbf{Integrability Problem} can be reduced to \textbf{Decomposition Problem}. 

\begin{lem}
Let $f=a/b\in C(x)$ satisfy the three conditions: (i) $\gcd(a, b)=1$;
(ii) $\deg(a)<\deg(b)$; (iii) $b$ is squarefree, i.e., $\gcd(b, b')=1$.
Then $f$ is integrable in $C(x)$ if and only if $a=0$.
\end{lem}
\begin{proof}
We only need to show the necessity. Suppose that $f=g'$ for some $g\in C(x)$ and $a\neq0$. Since $\deg(a)<\deg(b)$, $g$ cannot be a polynomial in $C[x]$. Thus we can write $g=P/Q$ with $P, Q\in C[x]$ and $g$ has at least one pole $\beta\in \overline{C}$ with $Q(\beta)=0$ and $P(\beta)\neq 0$. Write $Q= (x-\beta)^m \tilde{Q}$ with $\tilde{Q}(\beta)\neq 0$ and $m\geq 1$. Then
\[f = g' = \frac{-mP\tilde{Q}+(x-\beta)(P'\tilde{Q}-P\tilde{Q}')}{(x-\beta)^{m+1}\tilde{Q}^2}=\frac{a}{b}.\]
Therefore, we have $b(-mP\tilde{Q}+(x-\beta)(P'\tilde{Q}-P\tilde{Q}'))=a(x-\beta)^{m+1}\tilde{Q}^2$. This implies that $(x-\beta)^{m+1}$ divides $b$, which contradicts the assumption that $b$ is squarefree.
\end{proof}

We now explain the main steps of the Ostrogradsky--Hermite reduction that solves the decomposition problem for 
rational functions in $C(x)$.

\begin{enumerate}
\item For a rational function $f\in C(x)$, we first write $f = p + a/b$, where $p,\,a,\,b\in C[x]$ with $\gcd(a,b)=1$ and $\deg(a)<\deg(b)$. 
Note that $p$ can be written as $p=q'$ for some $q\in C[x]$, so we obtain $f=q'+a/b$.
\item Let $b=b_1b_2^2\cdots b_m^m$ be a squarefree factorization of $b$. Then we get a squarefree partial fraction decomposition of $a/b$:
\[\frac{a}{b}=\sum_{i=1}^m\frac{a_i}{b_i^i},\]
where each $a_i \in C[x]$ satisfies that $\deg(a_i)<i\deg(b_i)$.
\item In the following, we use the trick of integration by part to decompose a simple fraction of the form $A/B^m$, where $A,\,B\in C[x]$, $m\geq 2$, $B$ is squarefree and $\deg(A)<m\deg(B)$. 
\begin{align*}
\frac{A}{B^m}&=\frac{UB+VB'}{B^m} =\frac{U}{B^{m-1}}+\frac{VB'}{B^m}\\
&=\frac{U}{B^{m-1}}+\left(\frac{(1-m)^{-1}V}{B^{m-1}}\right)'-\frac{(1-m)^{-1}V'}{B^{m-1}}\\
&=\left(\frac{(1-m)^{-1}V}{B^{m-1}}\right)'+\frac{U-(1-m)^{-1}V'}{B^{m-1}}.
\end{align*}
Repeating the above process, we get
\[\frac{A}{B^m}=\left(\frac{u}{B^{m-1}}\right)'+\frac{v}{B},\]
where $u,\,v\in C[x]$, $\deg(u)<(m-1)\deg(B)$ and $\deg(v)<\deg(B)$. Applying the above reduction to each fraction $a_i/b_i^i$
in the squarefree partial fraction decomposition of $a/b$ yields
\[\frac{a}{b}=\left(\frac{p}{b^-}\right)'+\frac{q}{b^*},\]
where $b^-=\gcd(b,b')$, $b^*=b/b^-$, $p,\,q\in C[x]$, $\deg(p)<\deg(b^-)$ and $\deg(q)<\deg(b^*)$.
\end{enumerate}
Thus any rational function $f\in C(x)$ can be decomposed by the Ostrogradsky--Hermite Reduction into the form
\begin{align}\label{eq-HR-form}
f=g'+\frac{a}{b},
\end{align}
where $a,\,b\in C[x]$ with $\gcd(a,b)=1$, $\deg(a)<\deg(b)$ and $b$ is squarefree.
Such a decomposition is often called an additive decomposition in $C(x)$ with respect to the derivation $'$.  

\begin{rem}
\begin{enumerate}
\item To compute the form~\eqref{eq-HR-form}, we only need to perform the ring operations in $C[x]$ without any computation in some algebraic extension 
of $C$. So the Ostrogradsky--Hermite reduction is called a ``rational'' algorithm for computing additive decompositions of rational functions.
\item We can also reduce the decomposition problem to a problem of solving a linear system, namely the Horowitz--Ostrogradsky approach. Specifically, for a proper fraction $a/b$, we make an ansatz $p=\sum_{i=0}^{d^--1}p_ix^i$ and $q=\sum_{j=0}^{d^*-1}q_jx^j$, where $d^-=\deg(b^-)$, $d^*=\deg(b^*)$,  and $p_i, q_j$ are the undetermined coefficients in~$C$. Then we get a
    linear system in the unknowns $p_i$ and $q_j$ by comparing the coefficients of the equality
\[\frac{a}{b}=\left(\frac{p}{b^-}\right)'+\frac{q}{b^*}.\]
The existence of solutions of this system is guaranteed by the Ostrogradsky--Hermite reduction.
\end{enumerate}
\end{rem}

\begin{ex} Consider the following rational function
\[f=\frac{x^7-24x^4-4x^2+8x-8}{x^8+6x^6+12x^4+8x^2} \in\bQ(x).\] 
We will decompose it into the form~\eqref{eq-HR-form} in two different ways. 
Following the process of the Ostrogradsky--Hermite reduction, we first compute the squarefree factorization of the denominator $b=x^8+6x^6+12x^4+8x^2$ 
which is $b=x^2(x^2+2)^3$. Then we get the squarefree partial fraction decomposition for $f$ of the form
\[ f=\frac{x-1}{x^2}+\frac{x^4-6x^3-18x^2-12x+8}{(x^2+2)^3}.\] 
Performing the reduction for each fraction yields 
\[f=\left(\frac{1}{x}+\frac{6x}{(x^2+2)^2}-\frac{x-3}{x^2+2}\right)'+\frac{1}{x}.\]

In the Horowitz-Ostrogradsky approach, we make the following ansatz 
\[ f=\left(\frac{\sum_{i=0}^{4}p_ix^i}{x(x^2+2)^2}\right)'+\frac{\sum_{j=0}^{2}q_jx^j}{x(x^2+2)},\] 
where $p_i$ and $q_j$ are undetermined coefficients in $\bQ$. Solving the corresponding linear system yields 
\[ (p_0,p_1,p_2,p_3,p_4,q_0,q_1,q_2)=(4,6,8,3,0,2,0,1).\]
This leads to an additive decomposition of $f$
\begin{align*}
f&=\left(\frac{3x^3+8x^2+6x+4}{x(x^2+2)^2}\right)'+\frac{x^2+2}{x(x^2+2)}\\
&=\left(\frac{3x^3+8x^2+6x+4}{x(x^2+2)^2}\right)'+\frac{1}{x}.
\end{align*}
\end{ex}

We now focus on computing the logarithmic part in rational integration. Assume that $f = a/b$, where $a,\,b\in C[x]$, $\gcd(a,b)=1$ and $b$ is squarefree. 
Let $\beta_1,\,\beta_2,\ldots,\beta_n$ be all the roots of $b$ in $\overline{C}$. Then we can decompose $f$ as 
\[\frac{a}{b}=\sum_{i=1}^n\frac{\alpha_i}{x-\beta_i},\]
which leads to 
\[\int\frac{a}{b} \, dx=\sum_{i=1}^n\alpha_i\log(x-\beta_i),\]
where $\alpha_i=\frac{a(\beta_i)}{b'(\beta_i)}$ by Lagrange's formula for residues. 
This formula motivates the definition of \emph{Rothstein-Trager resultants}, which is
\[ R(z)=\resultant_x(b,a-zb')\in C[z]\] 
for $f=a/b\in C(x)$.

\begin{ex}
Consider the rational function $f = 1/(x^3+x)$. Then we have
\[\frac{1}{x^3+x}=\frac{1}{x}+\frac{-\frac{1}{2}}{x-i}+\frac{-\frac{1}{2}}{x+i},\]
which implies that
\[\int\frac{1}{x^3+x} \, dx=\log(x)-\frac{1}{2}\log(x-i)-\frac{1}{2}\log(x+i)=\log(x)-\frac{1}{2}\log(x^2+1).\]
Note that the final expressions in the logarithmic part only involve numbers in $\bQ$.  
\end{ex}
The following theorem shows that it is sufficient to compute the roots of the Rothstein-Trager resultant if we want to
get the logarithmic part (for the proof, see~\cite{BronsteinBook}).

\begin{thm}
Let $f=a/b\in C(x)$ be such that $a,\,b\in C[x]$, $\gcd(a,b)=1$ and $b$ is squarefree. Let $R(z)=\resultant_x(b,a-zb')\in C[z]$. Then
\[\int f dx=\sum_{\alpha\in\bar{C},\,R(\alpha)=0}\alpha\log(g_\alpha),\]
where $g_\alpha=\gcd(b,a-\alpha b')\in C(\alpha)[x]$.
\end{thm}

\begin{ex}
\begin{enumerate}
\item Let $f=1/(x^3+x)\in\bQ(x)$. Then 
\begin{align*}
R(z)&=\resultant_x(x^3+x,1-z(3x^2+1))\\
    &=-4z^3+3z+1\\
    &=-4(z-1)\left(z+\frac{1}{2}\right)^2.
\end{align*}
So we have 
\begin{align*}
    \int f dx=1\cdot \log(g_1)+\left(-\frac{1}{2}\right)\log(g_{-\frac{1}{2}}),
\end{align*}
where $$g_1=\gcd(x^3+x,1-(3x^2+1))=x$$ and $$g_{-\frac{1}{2}}=\gcd\left(x^3+x,1+\frac{1}{2}\left(3x^2+1\right)\right)=x^2+1.$$ So we finally obtain
\[\int f dx=1\cdot \log(x)+\left(-\frac{1}{2}\right)\log(x^2+1).\]
\item Let $f=1/(x^2-2)\in\bQ(x)$. Then $R(z)=-8z^2+1$ with the roots $\sqrt{2}/4$ and $-\sqrt{2}/4$. So we have 
\[\int f dx=\frac{\sqrt{2}}{4}\log\left(x-\sqrt{2}\right)-\frac{\sqrt{2}}{4}\log\left(x+\sqrt{2}\right).\]
\end{enumerate}
\end{ex}

\section{Implementation of Rational Integration}

The classical algorithms for rational integration as introduced in the previous section can be easily implemented
in a computer algebra system. To give an idea what such a code might look like, let us produce a simple implementation
in SageMath.

Polynomials in SageMath know what their derivative is, but rational functions do not. So to get started, we write
a function for differentiating rational functions.  

\begin{verbatim}
def der(rat):
    p = rat.numerator()
    q = rat.denominator()
    return (p.derivative()*q - p*q.derivative())/q^2
\end{verbatim}

This was easy. In the next step, we write a function for integrating polynomials. This is also easy:

\begin{verbatim}
def intpoly(p):
    x = p.parent().gen()
    q = 0
    for i in range(p.degree() + 1):
        q += p[i]/(i+1)*x^(i+1)
    return q
\end{verbatim}

For the implementation of the Ostrogradsky--Hermite reduction process, it is useful to prepare two auxiliary functions.
The first takes a polynomial $p$ as input and returns a triple $(u,v,m)$ such that $p=uv^m$, $v$ is
squarefree, and $u$ and $v$ are coprime. A somewhat brutal way to implement this is as follows.

\begin{verbatim}
def split(p):
    if p.degree() == 0:
       return (p, 1, 1)
    q = p.factor()
    m = max([e for (_, e) in q])
    u = prod([u^e for (u, e) in q if e < m])
    v = prod([u for (u, e) in q if e == m])
    return (u, v, m)
\end{verbatim}

It is brutal in the sense that we rely on polynomial factorization. This is not necessary. In order to
separate the factors with the highest multiplicity, it suffices to apply squarefree decomposition.
Here is an implementation of this alternative.

\begin{verbatim}
def split(p):
    if p.degree() == 0:
       return (p, 1, 1)
    q = p.squarefree_decomposition()
    return (p//q[-1][0]^q[-1][1], q[-1][0], q[-1][1])
\end{verbatim}

The second auxiliary function is for solving modular equations. It takes three polynomials $a,u,v$ as
input and computes a polynomial $b$ such that $\deg(b)<\deg(v)$ and $a=bu\bmod v$. It is assumed 
that $u$ and $v$ are coprime.

The function is based on the extended euclidean algorithm. Applied to $u$ and~$v$, this algorithm
computes $g,b,c$ such that $g=bu+cv$ and $g=\gcd(u,v)$. Since $u$ and $v$ are assumed to be coprime,
we will have $g=1$, but more generally, the approach works whenever $a$ is a multiple of~$g$.
Multiplying the equation $g=bu+cv$ with $a/g$ (which is then a polynomial) gives $a=\frac{ba}{g}u+\frac{ca}{g}v$.
Finally, for every $q$ we have
\[
a=\Bigl(\frac{ba}{g}-qv\Bigr)u+\Bigl(\frac{ca}{g}+qu\Bigr)v,
\]
so we are free to replace
the polynomial $\frac{ba}{g}$ by $\frac{ba}{g}\bmod v$ to ensure the required degree condition.

\begin{verbatim}
def solvemod(a, u, v):
    g, b, c = u.xgcd(v) 
    b, c = b*a//g, c*a//g
    b = b % v
    return b
\end{verbatim}

We can now implement Ostrogradsky--Hermite reduction as follows.

\begin{verbatim}
def hermite(rat):
    p = rat.numerator()
    q = rat.denominator()
    g, p = p.quo_rem(q)
    if g != 0:
        g0, h = hermite(p/q)
        return (intpoly(g) + g0, h)
    u, v, m = split(q) 
    if m == 1:
        return (0, rat)
    else:
        b = solvemod(p, -(m-1)*u*v.derivative(), v)
        g, h = hermite(rat - der(b/v^(m-1)))
        return (g + b/v^(m-1), h)
\end{verbatim}

The first part handles the case when there is a nontrivial polynomial part. Getting this case out of
the way, we can assume in the second part that the input is a proper rational function. We then proceed
as described in the previous section.

The code presented here is not optimized for efficiency but for readability. If efficiency matters,
we may want to avoid the recursive calls in favor of an explicit loop. Also the recomputation of the
squarefree decomposition in each iteration is not optimal and can be easily avoided if we are willing
to write a few more lines of code. The result might look as follows.

\begin{verbatim}
def hermite(rat):
    p = rat.numerator()
    u = rat.denominator()
    if (u.degree() == 0):
        return (intpoly(p), 0)
    g, p = p.quo_rem(u)
    g = intpoly(g)
    dec = u.squarefree_decomposition()
    m = max(e for _, e in dec)
    factors = [1]*(m + 1)
    for f, e in dec:
        factors[e] = f
    while m > 1:
        v = factors[m]
        u = u//v^m
        b = solvemod(p, -(m-1)*u*v.derivative(), v)
        p = (p + (m-1)*b*u*v.derivative()\
                 - b.derivative()*u*v)//v
        g += b/v^(m-1)
        factors[m - 1] *= v
        u *= v^(m - 1)
        m -= 1
    return (g, p/factors[1])
\end{verbatim}

Let us now turn to the computation of the logarithmic part. In the following code, it
is assumed that the input rational function is proper and that its denominator is squarefree. Some
cryptic instructions are needed to introduce a new variable $c$ into the constant ring. We can then
form the resultant of $q$ and $p-zq'$, which amounts to eliminating the variable $x$ from the ideal
generated by $q$ and $p-zq'$. The result is a polynomial in $z$ only. We iterate over the irreducible
factors of this polynomial. Each factor yields a 
contribution of the form
\[
\sum_{\alpha:u(\alpha)=0}\alpha\log g(\alpha,x)
\]
to the closed form of the
integral. Such a contribution is encoded as a pair $(u,g)$ of polynomials, where $g$ is an element
of $(\set Q[z]/\<u>)[x]$. The output of the function is a list of such pairs.

\begin{verbatim}
def logpart(rat):
    p = rat.numerator()
    q = rat.denominator()
    C = p.parent().base_ring()
    x, z = C[p.parent().gen(), 'z'].gens()
    I = C[x,z].ideal([q, p - z*q.derivative()])
    out = []
    for u, _ in C[z](I.elimination_ideal(x).gen(0)).factor():
        if u not in ZZ:
            out.append((u, C['z'].quotient(u)[x](q)\
                           .gcd(p - z*q.derivative())))
    return out
\end{verbatim}

If we put all the code above into a file \verb|ratint.sage|, then we can use it in SageMath as follows.

\begin{verbatim}
sage: load("ratint.sage")
sage: x = QQ['x'].gen()
sage: hermite((x+1)^4*(x+2)^3/(x+4)^2/(x+5)^3)
((1/3*x^6 - 11/6*x^5 + 182/3*x^4 + 8585/6*x^3 + 8448*x^2
+ 22936*x + 30024)/(x^3 + 14*x^2 + 65*x + 100),
 (-1116*x - 684)/(x^2 + 9*x + 20))
sage: (x+1)^4*(x+2)^3/(x+4)^2/(x+5)^3 - der(_[0])
(-1116*x - 684)/(x^2 + 9*x + 20)
sage: logpart(_)
[(z - 3780, x + 4), (z + 4896, x + 5)]
\end{verbatim}

This means we have
\begin{alignat*}1
  &\int \frac{(x+1)^4(x+2)^3}{(x+4)^2(x+5)^3} \, dx\\
  &= \frac{\frac13x^6 - \frac{11}6x^5 + \frac{182}3x^4 + \frac{8585}6x^3 + 8448x^2 + 22936x + 30024}{x^3 + 14x^2 + 65x + 100}\\
  &\quad  + 3780\log(x+4) - 4896\log(x+5).
\end{alignat*}

\section{Classical summation methods}\label{sec:classical}


For a sequence $f(n)$, there are two types of summation problems: the indefinite summation problem 
is to find another sequence $g(n)$ such that 
\[f(n) = g(n+1) - g(n) \triangleq \Delta(g(n))\] and the definite summation problem
is to find a closed-form formula for the definite sum $\sum_{n=a}^{b} f(n)$. Those two problems
are closely connected by the discrete Newton-Leibniz formula
\[\sum_{n=a}^{b} f(n) = g(b+1) - g(a).\]
So one can get a nice formula for the definite sum if $g$ had a nice form. 
For example, let us take $f(n) = 1/(n(n+1))$. We can write $f(n)$ as
\[f(n) = \frac{1}{n} - \frac{1}{n+1} = g(n+1) - g(n) \quad \text{where $g(n) = -\frac{1}{n}$}. \]
Form the above telescoping relation, we get
\[\sum_{n=1}^{2025} f(n) = \frac{2025}{2026}.\]
In the following, we will recall some classical techniques for finding the indefinite sums
of polynomials, rational functions and binomial coefficients.  

%
%

We first study the summation problem on polynomials. 
Recall that the \emph{falling factorial} $x^{\underline{m}}$ of $x$ is defined by
\[ x^{\underline{m}}=x(x-1)\cdots(x-m+1) \quad \text{for all $m\in \bN$}. \]
As a discrete analogue of the derivative formula $(x^m)'=m x^{m-1}$, we have
 \[\Delta(x^{\underline{m}})=m x^{\underline{m-1}}.\]
Note that $\deg(x^{\underline{m}}) = m$, then
$ \left\{1,  x^{\underline{1}},  x^{\underline{2}},  \dots,  x^{\underline{m}},  \dots\right\}$ also forms a basis of $C[x]$, as a vector space over $C$.
Then 
\[ x^m=\sum_{i=0}^m {m\brace i} x^{\underline{i}},\]
where ${m \brace i}$ is the Stirling numbers of the second kind, which counts the number of partitions of the set $\{1,2,\dots,m\}$ into $i$ nonempty subsets.
Note that ${n \brace k}=0$ if $i>m$, ${m \brace m}=1$, ${m\brace 0}=0$ for $m\geq 0$, and the Stirling number satisfies the recurrence relation
\[{m\brace i}={m-1\brace i-1}+i {m-1\brace i} \]
for all $m, i \in \bN$ with $m>i>0$. This recurrence relation can be derived as follows.
\begin{align*}
    x^m &=x\cdot x^{m-1}=x\sum_{0\leq i<m}{m-1\brace i}x^{\underline{i}}=\sum_{0\leq i <m}{m-1\brace i }x\cdot x^{\underline{i}}\\
    &=\sum_{0\leq i<m}{m-1\brace i}(x-i) x^{\underline{i}}+\sum_{0\leq i<m}{m-1\brace i}ix^{\underline{i}}\\
    &=\sum_{0\leq i<m}{m-1\brace i}x^{\underline{i+1}}+\sum_{0\leq i<m}{m-1\brace i}ix^{\underline{i}}\\
    &=x^{\underline{m}}+\sum_{1\leq i<m}\left({m-1\brace i-1}+i{m-1\brace i}\right)x^{\underline{i}}.
\end{align*}

The following theorem gives an explicit solution to the indefinite summation problem on polynomials in $C[x]$.
\begin{thm}
Let $f=\sum_{i=0}^d a_ix^i\in C[x]$. Then $f= g(x+1)-g(x)$ for some $g\in C[x]$ of the form
\[g=\sum_{0\leq j\leq i\leq d}a_i{i\brace j}\frac{x^{\underline{j+1}}}{j+1}.\]
\end{thm}
\begin{proof}
Since $\Delta(x^{\underline{m}})=m x^{\underline{m-1}}$, we have $x^{\underline{m}}=\Delta\left(\frac{1}{m+1}x^{\underline{m+1}}\right)$ for $m\geq 0$. Hence
\begin{align*}
f &= \sum_{i=0}^d a_i x^i = \sum_{i=0}^d a_i \left( \sum_{j=0}^i \left\{ i \atop j \right\} x^{\underline{j}} \right) = \sum_{i=0}^d \sum_{j=0}^i a_i \left\{ i \atop j \right\} \Delta \left( \frac{1}{j+1} x^{\underline{j+1}} \right) \\
&= \Delta \left( \sum_{i=0}^d \sum_{j=0}^i a_i \left\{ i \atop j \right\} \frac{1}{j+1} x^{\underline{j+1}} \right).
\end{align*}
This completes the proof.
\end{proof}

\begin{ex} By the previous formula, we get
\[\sum_{k=0}^{n-1}k^2=\sum_{k=0}^{n-1}\Delta\left(\frac{k^{\underline{3}}}{3}+\frac{k^{\underline{2}}}{2}\right) \quad \text{and} \quad 
\sum_{k=0}^{n-1}k^3=\sum_{k=0}^{n-1}\Delta\left(\frac{k^{\underline{2}}}{2}+k^{\underline{3}}+\frac{k^{\underline{4}}}{4} \right).
\]
The telescoping relation $f(k)=g(k+1)-g(k)$ leads to the discrete Newton-Leibniz formula 
\[ \sum_{k=0}^{n-1}f(k)=g(n)-g(0).\] 
Therefore, the definite sums can be evaluated as
\[\sum_{k=0}^{n-1}k^2=\frac{n^{\underline{3}}}{3}+\frac{n^{\underline{2}}}{2}=\frac{n(n-1)(2n-1)}{6}\quad \text{and} \quad 
\sum_{k=0}^{n-1}k^3=\frac{n^{\underline{2}}}{2}+n^{\underline{3}}+\frac{n^{\underline{4}}}{4} = \frac{n^2(n-1)^2}{4}.
\]
\end{ex}

We next recall a method in~\cite{Abramov1971} for solving the indefinite summation problem on rational functions.  For a polynomial $P\in C[x]\setminus C$, the \emph{dispersion} of $P$, denoted by $\text{disp}(P)$, is defined as
\[ \max\left\{i\in \mathbb{Z} \mid \gcd(P(x),P(x+i)\neq 1\right\}. \] 
For example, if $P=x(x+3)(x-\sqrt{2})(x+\sqrt{2})$, then $\text{disp}_{\sigma}(P)=3$.
A polynomial $P$ is called \emph{shift-free} if $\text{disp}(P)=0$.

\begin{lem}
Let $f=\frac{a}{b}\in C(x)\setminus \{ 0\}$ with $b\notin C$ and $\text{gcd}(a,b)=1$. Write $\Delta(f)=\frac{P}{Q}$ with $\gcd(P,Q)=1$. Then $ \text{disp}(Q)=\text{disp}(b)+1$.
\end{lem}
\begin{proof}
Let $d=\text{disp}(b)$. Then there is a root $\beta$ of $b$ such that $b(\beta)=b(\beta+d)=0$. We have 
\[\Delta(f)=\frac{a(x+1)}{b(x+1)}-\frac{a(x)}{b(x)}=\frac{a(x+1)b(x)-a(x)b(x+1)}{b(x)b(x+1)}.\] 
Notice that $\beta-1$ and $\beta+d$ are not the roots of the numerator $a(x+1)b(x)-a(x)b(x+1)$ but $b(x)b(x+1)$ vanishes at these two points.
So we have  $\text{disp}(Q)\geq d+1$. But $\text{disp}(b(x)b(x+1))\leq d+1$ by definition. Therefore, we get $\text{disp}(Q)=d+1$.
\end{proof}
Since the dispersion of a shift-free polynomial is zero, the above lemma leads to the following summability criterion.
\begin{cor}
Let $f=a/b$ be such that $a,b\in C[x]$ with $\deg(a)<\deg(b)$, $\gcd(a,b)=1$ and $b$ shift-free. Then $f=\Delta(g)$ for some $g\in C(x)$ if and only if $a=0$.
\end{cor}

Similar to rational integration, we now show that any rational function $f\in C(x)$ can be decomposed as
\[f=\Delta(g)+r,
\]
where $g\in C(x)$ and $r=\frac{a}{b}$ satisfies that $\deg(a)<\deg(b)$, $\gcd(a,b)=1$ and $b$ shift-free.
We can achieve this goal by using Abramov's reduction that can be viewed as a discrete analogue of the Ostrogradsky--Hermite reduction. 
Let $\sigma$ denote the shift operator defined by $\sigma(f(x)) = f(x+1)$ for any $f\in C(x)$.
The main idea of Abramov's reduction is based on the following reduction formula
\begin{align*}
\frac{a}{\sigma^m(b)}&=\frac{a}{\sigma^m(b)}-\frac{\sigma^{-1}(a)}{\sigma^{m-1}(b)}+\frac{\sigma^{-1}(a)}{\sigma^{m-1}(b)}\\
&=\Delta\left(\frac{\sigma^{-1}(a)}{\sigma^{m-1}(b)}\right)+\frac{\sigma^{-1}(a)}{\sigma^{m-1}(b)}\\
&=\Delta(g)+\frac{\sigma^{-m}(a)}{b},\quad  \text{where $g = \frac{\sigma^{-1}(a)}{\sigma^{m-1}(b)}+\frac{\sigma^{-2}(a)}{\sigma^{m-2}(b)}+\dots+\frac{\sigma^{-m}(a)}{b}$.}
\end{align*}
For any rational funciton $f\in C(x)$, we can decompose it as 
\[
f=p+\sum_{i=1}^n\sum_{j=1}^{m_i}\sum_{\ell=0}^{t_{i,j}} \frac{a_{i,j, \ell}}{\sigma^{\ell}(b_{i}^j)}
\]
where $p, a_{i, j, \ell}, b_{i} \in C[x]$  satisfy that $\deg(a_{i, j, \ell})< \deg(b_i)$ and $b_i$ are irreducible polynomials such that $\sigma^{k}(b_i)\neq b_j$ for any $i\neq j$ and $k\in \bZ$.  
Now applying the above reduction formula to each simple fraction $a_{i, j, \ell}/\sigma^{\ell}(b_{i}^j)$ and the summation formula to polynomials yields
\[f = \Delta(g) + \sum_{i=1}^{n}\sum_{j=1}^{m_i} \frac{\tilde{a}_{i, j}}{b_i^j} = \Delta(g) + \frac{a}{b},\]
where $\deg(a)<\deg(b)$, $\gcd(a,b)=1$ and $b$ is shift-free.
We can get the above additive decomposition by using the greatest factorial factorization introduced by Paule in~\cite{Paule1995b} which avoids
the irreducible polynomial factorization.


We now briefly recall the classical techniques for hypergeometric summation which can be used to prove combinatorial identities involving binomial coefficients.  
An algorithmic way for hypergeometric summation is based on Gosper's algorithm which will be explained in Section~\ref{sec:Gosper}.

For nonnegative integers $n$ and $k$, the binomial coefficient is defined as 
\[\binom{n}{k}=\frac{n!}{k!(n-k)!},
\]
which can be interpreted as the number of ways to choose $k$ apples from a set of $n$ apples. Some basic identities on binomial coefficients are listed as follows.
\begin{gather}
\binom{n}{k}=\binom{n}{n-k}.\\
\binom{n}{k}=\frac{n}{k}\binom{n-1}{k-1}, \; k\neq 0.\\
\binom{r}{m}\binom{m}{k}=\binom{r}{k}\binom{r-k}{m-k}.\\
\binom{n}{k}=\binom{n-1}{k}+\binom{n-1}{k-1}.\\
\sum_{k=0}^n\binom{k}{m}=\binom{0}{m}+\binom{1}{m}+\dots+\binom{n}{m}=\binom{n+1}{m+1}.\\
\sum_{k=0}^n \binom{n}{k} = 2^n.\\
\sum_{k=0}^n(-1)^k\binom{n}{k}=0.\\ \label{EQ:bi9}
\sum_{k=0}^n\binom{r+k}{k}=\binom{r+n+1}{n}.\\ 
\sum_k\binom{r}{m+k}\binom{s}{n-k}=\binom{r+s}{m+n}.
\end{gather}
The last identity is called Chu-Vandermonde's identity. A combinatorial proof of this identity is as follows. 
Firstly, by changing the variables $k\rightarrow k-m\; , n\rightarrow n-m$, the original identity becomes
\[
\sum_{k}\binom{r}{k}\binom{s}{n-k} = \binom{r+s}{n},
\]
The right-hand side represents the number of ways to choose $n$ people from
$r$ men and $s$ women. On the left-hand side, each term in the sum represents the number of ways to choose $k$ men from $r$ men 
and $n-k$ women from $s$ women. Summing these terms for $k$ from $0$ to $r$ coincides with the right-hand side. Many combinatorial identities 
can be proved by using the above list of basic identities. Let us consider the following identity
\[\sum_{k=0}^m \frac{\binom{m}{k}}{\binom{n}{k}}=\frac{n+1}{n+1-m}.
\]
By a direct calculation, we have the identity
\[\frac{\binom{m}{k}}{\binom{n}{k}} = \frac{\binom{n-k}{m-k}}{\binom{n}{m}}.\]
Then we have
\[
\sum_{k=0}^m \frac{\binom{m}{k}}{\binom{n}{k}} = \sum_{k=0}^m \frac{\binom{n-k}{m-k}}{\binom{n}{m}} = \frac{\sum_{k=0}^{m} \binom{n-k}{m-k}}{\binom{n}{m}}= \frac{\sum_{k=0}^{m} \binom{n-m+k}{k}}{\binom{n}{m}} = \frac{\binom{n+1}{m}}{\binom{n}{m}}= \frac{n+1}{n+1-m},
\]
where the identity~\eqref{EQ:bi9} is used. More difficult combinatorial identities will be proved by Zeilberger's method of creative telescoping in Section~\ref{sec:ct}, for instance
\[\sum_{k=0}^n\binom{2n-2k}{n-k}\binom{2k}{k}=4^n,
\]
and we can also show that the sum 
\[T_n := \sum_{k=0}^n\binom{2n-2k}{n-k}^2\binom{2k}{k}^2\] satisfies the recurrence relation
\[n^3  T_n=16(n-1/2)(2 n^2-2 n+1)T_{n-1}-256(n-1)^3 T_{n-2},
\]
which is crucial in the proof of the irrationality of $\zeta(3)$.

\section{Sister Celine's Method}\label{sec:Celine}

Before the availability of symbolic summation algorithms, simplifying binomial sums was so difficult
that people found it hard to imagine that computers could offer any reasonable support on the matter. 
For example, the first edition of The Art of Computer Programming~\cite{knuth68} includes the following exercise:
\begin{quote}
 \textbf{63.} [46] Develop computer programs for simplifying sums that involve binomial coefficients.
\end{quote}
Here 63 is the number of the exercise and 46 is an estimate of the hardness of the task. It is ranked
as a difficult research problem. Meanwhile, we can consider this research problem as being solved, and
consequently, the exercise has been removed from later editions of the book.

The book A=B by Petkov\v sek, Wilf, and Zeilberger~\cite{PWZbook1996} can be read as a solution to the exercise.
It is still a good starting point for getting into the theory. Chapter~4 of A=B is about a summation
algorithm known as Sister Celine's method, named after Sister Mary Celine Fasenmyer, who invented the
method in the 1940s~\cite{Fasenmyer1947,Fasenmyer1949}. The basic strategy of her approach for proving binomial summation
identities is as follows.
\begin{enumerate}
\item Given a sum $S(n):=\sum_k f(n,k)$ construct a linear recurrence with polynomial coefficients
for it, like
\[
    p_0(n)S(n) + p_1(n)S(n+1) + \cdots + p_r(n)S(n+r) = 0    
\]
\item Check whether the conjectured closed form satisfies the recurrence.
\item Check whether the conjectured identity is true for the first few values of~$n$.
\item Conclude that the identity is true for all~$n$.  
\end{enumerate}
The crucial step is the first, but before we discuss this step more closely, let us see the
other steps in an example.

\begin{ex}
\begin{enumerate}
\item Suppose we want to prove $\sum_k\binom nk=2^n$ ($n\geq0$). In step~1, we somehow obtain
  the recurrence $2S(n)-S(n+1)=0$. In step~2, we check that $2\cdot 2^n-2^{n+1}=0$ is indeed true
  for all~$n$. In step~3, we check that $S(0)=1=2^0$, $S(1)=1+1=2=2^1$ and $S(2)=1+2+1=4=2^2$.
  As the initial values match, the proof is complete. (Actually, it would have been sufficient
  to check a single initial value.)
\item Now suppose we want to prove $\sum_k(-1)^k\binom{2n}{k+n}^2=\frac{(2n)!}{n!^2}$ ($n\geq0$).
  Here we obtain the recurrence
  \begin{alignat*}1
    &16 (n+1) (2 n+1) (4 n+7) S(n)\\
    &\quad-2 (4 n+5) (8 n^2+20 n+11) S(n+1)\\
    &\qquad+(n+2) (2 n+3) (4 n+3) S(n+2)=0
  \end{alignat*}
  in step~1. In order to check that the conjectured closed form $\frac{(2n)!}{n!^2}$ also
  satisfies this recurrence, plug this term into the equation and then divide the whole
  equation by $\frac{(2n)!}{n!^2}$. Taking into account that 
  \[
   \frac{\frac{(2n)!}{n!^2}}{\frac{(2n)!}{n!^2}}=1,\quad
   \frac{\frac{(2n+2)!}{(n+1)!^2}}{\frac{(2n)!}{n!^2}}=\frac{2(2n+1)}{n+1},\quad
   \frac{\frac{(2n+4)!}{(n+2)!^2}}{\frac{(2n)!}{n!^2}}=\frac{4(2n+1)(2n+3)}{(n+2)(n+1)}
  \]
  are rational functions in~$n$, checking the equation amounts to simplifying a rational
  expression to zero. This is easy. Checking the initial values in step~3 is also easy,
  so the proof is again complete. 
\end{enumerate}
\end{ex}

We restrict the attention to identities $\sum_k f(n,k)=F(n)$ whose right-hand sides $F(n)$
are such that we can proceed like in the example above.

\begin{defi}\label{def:hg1}
 A function $F(n)$ is called a \emph{hypergeometric term} if there is a rational function
 $u(n)$ such that $F(n+1)/F(n)=u(n)$ for (almost) all~$n$.
\end{defi}

Typical examples are polynomials, rational functions, exponentials (e.g., $(-1)^n$ or $555^n$),
the factorial $n!$, and products and quotients of hypergeometric terms. If $F(n)$ is a
hypergeometric term, then so is $F(an+b)$ for every fixed $a,b\in\set N$.

A rational function $u$ with no roots or poles in the integers together with an initial value
$F(0)$ uniquely determines the hypergeometric term $F(n)=\prod_{k=0}^{n-1}u(k)$. 

The definition extends to bivariate functions as follows. We will consider
sums whose summand is a hypergeometric term.

\begin{defi}\label{def:hg2}
 A function $f(n,k)$ is called a \emph{hypergeometric term} if there are rational functions
 $u(n,k)$ and $v(n,k)$ such that
 \[
   \frac{f(n+1,k)}{f(n,k)}=u(n,k)
   \quad\text{and}\quad
   \frac{f(n,k+1)}{f(n,k)}=v(n,k)
   \quad\text{for (almost) all $n,k$.}
 \]
\end{defi}

Examples include again polynomials, rational functions, exponentials, the binomial coefficient
$\binom nk$, and products and quotients of these.

\begin{ex}
  Let us find out which of the following four expressions are bivariate
  hypergeometric terms:
  \[
    \text{(i) }\ 3^{nk+1},\quad 
    \text{(ii) }\ \frac{\Gamma(n+3k-\pi)}{\Gamma(2n-k+\frac12)},\quad 
    \text{(iii) }\ n^k,\quad 
    \text{(iv) }\ \binom{dn}{k}. 
  \]
  It is clear that (i) is not hypergeometric, because either shift quotient
  is not a rational function but involves $3^n$ resp.~$3^k$. Item (iii)
  is hypergeometric w.r.t.~$k$ but not w.r.t.~$n$. The binomial coefficient
  in (iv) is not hypergeometric for symbolic~$d$, but it is bivariate
  hypergeometric when $d$ is chosen to be a concrete integer. The only
  example that is bivariate hypergeometric without any restrictions is~(ii).
\end{ex}

If $u$ and $v$ have no roots or poles in $\set Z^2$, then $f$ is uniquely determined by $u$ and $v$
and the value $f(0,0)$. Typically however, $u$ and $v$ do have roots or poles. In this case, manual
inspection may be required to check the results of a ``formal'' computation.

Not every pair $(u,v)$ of rational functions gives rise to a consistent definition of a bivariate
hypergeometric term. A necessary condition is that $u(n,k+1)v(n,k)=u(n,k)v(n+1,k)$. This is because
we may rewrite $f(n+1,k+1)$ to $u(n,k+1)f(n,k+1)$ and then further to $u(n,k+1)v(n,k)f(n,k)$, or alternatively
to $v(n+1,k)f(n+1,k)$ and then further to $v(n+1,k)u(n,k)f(n,k)$, and the results must match.

\medskip
Given a hypergeometric term $f(n,k)$, our task is to find a recurrence for the sum $\sum_k f(n,k)$.
The idea is to find a recurrence for the summand $f(n,k)$ which can be translated into a recurrence
for the sum. For simplicity, let us assume that $f(n,k)$ is such that for every $n$ there are only
finitely many integers $k$ such that $f(n,k)$ is nonzero (``natural boundaries''). Then we do not
need to worry too much about the boundaries of sums. 

\begin{ex}
\begin{enumerate}
\item For $f(n,k)=\binom nk$ we have the Pascal triangle recurrence
 \[
   f(n+1,k+1)-f(n,k)-f(n,k+1)=0.
 \]
 Summing this equation over all $k$ yields
 \[
   \underbrace{\sum_k f(n+1,k+1)}_{=S(n+1)} - \underbrace{\sum_k f(n,k)}_{=S(n)} - \underbrace{\sum_k f(n,k+1)}_{=S(n)} = 0,
 \]
 which simplifies to $S(n+1)-2S(n)=0$.
\item For $f(n,k)=(-1)^k\binom{2n}{n+k}^2$, we have the following recurrence:
  \begin{alignat*}1
    &(n+1) (2 n+1) (4 n+7) f(n,k)\\
    &+4 (n+1) (2 n+1) (4 n+7) f(n,k{+}1)
    +(4 n+5) (4 n^2+10 n+5) f(n{+}1,k{+}1)\\
    &+6 (n+1) (2 n+1) (4 n+7) f(n,k{+}2)
    -4 (4 n+5) (6 n^2+15 n+8) f(n{+}1,k{+}2)\\
    &+4 (n+1) (2 n+1) (4 n+7) f(n,k{+}3)
    +(4 n+5) (4 n^2+10 n+5) f(n{+}1,k{+}3)\\
    &+(n+1) (2 n+1) (4 n+7) f(n,k{+}4)\\
    &+(n+2) (2 n+3) (4 n+3) f(n{+}2,k{+}2)=0.
 \end{alignat*}
 This is more complicated than Pascal's recurrence, but note that we can nevertheless proceed as before:
 by summing the equation over all~$k$ we obtain (after a little bit of simplification), the recurrence 
  \begin{alignat*}1
    &16 (n+1) (2 n+1) (4 n+7) S(n)\\
    &\quad-2 (4 n+5) (8 n^2+20 n+11) S(n+1)\\
    &\qquad+(n+2) (2 n+3) (4 n+3) S(n+2)=0
  \end{alignat*}
  for the sum $S(n)=\sum_k f(n,k)$.
\end{enumerate}
\end{ex}

The summand $f(n,k)$ is specified in terms of the rational functions $u$ and~$v$, but the recurrence equations
\[
  f(n+1,k)=u(n,k)f(n,k)
  \quad\text{and}\quad
  f(n,k+1)=v(n,k)f(n,k)
\]
in general do not admit a direct translation to a recurrence for the sum, because $u$ and $v$ involve the
summation variable~$k$. A crucial feature of the recurrences in the example above was that their coefficients
are polynomials in $n$ only. This allowed us to rewrite, for example, $\sum_k (n+1)(2n+1)(4n+7)f(n,k)$ to
$(n+1)(2n+1)(4n+7)\sum_k f(n,k)$. Our goal is to construct such a ``\emph{$k$-free recurrence}'' for a given
hypergeometric term~$f(n,k)$. This can be done with linear algebra.

Suppose we seek a recurrence of the following form:
\begin{alignat*}1
&  a_{0,0}(n)f(n,k) + a_{1,0}(n)f(n+1,k)\\
&+ a_{0,1}(n)f(n,k+1) + a_{1,1}(n)f(n+1,k+1)\\
&+ a_{0,2}(n)f(n,k+2) + a_{1,2}(n)f(n+1,k+2) = 0.
\end{alignat*}
Replace $f(n+1,k+2)$ by $v(n+1,k+1)f(n+1,k+1)$, then $f(n+1,k+1)$ by $v(n+1,k)f(n+1,k)$, then $f(n+1,k)$
by $u(n,k)f(n,k)$, and so on, to obtain
\begin{alignat*}1
& \bigl( a_{0,0}(n) + a_{1,0}(n)u(n,k) + a_{0,1}(n)v(n,k)\\
& + a_{1,1}(n)u(n,k)v(n+1,k)+ a_{0,2}(n)v(n,k)v(n,k+1)\\
& + a_{1,2}(n)u(n,k)v(n+1,k)v(n+1,k+1) \bigr)f(n,k) = 0.
\end{alignat*}
After dividing by~$f(n,k)$, the left-hand side is an explicit rational function in $n$ and $k$ whose numerator
depends linearly on the unknown coefficients~$a_{i,j}(n)$. We now equate the coefficients with respect to $k$
to zero and solve the resulting linear system. 

\begin{ex}
 For $f(n,k)=\binom nk$, let us search for a $k$-free recurrence of the form
 \[
  a_{0,0}(n)\binom nk + a_{1,0}(n)\binom{n+1}k + a_{0,1}(n)\binom n{k+1} + a_{1,1}(n)\binom{n+1}{k+1} = 0
 \]
 Dividing by $\binom nk$ and working out the quotients leads to
 \[
  a_{0,0}(n) + a_{1,0}(n)\frac{n+1}{n+1-k} + a_{0,1}(n)\frac{n-k}{k+1} + a_{1,1}(n)\frac{n+1}{k+1} = 0
 \]
 Bringing the left-hand side on a common denominator and equating the coefficients of the numerator with
 respect to powers of $k$ to zero leads to the linear system
 \[
  \begin{pmatrix}
    -n-1 & -n (n+1) & -n-1 & -(n+1)^2 \\
    -n & 2 n+1 & -n-1 & n+1 \\
    1 & -1 & 0 & 0 \\
  \end{pmatrix}\begin{pmatrix}
    a_{0,0}\\a_{0,1}\\a_{1,0}\\a_{1,1}
  \end{pmatrix}=0
 \]
 The solution space of this sytem is generated by the vector $(-1,-1,0,1)$, which amounts to Pascal's
 recurrence $-f(n,k)-f(n,k+1)+f(n+1,k+1)=0$.  
\end{ex}

In summary, Sister Celine's method works as follows.

\begin{alg}[Sister Celine's method]

\noindent INPUT: a hypergeometric term $f(n,k)$, specified by the rational functions $u$ and $v$ as in Def.~\ref{def:hg2}

\noindent OUTPUT: a recurrence for the sum $S(n)=\sum_k f(n,k)$

\begin{enumerate}
\item Choose $r,s\in\set N$.
\item Use linear algebra to search for a $k$-free recurrence of $f(n,k)$ of order $r$ with respect to $n$
 and order $s$ with respect to~$k$.
\item If there is one, translate it to a recurrence for $S(n)$ and return it.
\item Otherwise, increase $r$ and $s$ and try again. 
\end{enumerate}
\end{alg}

The question is now whether this method always succeeds. The answer is no. There are two ways how it can fail:
a) if a certain hypergeometric term $f(n,k)$ does not satisfy any $k$-free recurrence, so that the method
keeps increasing $r$ and $s$ forever without terminating, and
b) if a $k$-free recurrence for $f(n,k)$ is found but for some reason it cannot be translated into a
nontrivial recurrence for $S(n)$.

Fortunately, there are satisfactory fixes for both issues. The first issue can be repaired by slightly
restricting the scope of the method. It can be shown that every hypergeometric term $f(n,k)$ can be written
in the form
\[
  q(n,k)\phi^n\psi^k \prod_{m=1}^M (a_mn+b_mk+c_m)!^{e_m}
\]
for some rational function~$q$, some constants $\phi,\psi,c_m$, and some integers $a_m,b_m,e_m$. (It is not
difficult to check that every such expression gives rise to a hypergeometric term, but the converse we
claim here is not trivial.) If we define that $f(n,k)$ is called \emph{proper hypergeometric} if it can be
written as above, but with $q$ being a polynomial, then it can be shown that every proper hypergeometric
term satisfies a $k$-free recurrence. A proof can be found in~\cite[Thm. 4.4.1]{PWZbook1996} or in~\cite[Thm. 5.14]{Kauers2023}.
Applied to a proper hypergeometric term, Sister Celine's method will not run into an endless loop.

As an example for the second issue, observe that also
\[
 f(n,k)-f(n+1,k+1)-f(n,k+2)+f(n+1,k+2)=0
\]
is a $k$-free recurrence for $f(n,k)=\binom nk$, but if we try to translate it to a recurrence for the sum
$S(n)=\sum_k\binom nk$, we obtain $0=0$, which is correct but not very useful. 
To make it clear what is going on here, it is best to formulate things in terms of operators.
Write $S_x$ for the operator that maps $f(x)$ to $f(x+1)$ and write $x$ for the operator
that maps $f(x)$ to $x f(x)$. Also write $\Delta_x:=S_x-1$ for the forward shift operator,
which maps $f(x)$ to $f(x+1)-f(x)$. 

Note that we have the relations $S_xx=(x+1)S_x$, $x\Delta_x=\Delta_xx-1$, and $\sum_x \Delta_x=0$.
In view of the latter, we write the $k$-free recurrence in the form
\[
   \bigl(P(n,S_n) + \Delta_k Q(n,\Delta_k,S_n)\bigr)\cdot f(n,k)=0.
\]
Then we see that $P(n,S_n)$ annihilates the sum $\sum_k f(n,k)$. We are in trouble if this is zero,
i.e., if the operator corresponding to the $k$-free recurrence happens to have $\Delta_k$ as a left
factor. If this is the case, i.e., if we have 
\[
  \Delta_k Q(n,\Delta_k,S_n)\cdot f(n,k)=0
\]
for some operator~$Q$, then we can multiply this equation (from the left) with $k$ and apply the
commutation rule $x\Delta_x=\Delta_xx-1$ to obtain
\[
 \bigl(-Q(n,\Delta_k,S_n) + \Delta_k k Q(n,\Delta_k,S_n)\bigr)\cdot f(n,k)=0.
\]
We then write $-Q(n,\Delta_k,S_n)$ as $\tilde P(n,S_n)+\Delta_k\tilde Q(n,\Delta_k,S_n)$, and repeat
the process if $\tilde P(n,S_n)$ is again the zero operator. After at most $\deg_{\Delta_k}Q$
repetitions, we are guaranteed to get a nontrivial result (Wegschaider's lemma, \cite[Lemma~5.19]{Kauers2023}).

At first glance, multiplying a $k$-free recurrence with $k$ seems counterproductive. Indeed, after
performing this multiplication, the recurrence is no longer $k$-free. Instead, at the end of the
process we have a recurrence which can be written in the form
\[
   \bigl(P(n,S_n) + \Delta_k Q(n,k,\Delta_k,S_n)\bigr)\cdot f(n,k)=0,
\]
where $P$ is a nonzero $k$-free operator and $Q$ is some other operator that may or may not involve
$k$ and may or may not be the zero operator. This is good enough, because a presence of $k$ in $Q$
does not prevent us from concluding that $P(n,S_n)$ is an annihilating operator for the sum $\sum_k f(n,k)$.

\section{Implementation of Sister Celine's Method}

We have just learned that Sister Celine's method was the first
deterministic method to derive a recurrence equation satisfied by
a hypergeometric sum. As we will see later, it is usually not the
most efficient one (in terms of computation time and size of the
output), but it is very easy to implement (hence, most efficient in
terms of implementation time). In this section, we want to
demonstrate how this can be done. We use the computer algebra
system \emph{Mathematica} for this purpose.

In addition to the two summation identities from the previous section,
here are some more examples that can be treated with Sister Celine's
method and that we may use as test cases:
\begin{align*}
  & \sum_{k=0}^n \binom{n}{k}^{\!2} = \binom{2n}{n}
  \\
  & \sum_{k=0}^n \binom{n}{k}^{\!2} \binom{n+k}{k}^{\!2}
  \quad\leadsto\text{ second-order recurrence}
  \\
  & \sum_{k} (-1)^k \binom{l+m}{l+k} \binom{m+n}{m+k} \binom{n+l}{n+k} =
  \frac{(l+m+n)!}{l! \, m! \, n!}
\end{align*}

For our implementation, we assume that the hypergeometric summand
expression is called~$f(n,k)$, where $k$ is the summation variable
and $n$ is the free parameter. Also we assume that the orders
$r$ and~$s$ are prescribed by the user, so that we need not implement
a loop for increasing the orders. Before starting with the implementation,
we recall the main steps of Sister Celine's method:
\begin{enumerate}
\item Choose $r,s\in\set N$ (order in~$n$, order in~$k$).
\item\label{step:ansatz} Make an ansatz for a $k$-free recurrence:
  $\sum_{i=0}^r\sum_{j=0}^s c_{i,j}(n) \cdot f(n+i,k+j)$.
\item\label{step:divide} Divide the ansatz by $f(n,k)$ and simplify.
\item\label{step:denom} Multiply the ansatz by its common denominator.
\item\label{step:coeffcomp} Perform coefficient comparison with respect to~$k$.
\item\label{step:solve} Solve the linear system for the unknowns $c_{i,j}(n)$.
\item\label{step:sum} Sum over the $k$-free recurrences and return the result.
\end{enumerate}

\noindent
According to Step~\ref{step:ansatz},
we start by writing down the ansatz
\begin{verbatim}
ansatz = Sum[c[i, j] * (f /. {n -> n + i, k -> k + j}),
  {i, 0, r}, {j, 0, s}]
\end{verbatim}
Using, for example, $f(n,k)=\binom{n}{k}$ and $r=s=1$, as we did in the
previous section in order to derive Pascal's identity, we obtain
\begin{multline*}
  c[0, 0]\,\text{Binomial}[n, k] + c[0, 1]\,\text{Binomial}[n, 1 + k] \\
  + c[1, 0]\,\text{Binomial}[1 + n, k] + c[1, 1]\,\text{Binomial}[1 + n, 1 + k].
\end{multline*}
Now we have to divide this ansatz by~$f(n,k)$ and simplify (Step~\ref{step:divide}).
The best way
to force the system to simplify the shift quotient of a hypergeometric
term into a rational function is by using the command \texttt{FunctionExpand}.
Naively, this could be done as follows:
\begin{verbatim}
ansatz = FunctionExpand[Sum[
  c[i, j] * (f /. {n -> n + i, k -> k + j}),
  {i, 0, r}, {j, 0, s}] / f]
\end{verbatim}
If we try this with the same input as above, we correctly obtain
\[
  (-k + n) \left( \frac{c[0, 0]}{-k + n} + \frac{c[0, 1]}{1 + k}
    + \frac{(1 + n)\,c[1, 0]}{(-1 + k - n) (k - n)}
    + \frac{(1 + n)\,c[1, 1]}{(1 + k) (-k + n)} \right).
\]
However, we claim that this is not the most clever way and that our
implementation will be more efficient if we tell the system more
precisely what to do. In order to demonstrate this, we generate
a more complicated example (the one above was obtained instantaneously):
\begin{verbatim}
Product[(Table[RandomInteger[{-10, 10}], {3}] .
  {n, k, 1})!^((-1)^i), {i, 6}]
\end{verbatim}
yielding the expression
\begin{equation}\label{eq:expr}
  \frac{(-9k-4n-2)! \, (9k-3n-2)! \, (3k+7n-3)!}{%
    (-k+3n+8)! \, (k+5n-8)! \, (2k+9n+5)!}.
\end{equation}
We measure the time (using Mathematica's \texttt{Timing} command)
how long it takes to divide the ansatz by the summand~$f(n,k)$ and
to simplify all its coefficients to rational functions. The output
for our more complicated example is a bit unhandy and therefore not
shown here. The computation took $6.26$ seconds, which is embarrassingly slow.

Hence, it may not be a good idea to apply \texttt{FunctionExpand}
to the whole expression, because then the command has to figure out
which parts can be combined. Instead, we expect it to work better
when we simplify each shift quotient individually. This is achieved
by the following code:
\begin{verbatim}
ansatz = Sum[c[i, j] *
  FunctionExpand[(f /. {n -> n + i, k -> k + j}) / f],
  {i, 0, r}, {j, 0, s}]
\end{verbatim}
Calling this with the expression~\eqref{eq:expr} and $r=s=1$ as before,
we now get the result in much shorter time, namely $0.05$ seconds!
In addition, the result is given in a more simplified form, as can
be seen by comparing the output for $f(n,k)=\binom{n}{k}$ with the
previous output:
\[
  c[0, 0] + \frac{(-k + n)\,c[0, 1]}{1 + k}
  + \frac{(1 + n)\,c[1, 0]}{1 - k + n}
  + \frac{(1 + n)\,c[1, 1]}{1 + k}.
\]
Still, we are not satisfied yet, because calling the above code with
expression~\eqref{eq:expr} and $r=s=10$ results in a computation that
takes more than $90$ seconds. We are able to reduce it to $0.19$ seconds
by explicitly computing and using the rational function certificates $u(n,k)$
and~$v(n,k)$ of the hypergeometric input~$f(n,k)$:
\begin{verbatim}
u = FunctionExpand[(f /. n -> n + 1) / f]
v = FunctionExpand[(f /. k -> k + 1) / f]
\end{verbatim}
Having these certificates, we can easily compute any higher shift
quotient, e.g.,
\[
  \frac{f(n+10,k)}{f(n,k)} = \prod_{i=0}^9 u(n+i,k).
\]
Implementing this idea results in the following code (note that it
does not matter whether we first use~$u$ to reduce the shifts in $n$-direction
and then~$v$ to do the $k$-direction, or vice versa; the result will be
the same, which is ensured by the so-called compatibility conditions):
\begin{verbatim}
ansatz = Sum[c[i, j] *
  Product[u /. {n -> n + i1, k -> k + j}, {i1, 0, i - 1}] *
  Product[v /. k -> k + j1, {j1, 0, j - 1}],
  {i, 0, r}, {j, 0, s}]
\end{verbatim}
In view of the fact that later we will have to solve a linear system
for the unknowns $c_{i,j}$, it is advantageous to construct the
corresponding matrix from the very beginning, instead of carrying
large expressions that represent the equations of the system, from
which one has to extract the coefficients w.r.t.\ the~$c_{i,j}$.
It is achieved by replacing \texttt{Sum} in the previous code snippet
by \texttt{Table}, and by wrapping the whole four lines with
\texttt{Flatten} in order to convert the two-dimensional array
to a one-dimensional list, each of whose entries corresponds to a
column of the system matrix. Also the explicit multiplication by
$c_{i,j}$ can then be omitted.

In Step~\ref{step:denom} we are asked to multiply by the common denominator.
This is easily coded for a single expression, and only slightly more
complicated if we have already split our ansatz into pieces, as
described in the previous paragraph: we have to compute the least
common multiple of the denominators of the individual pieces, and then
multiply:
\begin{verbatim}
Together[ansatz * (PolynomialLCM @@ Denominator[ansatz])]
\end{verbatim}
Note that commands like \texttt{Together} (combine fractions and
cancel common factors) and \texttt{Denominator} are ``listable'',
i.e., the denominator of a list of rational functions is a list
of the respective denominators.

We now have to perform coefficient comparison with respect to~$k$
(Step~\ref{step:coeffcomp}), for which purpose we use the command
\texttt{CoefficientList} that is mapped over the list \texttt{ansatz}.
The Mathematica syntax for this is \texttt{/@}.
One has to take care because the individual
parts of the ansatz may have different degrees in~$k$, resulting in
a ragged array instead of a rectangular matrix. The command
\texttt{PadRight} cures this. Finally we have to \texttt{Transpose}
since the elements in our ansatz list should correspond to columns
of the matrix, not rows:
\begin{verbatim}
mat = Transpose[PadRight[CoefficientList[#, k]& /@ ansatz]]
\end{verbatim}

Step~\ref{step:solve} requires solving the system, i.e., computing
the kernel of the constructed matrix. The following command yields
a basis for the kernel:
\begin{verbatim}
ker = NullSpace[mat]
\end{verbatim}

If this kernel is trivial, then the try was unsuccessful, which means
that the user will have to increase $r$ or~$s$. On the other hand, it could
be that the kernel has dimension greater than~$1$; this happens, for instance,
when the values chosen for $r$ and $s$ are not the minimal ones. We
decide to return a list of recurrences, one for each basis element of
the kernel. In case of failure, this will be the empty list. Each
basis vector of the kernel corresponds to a $k$-free recurrence that
we have to sum with respect to~$k$ (Step~\ref{step:sum}).
However, ``summing'' just means
replacing $f(n+i,k+j)$ by $\text{SUM}[n+i]$, where the new function
$\text{SUM}[n]$ represents the sum $\sum_k f(n,k)$. It is achieved by
taking the scalar product of each kernel vector with a vector of the
corresponding $\text{SUM}[n+i]$ expressions:
\begin{verbatim}
re = ker . Flatten[Table[SUM[n + i], {i, 0, r}, {j, 0, s}]]
\end{verbatim}

Before returning the recurrences, one may want to polish them a little
bit, e.g., by removing denominators that could be produced by the
\texttt{NullSpace} procedure. Putting everything together, we obtain
a basic implementation of Sister Celine's method:
\begin{verbatim}
Celine[f_, n_, k_, r_, s_] :=
Module[{u, v, mat, ker, re},
  u = FunctionExpand[(f /. n -> n + 1) / f];
  v = FunctionExpand[(f /. k -> k + 1) / f];
  mat = Flatten[Table[
    Product[u /. {n -> n + i1, k -> k + j}, {i1, 0, i - 1}] *
    Product[v /. k -> k + j1, {j1, 0, j - 1}],
    {i, 0, r}, {j, 0, s}]];
  mat = Together[mat * (PolynomialLCM @@ Denominator[mat])];
  mat = Transpose[PadRight[CoefficientList[#, k]& /@ mat]];
  ker = NullSpace[mat];
  If[ker === {}, Return[{}]];
  re = ker . Flatten[Table[SUM[n + i], {i, 0, r}, {j, 0, s}]];
  re = Collect[Numerator[Together[#]], _SUM, Expand]& /@ re;
  Return[re];
];
\end{verbatim}

We can try a few examples to convince ourselves that everything
works fine. For example, the input
\begin{verbatim}
Celine[Binomial[n, k], n, k, 1, 1]
\end{verbatim}
yields the output
\[
  \bigl\{-2\,\text{SUM}[n] + \text{SUM}[1 + n]\bigr\}.
\]
Similarly, and very quickly, the input
\begin{verbatim}
Celine[Binomial[n, k]^2, n, k, 2, 2]
\end{verbatim}
yields the output
\[
  \bigl\{(-6 - 4 n)\,\text{SUM}[1 + n] + (2 + n)\,\text{SUM}[2 + n]\bigr\}.
\]
Also the other example from Section~\ref{sec:Celine} is computed in very short time:
\begin{verbatim}
Timing[Celine[(-1)^k * Binomial[2n, n + k]^2, n, k, 2, 4]]
\end{verbatim}
yields
\begin{align*}
  & \bigl\{0.511511, \bigl\{
  \bigl(112 + 400 n + 416 n^2 + 128 n^3\bigr)\,\text{SUM}[n] \\
  &\quad {} + \bigl(-110 - 288 n - 240 n^2 - 64 n^3\bigr)\,\text{SUM}[1 + n] \\
  &\quad {} + \bigl(18 + 45 n + 34 n^2 + 8 n^3\bigr)\,\text{SUM}[2 + n]\bigr\}\bigr\}
\end{align*}
For the Ap\'{e}ry numbers
\[
  \sum_{k=0}^n \binom{n}{k}^2 \binom{n+k}{k}^2,
\]
the computation takes a bit longer and does
not deliver the minimal-order recurrence (whose order is~$2$):
\begin{verbatim}
Timing[Celine[Binomial[n, k]^2 * Binomial[n + k, k]^2,
  n, k, 4, 3]]
\end{verbatim}
produces
\begin{align*}
  & \bigl\{10.3229, \bigl\{
  \bigl(-504 - 2076 n - 3408 n^2 - 2832 n^3 - 1248 n^4 - 276 n^5 - 24 n^6\bigr)
  \,\text{SUM}[n]
  \\ & \quad
  + \bigl(63000 + 194316 n + 245760 n^2 + 162672 n^3 + 59256 n^4 + 11232 n^5
  \\ & \quad
  {} + 864 n^6\bigr)
  \,\text{SUM}[1 + n] +
  \bigl(-277560 - 734604 n - 798792 n^2 - 457224 n^3
  \\ & \quad
  {} - 145392 n^4 - 24360 n^5 - 1680 n^6\bigr)
  \,\text{SUM}[2 + n] +
  \bigl(224280 + 564636 n
  \\ & \quad
  {} + 578136 n^2 + 308280 n^3 + 90360 n^4 + 13824 n^5 + 864 n^6\bigr)
  \,\text{SUM}[3 + n] + \bigl(-9216
  \\ & \quad {} - 22272 n - 21696 n^2 - 10896 n^3 - 2976 n^4 - 420 n^5 - 24 n^6\bigr)
  \,\text{SUM}[4 + n]\bigr\}\bigr\}.
\end{align*}
We will come back to these issues in Section~\ref{sec:Zeilberger}.

\section{Creative Telescoping}\label{sec:ct}

The bottom line of our discussion of Sister Celine's method was that annihilating operators of the form
\[
  P(n,S_n) - \Delta_k Q(n,k,\Delta_k,S_n)
\]
are useful for solving summation problems.
Here, the operator $P$ is supposed to be nonzero and free of $k$ and $Q$ is some other operator that may
or may not involve $k$ and may or may not be nonzero.
If a sequence $f(n,k)$ has an annihilating operator of the above form, then $P(n,S_n)$ is an annihilating
operator for the sum $S(n)=\sum_k f(n,k)$ (provided that a technical condition on the summation range
is satisfied).

We call $P(n,S_n)$ a \emph{telescoper} and $Q(n,k,\Delta_k,S_n)$ a \emph{certificate}. This terminology has
the following background. A hypergeometric term $f(n,k)$ is said to \emph{telescope} (w.r.t.~$k$) if there
is another hypergeometric term $g(n,k)$ such that $f(n,k)=g(n,k+1)-g(n,k)=\Delta_k\cdot g(n,k)$. In this
case, thanks to the telescoping effect, we have $\sum_{k=0}^m f(n,k) = g(n,m+1) - g(n,0)$. Not every
hypergeometric term telescopes in this sense. The effect of applying $P(n,S_n)$ to $f(n,k)$ is that it
maps $f(n,k)$ to a telescoping hypergeometric term. Whether a hypergeometric term telescopes or not is
not obvious, and whether a given operator $P$ is a telescoper for a given hypergeometric term $f$ is
also not obvious. However, if we know not only $P$ but also~$Q$, then we can check easily whether $(P-\Delta_kQ)\cdot f$
is zero. Thus $Q$ certifies that $P$ is a telescoper.

The concept of telescopers applies more generally, as indicated in the following examples. 

\begin{ex}
\begin{enumerate}
\item Let $W_n=\int_0^{\pi/2}(\sin(x))^ndx$. The indefinite integral $\int\sin(x)^ndx$ is not easily expressible,
but we have
\[
  \int((n+1)\sin(x)^n-(n+2)\sin(x)^{n+2})dx=\sin(x)^{n+1}\cos(x),
\]
and consequently
\[
  (n+1)W_n - (n+2)W_{n+2}=\bigl[\sin(x)^{n+1}\cos(x)\bigr]_0^{\pi/2}=0.
\]
\item Consider the formal power series $S(x)=\sum_{n=0}^\infty\binom{2n}nx^n$. The indefinite sum $\sum_{n=0}^N\binom{2n}nx^n$
is not easily expressible, but we have
\[
  \sum_{n=0}^N\bigl(2\binom{2n}nx^n-(1-4x)\binom{2n}n(x^n)'\bigr)=(N+1)\binom{2N+2}{N+1}x^{N+1},
\]
and consequently
\[
  2S(x)-(1-4x)S'(x)=\lim_{N\to\infty}(N+1)\binom{2N+2}{N+1}x^{N+1}=0,
\]
where the limit is meant in the topology of formal power series.
The differential equation $2S(x)-(1-4x)S'(x)=0$ together with the initial value $S(0)=1$ implies
the identity $S(x)=\frac1{\sqrt{1-4x}}$.
\item For $f(x,t)=\frac1{\sqrt{(1-x^2)(1-tx^2)}}$, consider $K(t)=\int_0^1f(x,t)dx$. The indefinite integral $\int f(x,t)dx$
is not easily expressible, but we have
\[
  \int(4(1-t)t\frac{\partial^2}{\partial t^2}f + 4(1-2t)\frac{\partial}{\partial t}f - f) dx=-x\sqrt{\frac{1-x^2}{(1-tx^2)^3}},
\]
and consequently
\[
 4(1-t)t K''(t) + 4(1-2t)K'(t) - K(t) = \bigl[-x\sqrt{\frac{1-x^2}{(1-tx^2)^3}}\bigr]_{x=0}^1=0.
\]
\end{enumerate}
\end{ex}

In these examples, we see that operators of the form $P-D_xQ$ instead of $P-\Delta_kQ$ are useful for integration,
and if the free variable is continuous rather than discrete, we get a differential operator (w.r.t. $t$) rather than
a recurrence operator (w.r.t. $n$) as~$P$. The terminology of telescopers and certificates extends to these (and some further)
variants. The search for telescopers is called \emph{creative telescoping,} a term that first appeared in van der Poorten's
account on Ap\'ery's proof of the irrationality of $\zeta(3)$~\cite{vanDerPoorten79}. The differential version is also known as
\emph{differentiating under the integral sign} and was promoted by Feynman.

Sister Celine's method is a rudimentary algorithm for creative telescoping applicable to hypergeometric terms. More
sophisticated algorithms for this case will be discussed in the next two sections. For the remainder of the present
section, let us focus on the differential case and on rational functions. The task is thus as follows: given a rational
function $f\in C(x,y)$, find operators $P(x,D_x)$ (nonzero!) and $Q(x,y,D_x,D_y)$ such that $(P-D_yQ)\cdot f=0$.

Equivalently, we can ask for an operator $P(x,D_x)$ (nonzero!) and a rational function $g\in C(x,y)$ such that $P\cdot f=D_y\cdot g$.

Recall the Ostrogradsky--Hermite reduction algorithm from Sect.~\ref{sec:hermite}. Given a rational function $f\in C(y)$, it produces two rational
functions $g,h\in C(y)$ such that:
\def\num{\operatorname{num}}\def\den{\operatorname{den}}
\begin{enumerate}
\item $f=D_y\cdot g+h$
\item the denominator $\den(h)$ of $h$ is squarefree
\item the numerator $\num(h)$ of $h$ has a lower degree than the denominator $\den(h)$. 
\end{enumerate}
The key property of this output is that $f$ is integrable in $C(y)$ if and only if $h=0$. 

Now consider a bivariate rational function $f\in C(x,y)$ and apply Ostrogradsky--Hermite reduction w.r.t.~$y$ to $f$ and its derivatives
w.r.t. $x$:
\begin{alignat*}1
  f &= D_y\cdot g_0 + h_0\\
  D_x\cdot f&= D_y\cdot g_1 + h_1\\
  D_x^2\cdot f&= D_y\cdot g_2 + h_2\\
   &\vdots
\end{alignat*}
If we can find a $C(x)$-linear relation among the rational functions $h_0,h_1,h_2,\dots$, say
\[
  c_0h_0 + \cdots + c_rh_r = 0
\]
for certain $c_0,\dots,c_r\in C(x)$, not all zero, then we get
\[
  (c_0 + c_1 D_x + \cdots + c_r D_x^r)\cdot f = D_y g,
\]
where $g=c_0g_0 + \cdots + c_r g_r$, because the $c_0,\dots,c_r$ are free of $y$ and therefore commute with~$D_y$.
We see that a linear dependence among the $h_0,h_1,\dots$ directly translates into a telescoper for~$f$.
The approach to compute a telescoper by searching for a linear dependence among $h_0,h_1,\dots$ is known
as \emph{reduction-based telescoping.}

\def\sqfp{\operatorname{sqfp}}%
Since the denominator of each $h_i$ divides the squarefree part of the denominator of~$f$ and the numerator degree
of each $h_i$ is bounded by its denominator degree, all the $h_0,h_1,\dots$ belong to a finite-dimensional
$C(x)$-linear subspace of~$C(x,y)$. It is therefore guaranteed that we find a linear dependence among
these rational functions. In fact, we can say more precisely that we must find a telescoper of order at most
$\deg_y\sqfp(\den(f))$, where $\sqfp(\cdot)$ denotes the squarefree part. 

Not only does every linear relation among $h_0,h_1,\dots$ translate into a telescoper for~$f$, but also every
telescoper for $f$ gives rise to a linear relation among $h_0,h_1,\dots$. Indeed, if $P=c_0+c_1D_x+\cdots+c_rD_x^r$
is a telescoper for~$f$, then $P\cdot f$ is integrable in $C(x,y)$ w.r.t.~$y$. Also $D_y\cdot(c_0g_0+\cdots+c_rg_r)$
is clearly integrable w.r.t.~$y$, and therefore $c_0h_0+\cdots+c_rh_r=P\cdot f-D_y\cdot(c_0g_0+\cdots+c_rg_r)$ is
integrable w.r.t.~$y$. But since $c_0h_0+\cdots+c_rh_r$ is a proper rational function with a squarefree denominator,
it can only be integrable if it is zero. This completes the proof that every telescoper translates to a linear
dependence among $h_0,h_1,\dots$. In particular, the linear relation with the smallest $r$ corresponds to the
telescoper of minimal order.

\medskip
One feature of reduction-based telescoping is that the linear systems we have to solve in order to find the
linear relation among $h_0,h_1,\dots$ are relatively small. If we do not care about efficiency, there is an
easier way to find a telescoper for a rational function. This alternative approach, which goes back to Apagodu
and Zeilberger~\cite{mohammed05,apagodu06,chen14a}, does not rely on Ostrogradsky--Hermite reduction but only uses linear algebra.

Write $f=\frac pq$ for two polynomials $p,q\in C[x,y]$, and assume for simplicity that $\deg_y p<\deg_y q$.
The first few derivatives of $f$ w.r.t. $x$ can be expressed in terms of $p$ and $q$ and their derivatives.
\begin{alignat*}1
  f &= \frac pq\\
  D_x\cdot f&= \frac{p'q-pq'}{q^2}\\
  D_x^2\cdot f&= \frac{\text{something}}{q^3}\\
              &\vdots
\end{alignat*}
The precise expressions in the numerators are not so important. It suffices to observe that we can take
$q^{i+1}$ as denominator of $D_x^i\cdot f$, and if we do so, then the $y$-degree of the numerator is at most
$\deg_yp+i\deg_yq$. (We do not care if some cancelation is possible to make things smaller.) The
rational functions $f,D_x\cdot f,\dots,D_x^r\cdot f$ have $q^{r+1}$ as a common denominator, and if we write
them all with this denominator, the numerators all have $y$-degree at most $\deg_yp+r\deg_yq$.

It follows that, for every choice $c_0,\dots,c_r\in C(x)$ (free of~$y$), we have that
\[
  (c_0+c_1D_x+\cdots+c_rD_x^r)\cdot f= \frac{\text{something}}{q^{r+1}},
\]
where ``something'' is a certain element of $C(x)[y]$ of degree at most $r\deg_yq+\deg_yp$. Now for some
$b_0,\dots,b_s\in C(x)$ (also free of~$y$), consider the rational function
\[
  g = \frac{b_0+b_1y+\cdots+b_sy^s}{q^r}
\]
and observe that
\[
  D_y\cdot g =\frac{\text{something}}{q^{r+1}}
\]
where ``something'' is a certain element of $C(x)[y]$ of degree at most $\deg_yq+s-1$.
For $s=(r-1)\deg_yq+\deg_yp+1$ the degree bound becomes $r\deg_yq+\deg_yp$ and matches
the degree bound of the first ``something''. 

We want the two ``something''s to become equal. To this end, regard the $c_0,\dots,c_r$
and $b_0,\dots,b_s$ as undetermined coefficients and compare coefficients with respect
to $y$ in order to obtain a linear system over $C(x)$ for the undetermined coefficients.

This system will have $(r+1)+((r-1)\deg_yq+\deg_yp+2)$ variables and $1+r\deg_y q+\deg_yp$
equations, so it is guaranteed to have a nontrivial solution as soon as $r\geq\deg_yq-1$,
because in this case, it has more variables than equations.
Nontrivial means that at least one of the $c_0,\dots,c_r$ or at least one of the $b_0,\dots,b_s$
is nonzero. This is not quite enough, because we want a nonzero telescoper, so we need to
ensure that at least one of $c_0,\dots,c_r$ is nonzero. We encourage the reader to find
out why this is always the case.

\begin{ex}
  Let us compute a telescoper for $f=\frac1{xy^3+y+1}$ with both methods explained above.
  \begin{enumerate}
  \item With Ostrogradsky--Hermite reduction, we find
    \begin{alignat*}3
      f &= D_y\cdot0 &&+\frac1{1+y+xy^3}\\
      D_x\cdot f&=D_y\cdot\frac{3xy^2+9xy+2y+2}{x(27x+4)(xy^3+y+1)} &&+\frac{3(y-3)}{(27x+4)(xy^3+y+1)}\\
      D_x^2\cdot f&=D_y\cdot\frac{-162x^3y^5+\cdots-16y-8}{x^2(27x+4)^2(xy^3+y+1)^2}&&+\frac{6(54x-27xy-y-1)}{x(27x+4)^2(xy^3+y+1)}.
    \end{alignat*}
    The linear system
    \[
      \begin{pmatrix}
        1 & -9/(27x+4) & 6(54x-1)/x/(27x+4)^2 \\
        0 & 3/(27x+4) & -6(27x+1)/x/(27x+4)^2
      \end{pmatrix}\begin{pmatrix}
        c_0\\c_1\\c_2
      \end{pmatrix}=0
    \]
    has a solution space generated by $(6,2(27x+1),x(27x+4))$. Therefore,
    \[
      6 + 2(27x+1)D_x + x(27x+4)D_x^2 
    \]
    is a telescoper for~$f$.
  \item We have
    \[
      (c_0 + c_1D_x + c_2D_x^2)\cdot f=
      \tfrac{(c_0 x^2-c_1 x+2 c_2)y^6+ (2 c_0 x-c_1)y^4+ (2 c_0 x-c_1)y^3+c_0y^2+2 c_0 y+c_0}{(1+y+xy^3)^3}
    \]
    and
    \begin{alignat*}1
      &D_y\cdot\frac{b_0+b_1y+b_2y^2+b_3y^3+b_4y^4}{(xy^3+y+1)^2}\\
      &=\tfrac{-2 b_4 x y^6-3 b_3 x y^5+ (2 b_4-4 b_2 x)y^4+(-5 b_1 x+b_3+4 b_4)y^3+ (3 b_3-6 b_0 x)y^2+ (2 b_2-b_1)y+(b_1-2 b_0)
      }{(1+y+xy^3)^3}.
    \end{alignat*}
    Equating coefficients leads to the linear system
    \[
      \begin{pmatrix}
        1 & 0 & 0 & 2 & -1 & 0 & 0 & 0 \\
        2 & 0 & 0 & 0 & 1 & -2 & 0 & 0 \\
        1 & 0 & 0 & 6 x & 0 & 0 & -3 & 0 \\
        2 x & -1 & 0 & 0 & 5 x & 0 & -1 & -4 \\
        2 x & -1 & 0 & 0 & 0 & 4 x & 0 & -2 \\
        0 & 0 & 0 & 0 & 0 & 0 & 3 x & 0 \\
        x^2 & -x & 2 & 0 & 0 & 0 & 0 & 2 x 
       \end{pmatrix}\begin{pmatrix}
         c_0\\c_1\\c_2\\b_0\\b_1\\b_2\\b_3\\b_4
       \end{pmatrix}=0.
    \]
    Its solution space is generated by the vector
    \[
      \Bigl(6,2 (27 x+1),x (27 x+4),-\tfrac{1}{x},\tfrac{2 (3 x-1)}{x},-\tfrac{1-9 x}{x},0,-3 (x+1)\Bigr).
    \]
    The first three components of this vector give rise to the telescoper.
  \end{enumerate}
\end{ex}


\section{Gosper's Algorithm}\label{sec:Gosper}

In 1978, Gosper presented a celebrated algorithm in~\cite{gosper78} for solving the indefinite summation problem
of hypergeometric sequences.  Recall that a sequence $H(n)$ is \textit{hypergeometric} over $K(n)$ if 
\[\frac{H(n+1)}{H(n)}\in K(n).\]
The sequences $n^2+1$, $1/n$, $2^n$, $n!$ and $\Gamma(\alpha+n)$ are all hypergeometric sequences.
Two hypergeometric sequences $H_1(n)$ and $H_2(n)$ are said to be \textit{similar}, denoted by $H_1\sim H_2$, if
\[\frac{H_1(n)}{H_2(n)}\in K(n).\]

\medskip
\textbf{Hypergeometric Summation Problem.}~~
Given a hypergeometric sequence $H(n)$, decide whether there exists another hypergeometric sequence $G(n)$ such that
\[H(n)=G(n+1)-G(n) \triangleq \Delta_n(G(n)).\]
If such a $G$ exists, we say that $H(n)$ is \emph{hypergeometric summable}.
\medskip

\begin{rem}
If $H(n)=\Delta_n(G(n))$ for some hypergeometric $G(n)$, then $G(n+1)=g(n)G(n)$ for some rational $g(n)\in K(n)$. 
Thus $H(n)=(g(n)-1)G(n)$, which means that $H(n)$ and $G(n)$ are similar over $K(n)$. 
\end{rem}

Let $H(n)$ be a hypergeometric sequence with $H(n+1)/H(n)=f(n)\in K(n)$. To decide the existence of $G(n)$, we make an ansatz $G(n)=y(n)H(n)$. Then $H(n)=\Delta_n(G(n))$ if and only if 
\begin{align*}
    H(n)&=y(n+1)H(n+1)-y(n)H(n)\\
    &=\left(y(n+1)f(n)-y(n)\right)H(n).
\end{align*}
This leads to a recurrence equation for $y(n)$ 
\begin{align}\label{eq-Gosper-eq-1}
    1=f(n)y(n+1)-y(n).
\end{align}
We need to solve this equation for a rational solution in~$K(n)$.
The first step of Gosper's algorithm is to find the denominator of~$y(n)$.

\begin{defi}
Let $f\in K(n)$. We call the triple $(p,q,r)\in K[n]^3$ a \textit{Gosper form} of $f$ if 
\[f=\frac{p(n+1)}{p(n)}\frac{q(n)}{r(n)},\]
where $q, r$ satisfy the following GCD condition
\begin{align}\label{eq-gosper-GCD-condition}
   \gcd(q(n),r(n+i))=1 \qquad\text{for any } i\in\bN. 
\end{align}
\end{defi}
We now show that any rational function has a Gosper form.
Let $f=a(n)/b(n)$ with $a, b\in K[n]$ and $\gcd(a, b)=1$. If $g=\gcd(a(n),b(n+j_0))\neq1$ for some $j_0\in\bN\setminus\{0\}$, 
then $g(n)\mid a(n)$ and $g(n)\mid b(n+j_0)$, which implies that $g(n-j_0)\mid b(n)$. Let $a(n)=g(n)\bar{a}(n)$ and $b(n)=g(n-j_0)\bar{b}(n)$. Then we have
\begin{align*}
    \frac{a(n)}{b(n)}&=\frac{g(n)\bar{a}(n)}{g(n-j_0)\bar{b}(n)}= \frac{g(n)}{g(n-1)}\frac{g(n-1)}{g(n-2)}\cdots\frac{g(n-j_0+1}{g(n-j_0)}\frac{\bar{a}(n)}{\bar{b}(n)}\\
    &=\frac{p_1(n+1)}{p_1(n)}\frac{\bar{a}(n)}{\bar{b}(n)},
\end{align*}
where $p_1(n)=g(n-1)\cdots g(n-j_0)$. We can iterate the above process for $\bar{a}/\bar{b}$ until it satisfies the condition~\eqref{eq-gosper-GCD-condition}.

Let the triple $(p,q,r)\in K[n]^3$ be a \textit{Gosper form} of $f$. Now we would transform the equation~\eqref{eq-Gosper-eq-1} to \textit{Gosper's equation}
\begin{align}\label{eq-Gosper's-eq}
    p(n)=q(n)z(n+1)-r(n-1)z(n).
\end{align}

\begin{lem}[Gosper's Key Lemma]
The equation~\eqref{eq-Gosper-eq-1} has a rational solution if and only if Gosper's equation~\eqref{eq-Gosper's-eq} has a polynomial solution.  
\end{lem}
\begin{proof}
Note that the equation~\eqref{eq-Gosper-eq-1} and Gosper's equation~\eqref{eq-Gosper's-eq} can be transformed into each other by taking $y(n)=(r(n-1)/p(n))z(n)$. Therefore, we only need to
show the necessity.  Assume that the equation~\eqref{eq-Gosper-eq-1} has a rational solution. 
Then Gosper's equation~\eqref{eq-Gosper's-eq} also has a rational solution, say $z(n)=a(n)/b(n)\in K(n)$ with $a, b\in K[n]$ and $\gcd(a, b)=1$. 
Now it suffices to show that $b(n)$ is a nonzero constant in $K$. Otherwise, suppose that $b$ is not a constant. Then $b(n)$ has at least one non-trivial irreducible factor $u(n)\in K[n]$.
For this factor, there exists a maximal integer $j\in\bN$ such that $u(n+j)\mid b(n)$ but $u(n+j+1)\nmid b(n)$. We claim that $u(n+j+1)\mid q(n)$.
Since $u(n+j)\mid b(n)$, we have $u(n+j+1)\mid b(n+1)$.  
From the equation 
\begin{align}\label{eq-pf-Gosper's-Key-Lem}
    p(n)b(n)b(n+1)=q(n)a(n+1)b(n)-r(n-1)a(n)b(n+1), 
\end{align} 
 we get $u(n+j+1)\mid q(n)a(n+1)b(n)$. Since $\gcd(a(n), b(n))=1$ and  $u(n+j)\mid b(n)$, we get $u(n+j+1)\nmid a(n+1)$.
 Thus $u(n+j+1)\mid q(n)$.
On the other hand,  there exists a maximal integer $i\in\bN$ such that $u(n-i)\mid b(n)$ but $u(n-i-1)\nmid b(n)$. We claim that $u(n-i)\mid r(n-1)$.
Since $u(n-i)\mid b(n)$, we have $u(n-i)\mid r(n-1)a(n)b(n+1)$ by the equality~\eqref{eq-pf-Gosper's-Key-Lem}.
Since $\gcd(a, b)=1$ and $u(n-i-1)\nmid b(n)$, we have $u(n-i)\nmid a(n)b(n+1)$.
Thus $u(n-i)\mid r(n-1)$.  Combining the above two claims, we get that $\gcd(q(n),r(n+i+j))\neq 1$, a contradiction to the condition~\eqref{eq-gosper-GCD-condition}.
\end{proof}

To find a polynomial solution of Gosper's equation~\eqref{eq-Gosper's-eq}, we only need to bound the degree of $z(n)$ and 
then solve a linear system over $K$ by the method of undetermined coefficients. We first write Gosper's equation in the form
\[p(n)=q(n)\Delta_n(z(n))+(q(n)-r(n-1))z(n).\]
Assume that $z(n)=\sum_{i=0}^dz_in^i$. Then $\Delta_n(z(n))=dz_d n^{d-1}+$ lower terms. There are three cases to be considered.
\begin{enumerate}
    \item If $\deg(q(n))\leq\deg(q(n)-r(n-1))$, then $d=\deg(p(n))-\deg(q(n)-r(n-1))$.
    \item If $\deg(q(n))>\deg(q(n)-r(n-1))+1$, then $d=\deg(p(n))-\deg(q(n))+1$.
    \item If $\deg(q(n))=\deg(q(n)-r(n-1))+1=\rho$, then we can assume that the leading coefficients of $q(n)$ and $q(n)-r(n-1)$ are $c_1$ and $c_2$ respectively. Then  
    $$q(n)\Delta_n(z(n))+(q(n)-r(n-1))z(n)=(c_1dz_d+c_2z_d)n^{\rho+d-1}+\text{ lower terms.}$$ So either $d=-c_2/c_1$ or $\deg(p(n))=\rho+d-1$. 
\end{enumerate}
In summary, we can take $d$ as 
\[d = \max(\{-c_2/c_1,\deg(p(n))-\max\{\deg(q(n), \deg(q(n)-r(n-1)+1))\}  +1\}\cap\bN).\]  
The main steps of Gosper's algorithm are described as follows.

\begin{alg}[Gosper's algorithm]
\qquad

\noindent INPUT: a hypergeometric sequence $H(n)$;

\noindent OUTPUT: a rational solution of the equation $H(n)=\Delta_n(y(n)H(n))$ for $y(n)$ if there exists one; ``No'' otherwise.
\begin{enumerate}
    \item Compute a Gosper form $(p,q,r)\in K[n]^3$ of $f(n)=H(n+1)/H(n)$.
    \item Set $d = \max(\{-c_2/c_1,\deg(p(n))-\max\{\deg(q(n), \deg(q(n)-r(n-1)+1))\}+1\}\cap\bN)$,  where $c_1$ and $c_2$ are the leading coefficients of $q(n)$ and $q(n)-r(n-1)$ respectively.
    \item Make an ansatz $z(n)=\sum_{i=0}^dz_i n^i$ and then obtain a linear system $\mathcal{L}$ with respect to $z_i$'s by comparing the coefficients of the equality
    \[p(n)=q(n)z(n+1)-r(n-1)z(n).\]
    \item Solving the linear system $\mathcal{L}$, return $y(n)=r(n-1)z(n)/p(n)$ if there exists a solution for $z(n)$; ``No'' otherwise.
\end{enumerate}
\end{alg}

\begin{ex}
Let $H=\binom{m}{k}/\binom{n}{k}$. Then we can compute a rational function $y(k)$ such that $H=\Delta_k(yH)$ by Gosper's algorithm as follows:
\begin{enumerate}
    \item Computing a Gosper form $(p,q,r)$ of $H(k+1)/H(k)=(m-k)/(n-k)$, we obtain $p=1$, $q=m-k$ and $r=n-k$.
    \item Finding a polynomial solution with degree zero of the equation 
    \[ 1=(m-k)z(k+1)-(n-k+1)z(k),\]
     we get $z(k)=1/(m-n-1)$. So we have $y(k)=r(k-1)z(k)/p(k)=(n-k+1)/(m-n-1)$.
\end{enumerate}
This leads to the identity $\sum_{k=0}^mH(k)=(n+1)/(n-m+1)$ as proved in Section~\ref{sec:classical}.
\end{ex}

\section{Zeilberger's Algorithm}\label{sec:Zeilberger}

At this point we may wonder: why do we need yet another algorithm?
Recall that we have already seen two symbolic summation algorithms:
\begin{itemize}
\item Gosper's algorithm solves the indefinite hypergeometric
  summation problem
  \[
    f(k) = g(k+1)-g(k) \implies \sum_{k=a}^b f(k) = g(b+1)-g(a),
  \]
  but is not suitable for definite hypergeometric summation problems,
  such as
  \[
    \sum_{k=0}^n \binom{n}{k} = 2^n.
  \]
\item Sister Celine's method applies to definite hypergeometric summation,
  but unfortunately is slow in practice and often does not return the
  minimal-order recurrence. One of the reasons is that the condition
  ``$k$-free'' is too strong: the ansatz
  \[
    A(n,S_n,S_k) = P(n,S_n) + \Delta_k\cdot Q(n,S_n,S_k)
  \]
  could be generalized to
  \[
    A(n,k,S_n,S_k) = P(n,S_n) + \Delta_k\cdot Q(n,k,S_n,S_k).
  \]
  Another disadvantage of Sister Celine's method is that one has
  to choose two parameters~$r$ and~$s$ (orders w.r.t.~$n$ and~$k$).
\end{itemize}
Zeilberger's algorithm, which we will discuss in this section, to some extent
cures the problems the previous two algorithms were suffering from.

Recall that creative telescoping is a method
that deals with parametrized definite sums and integrals,
which yields differential/recurrence equations for them.
For example, the well-known identity
\[
  \sum_{k=1}^\infty \frac{1}{k^2} = \frac{\pi^2}{6}
\]
is not in the scope for the creative telescoping method because it
has no parameter and the result is just a real number. Hence the result
cannot be described reasonably by a recurrence or differential equation.
Let us generalize this identity by introducing a parameter:
\[
  \sum_{k=1}^\infty \frac{1}{k(k+n)} = \frac{\gamma+\psi(n)}{n}.
\]
Now the sum (and thus its evaluation) depend on a parameter~$n$ and
therefore the identity is indeed amenable to creative telescoping.
However, the method will not immediately deliver the closed-form
right-hand side, but a recurrence equation for it, in this case
\[
  (n+2)^2 F_{n+2} = (n+1) (2 n+3) F_{n+1} - n (n+1) F_n,
\]
where $F_n$ represents the sum expression on the left-hand side.

The method that is nowadays known as \emph{Zeilberger's algorithm}~\cite{Zeilberger1990c}
was termed the ``fast algorithm'' by Zeilberger himself, in contrast to the
elimination-based ``slow algorithm'', which is rarely used any more---for
manifest reasons.

When Gosper invented his algorithm, he actually missed the
opportunity to formulate an algorithm for definite hypergeometric summation,
as we will explain in the following. Assume we have a conjecture of the form
$\sum_k f(n,k) = h(n)$ with $f(n,k)$ having finite support and with $h(n)$
being hypergeometric, i.e., $h(n+1)/h(n)=p(n)/q(n)$ for polynomials $p$ and~$q$.
How can we prove it using Gosper's algorithm? Starting from the identity
\[
  q(n)h(n+1) - p(n)h(n)=0 
\]
we obtain, by the definition of~$h(n)$,
\[
  q(n)\sum_k f(n+1,k) - p(n)\sum_k f(n,k) = 0. 
\]
This can be reformulated to
\[
  \sum_k\Bigl(q(n)f(n+1,k) - p(n)f(n,k)\Bigr) = 0,
\]
which is a summation identity that is in principle \emph{provable}
by Gosper's algorithm! Hence we obtain the following recipe:
\begin{enumerate}
\item Apply Gosper's algorithm to $q(n)f(n+1,k) - p(n)f(n,k)$.
\item (Hopefully) obtain $g(n,k)$ such that
  \[
    q(n)f(n+1,k) - p(n)f(n,k) = g(n,k+1) - g(n,k).
  \]
\item Apply $\sum_k$ to the above identity.
\item Note that $g(n,k)$ has finite support, since it is a rational function multiple of~$f$.
\item Denoting $S(n):=\sum_k f(n,k)$ this yields $q(n)S(n+1)-p(n)S(n) = 0$.
\item Check that $h(0)=S(0)$. Hence $S(n)=h(n)$ for all~$n$.
\end{enumerate}

\begin{ex}
Let's look at the (toy) example
\[
  \sum_k \binom{n}{k} = 2^n.
\]
We have $h(n)=2^n$ and hence $h(n+1)-2h(n)=0$. Now we construct
\[
  f(n+1,k) - 2f(n,k) =
  \binom{n+1}{k} - 2\cdot\binom{n}{k} =
  \underbrace{\frac{2k-n-1}{n-k+1}\binom{n}{k}}_{=:\bar{f}(n,k)}
\]
Gosper's algorithm applied to $\bar{f}$ succeeds:
\[
  g(n,k) = \frac{k}{k-n-1} \bar{f}(n,k) = -\binom{n}{k-1}.
\]
This function $g$ has finite support, hence $\sum_k \bar{f}(n,k)=0$.
The original identity follows.
\end{ex}

But what do we do if we don't know the evaluation of the sum:
\[
  S(n) := \sum_k f(n,k) = {?}
\]
(We still assume natural boundaries, i.e., that $f$ has finite support w.r.t.~$k$).
Here, Zeilberger's algorithm enters the stage.

From the existence result in Section~\ref{sec:Celine},
we \emph{know} that a recurrence for $S(n)$ exists,
provided that $f(n,k)$ is a proper hypergeometric term in $n$ and~$k$.
But we \emph{don't know} its order and its coefficients. Hence, we can
proceed in the following way: Try increasing orders $r=0,1,\dots$ until the
algorithm succeeds, which will eventually be the case under the above
assumptions.
For each~$r$, write the recurrence with undetermined coefficients~$p_i\in C(n)$:
\[
  p_r(n)S(n+r) + \dots + p_1(n)S(n+1) + p_0(n)S(n) = 0
\]
and apply Gosper's algorithm to $p_r(n)f(n+r,k)+\dots+p_0(n)f(n,k)$.

And now happens what Zeilberger calls ``the miracle''. When applying
a parametrized variant of Gosper's algorithm to the hypergeometric term
\[
  \bar{f}(n,k)=p_r(n)f(n+r,k)+\dots+p_1(n)f(n+1,k)+p_0(n)f(n,k),
\]
the algorithm works in the very same way, despite the unknown parameters~$p_i$.
The first reason is that the $p_i$ appear \emph{only} in $c(k)$ in Gosper's equation
\[
  a(k)\cdot z(k+1)-b(k-1)\cdot z(k)=c(k).
\]
The second reason is that the $p_i$ appear \emph{linearly},
hence the final linear system can be solved
simultaneously for the $p_i$ and for the coefficients of~$z(k)$:
\[
  z(k) = \sum_{i=0}^d z_i(n) k^i.
\]
Because this is a fundamental result and the main reason 
why Zeilberger's algorithm works, we present its proof here
(which is taken from~\cite{PWZbook1996}).

We want to try a telescoper of order~$r$ on a hypergeometric input $f(n,k)$, i.e.,
we apply Gosper's algorithm to $\bar{f}(n,k):=\sum_{i=0}^r p_i(n)f(n+i,k)$.
Let
\[
  \frac{f(n+1,k)}{f(n,k)} =: u(n,k) =: \frac{u_1(n,k)}{u_2(n,k)}
\]
and
\[
  \frac{f(n,k+1)}{f(n,k)} =: v(n,k) =: \frac{v_1(n,k)}{v_2(n,k)}
\]
with $u_1, u_2, v_1, v_2\in C[n, k]$, and $\gcd(u_1,u_2)=\gcd(v_1,v_2)=1$. Then the shift-quotient of our new
hypergeometric term~$\bar{f}(n,k)$ can be computed as follows:
\begin{align*}
  \frac{\bar{f}(n,k+1)}{\bar{f}(n,k)}
  &= \frac{\sum_{i=0}^r p_i(n)\cdot f(n+i,k+1)}{\sum_{i=0}^r p_i(n)\cdot f(n+i,k)} \\
  &= \frac{\sum_{i=0}^r p_i(n) \left(\prod_{j=0}^{i-1}u(n+j,k+1)\right) f(n,k+1)}{
    \sum_{i=0}^r p_i(n) \left( \prod_{j=0}^{i-1} u(n+j,k)\right) f(n,k)} \\
  &= \frac{\sum_{i=0}^r p_i(n)
    \left(\prod_{j=0}^{i-1} u_1(n+j,k+1)\right)
    \left(\prod_{j=i}^{r-1} u_2(n+j,k+1)\right)}{%
    \sum_{i=0}^r p_i(n)
    \left(\prod_{j=0}^{i-1} u_1(n+j,k)\right)
    \left(\prod_{j=i}^{r-1} u_2(n+j,k)\right)} \\
  &\quad \times\frac{\left(\prod_{j=0}^{r-1} u_2(n+j,k)\right) v_1(n,k)}{%
    \left(\prod_{j=0}^{r-1} u_2(n+j,k+1)\right) v_2(n,k)}.
\end{align*}
By denoting the expression on the last line by~$w(k)$, we realize that
the result takes the form
\[
  \frac{c_0(k+1)}{c_0(k)}\cdot w(k)
\]
for some polynomial~$c_0(k)$ that is a linear expression in the
parameters~$p_i$. Also we note that $w(k)$ does not depend
on any of the~$p_i$. Now we compute the Gosper form of~$w(k)$:
\[
  w(k) = \frac{a(k)c_1(k+1)}{b(k)c_1(k)}.
\]
Let $c(k):=c_0(k)\cdot c_1(k)$, then
\[
  \frac{a(k) c(k+1)}{b(k) c(k)}
\]
is a Gosper form for $\bar{f}(n,k)$. The corresponding Gosper equation is
\[
  a(k) z(k+1) - b(k-1) z(k) = c(k),
\]
which satisfies all the claimed properties. We summarize the
steps described above as an algorithm.

\begin{alg}[Zeilberger's algorithm]

\noindent INPUT: a bivariate hypergeometric sequence $f(n,k)$;

\noindent OUTPUT: a telescoper $P\in C(n)\<S_n>$ and a certificate $g(n,k)$
such that the telescopic relation $P(f)=\Delta_k(g)$ holds.
\begin{enumerate}
\item Initialize $r=0$.
\item Introduce undetermined coefficients $p_0,\dots,p_r$.
\item Apply a parametrized variant of Gosper's algorithm to $\sum_{i=0}^r p_i(n)f(n+i,k)$.
\item If a solution $(y, p_0, \dots, p_r) \in C(n,k)\times C(n)^{r+1}$ was found,
  then return $P=\sum_{i=0}^r p_i(n)S_n^i$ and $g(n,k)=y(n,k)\cdot f(n,k)$.
\item If no solution was found, increase $r$ by~$1$ and go back to Step~2.
\end{enumerate}
\end{alg}

Let us discuss some features of Zeilberger's algorithm.
One advantage compared to Sister Celine's method is that we only need
to choose~$r$, the order w.r.t.~$S_n$. We do not need to prescribe~$s$,
the order w.r.t.~$S_k$, because it will be determined automatically by
the algorithm. In practice, Zeilberger's
algorithm is usually more efficient than Sister Celine's method,
and it is guaranteed to find the telescoper of minimal order (which is not
the case for Sister Celine's method). However, note that the
minimal telescoper is not necessarily the minimal-order recurrence
satisfied by the sum.

Another interesting aspect is the so-called WZ phenomenon.
Given an identity $\sum_k \bar{f}(n,k) = h(n)$ with hypergeometric
right-hand side~$h$. Define $f(n,k):=\bar{f}(n,k)/h(n)$, hence
$S(n):=\sum_k f(n,k)=1$. Therefore the sum satisfies $S(n+1)-S(n)=0$.
Now apply Gosper's algorithm to $f(n+1,k)-f(n,k)$.
If it succeeds we receive $g(n,k)$ such that
\[
  f(n+1,k)-f(n,k) = g(n,k+1)-g(n,k)
\]
where $g(n,k)/f(n,k)=:r(n,k)$ is a rational function.
The pair $(f,g)$ is called a \emph{WZ pair}.
The identity $\sum_k \bar{f}(n,k) = h(n)$ is \emph{certified} solely by
the rational function~$r(n,k)$. And as a bonus, one obtains the
\emph{companion identity}
\[
  \sum_{n\geq0} g(n,k) =
  \sum_{j\leq k-1} \!\Bigl(\lim_{n\to\infty} f(n,j) - f(0,j)\Bigr).
\]

We end this section by presenting the Apagodu-Zeilberger algorithm (cf.\ Section~\ref{sec:ct}).
It is motivated by the idea of executing Zeilberger's algorithm once
and for all for a generic input, and then re-use the information
obtained in this way in all future calls of the algorithm with
concrete inputs. In particular, one obtains a bound on the order
of the telescoper and therefore can omit the loop over~$r$. The
Apagodu-Zeilberger algorithm is based on the following theorem,
which is taken from~\cite{mohammed05}, together with its proof.

\begin{thm}
Let $f(n,k)=p(n,k)\cdot h(n,k)$ be a proper hypergeometric term
such that the polynomial $p(n,k)$ is of maximal degree and
\[
  h(n,k) = \frac{%
    \left(\,\prod_{j=1}^A \bigl(\alpha_j\bigr)_{a'_jn+a_jk}\right)
    \left(\,\prod_{j=1}^B \bigl(\beta_j\bigr)_{b'_jn-b_jk}\right)
  }{%
    \left(\,\prod_{j=1}^C \bigl(\gamma_j\bigr)_{c'_jn+c_jk}\right)
    \left(\,\prod_{j=1}^D \bigl(\delta_j\bigr)_{d'_jn-d_jk}\right)
  } \cdot z^k
\]
with $a_j,a'_j,b_j,b'_j,c_j,c'_j,d_j,d'_j\in\set N$.
Furthermore, let
\[
  r := \max\Biggl(\ \sum_{j=1}^A a_j + \sum_{j=1}^D d_j\ , \
  \sum_{j=1}^B b_j + \sum_{j=1}^C c_j\ \Biggr).
\]
Then there exist polynomials $p_0(n),\dots,p_r(n)\in C[n]$, not all zero,
and a rational function $q(n,k)\in C(n,k)$ such that
$g(n,k):=q(n,k)f(n,k)$ satisfies
\[
  \sum_{i=0}^r p_i(n)f(n+i,k) = g(n,k+1)-g(n,k).
\]
\end{thm}
\begin{proof}
Let
\[
  \bar{h}(n,k) = \frac{
    \left(\prod_{j=1}^A (\alpha_j)_{a'_jn+a_jk}\right)
    \left(\prod_{j=1}^B (\beta_j)_{b'_jn-b_jk}\right)}{%
    \left(\prod_{j=1}^C (\gamma_j)_{c'_j(n+r)+c_jk}\right)
    \left(\prod_{j=1}^D (\delta_j)_{d'_j(n+r)-d_jk}\right)}
    \cdot z^k.
\]
Then we obtain for its shift-quotient
\begin{align*}
  \frac{\bar{h}(n,k+1)}{\bar{h}(n,k)} &= \frac{
    \left(\prod_{j=1}^A (\alpha_j+a'_jn+a_jk)_{a_j}\right)
    \left(\prod_{j=1}^D (\delta_j+d'_j(n+r)-d_jk-d_j)_{d_j}\right)}{%
    \left(\prod_{j=1}^B (\beta_j+b'_jn-b_jk-b_j)_{b_j}\right)
    \left(\prod_{j=1}^C (\gamma_j+c'_j(n+r)+c_jk)_{c_j}\right)} \cdot z \\
  & =: \frac{u(n,k)}{v(n,k)}.
\end{align*}
Let $g(n,k)=v(n,k-1)\cdot z(k)\cdot\bar{h}(n,k)$ and plug it into
the telescopic relation:
\[
  \sum_{i=0}^r p_i(n) p(n+i,k) h(n+i,k) =
  v(n,k) z(k+1) \bar{h}(n,k+1) - v(n,k-1) z(k) \bar{h}(n,k)
\]
then, dividing by $\bar{h}(n,k)$ yields
\[
  \underbrace{\sum_{i=0}^r p_i(n) p(n+i,k) \frac{h(n+i,k)}{\bar{h}(n,k)}}_{=:w(n,k)}
  = u(n,k) z(k+1) - v(n,k-1) z(k).
\]
Note that $w(n,k)$ is a polynomial, since
\begin{align*}
  \frac{h(n+i,k)}{\bar{h}(n,k)} &=
  \left(\prod\nolimits_{j=1}^A (\alpha_j+a'_jn+a_jk)_{i a'_j}\right)
  \left(\prod\nolimits_{j=1}^C (\gamma_j+c'_j(n+i)+c_jk)_{(r-i)c'_j}\right) \\
  & \times
  \left(\prod\nolimits_{j=1}^B (\beta_j+b'_jn-b_jk)_{i b'_j}\right)
  \left(\prod\nolimits_{j=1}^D (\delta_j+d'_j(n+i)-d_jk)_{(r-i)d'_j}\right).
\end{align*}
Now make an ansatz for $z(k)=\sum_{i=0}^s z_i k^i$ where
\[
  s := \deg_k(w) - \max\bigl(\deg_k(u),\deg_k(v)\bigr).
\]
Coefficient comparison w.r.t.~$k$ yields
$\deg_k(w)+1$ equations in the $(r+1)+(s+1)$ unknowns.
The condition that the number of unknowns should be greater than the number
of equations yields
\[
  r+s+2 \geq \deg_k(w)+2 \implies
  r \geq \max\bigl(\deg_k(u),\deg_k(v)\bigr).
\]
But note that
\[
  \deg_k(u) = \sum_{j=1}^A a_j + \sum_{j=1}^D d_j
  \quad\text{and}\quad
  \deg_k(v) = \sum_{j=1}^B b_j + \sum_{j=1}^C c_j.
  \qedhere
\]
\end{proof}

\section{D-Finite Functions in One Variable}\label{sec:DfiniteUniv}

\begin{defi}
  A function $f(x)$ is called D-finite (differentiably finite, or holonomic)
  if it satisfies a nontrivial linear ordinary differential equation (LODE)
  with polynomial coefficients, i.e.,
  \[
    p_r(x)f^{(r)}(x) + \dots + p_1(x)f'(x) + p_0(x) f(x) = 0
  \]
  with $p_0,\dots,p_r\in C[x]$ (not all zero).
\end{defi}

\begin{ex}
  The following functions are D-finite (w.r.t.~$x$):
  constant functions, $x^n$, $\sin(x)$, $\exp(x)$, \dots .
  The following functions are not D-finite: $\tan(x)$, $\Gamma(x)$.
\end{ex}

D-finite functions share many nice features: they are described by
finitely many initial conditions, hence by a finite amount of data,
which is desirable for computer implementations. They constitute a
quite rich class of functions, which contains most of the elementary functions,
and many special functions. Last but not least, they satisfy nice closure properties,
as we will demonstrate below.

\begin{thm}
  If $f(x)$ and $g(x)$ are D-finite, then also $f(x)+g(x)=:h(x)$ is D-finite.
\end{thm}
\begin{proof}
  Let $r,s$ denote the orders of $f,g$, respectively. We have to show that
  there exist $u\in\set N$ and $p_0,\dots,p_u\in C[x]$ such that
  \[
    p_u(x)h^{(u)}(x) + \dots + p_0(x) h(x) = 0,
  \]
  which is equivalent to
  \[
    p_u(x)\bigl(f^{(u)}(x)+g^{(u)}(x)\bigr)+\dots+p_0(x)\bigl(f(x)+g(x)\bigr) = 0.
  \]
  Using the LODE for $f$ we have for any $u\in\set N$ that $f^{(u)}(x)$ can be
  written as a linear combination of $f^{(r-1)}(x),\dots,f(x)$. The same is true
  for~$g(x)$ with $r$ replaced by~$s$. Now we get
  \begin{align*}
    & \ell_0(p_0,\dots,p_u)\cdot f(x) + \ell_1(p_0,\dots,p_u)\cdot f'(x) + \dots \\
    & \qquad {} + \ell_{r-1}(p_0,\dots,p_u)\cdot f^{(r-1)}(x) + {} \\
    & \ell_r(p_0,\dots,p_u)\cdot g(x) + \ell_{r+1}(p_0,\dots,p_u)\cdot g'(x) + \dots \\
    & \qquad {} + \ell_{r+s-1}(p_0,\dots,p_u)\cdot g^{(s-1)}(x) = 0,
  \end{align*}
  where each $\ell_j$ is a linear polynomial in the~$p_i$:
  \[
    \ell_j = \sum_{i=0}^u c_{i,j}(x)\cdot p_i(x).
  \]
  The linear system $\ell_0=\dots=\ell_{r+s-1}=0$ has $u+1$ unknowns, and
  hence it is guaranteed to have a nontrivial solution if $u:=r+s$ is chosen.
\end{proof}
Note that there may exist a solution for smaller~$u$, hence it is a good
idea to try. In practice, one loops for $u=1,2,\dots$.
Without proof (because it is similar to the proof of the previous theorem)
we present the next theorem about closure properties.

\begin{thm}\label{thm:closureD}
  If $f(x)$ and $g(x)$ are D-finite then also
  \renewcommand{\labelenumi}{(\roman{enumi})}
  \begin{enumerate}
  \item $f(x)\cdot g(x)$ (the proof is analogous to the previous theorem)
  \item $\int f(x)\,dx$
  \item $f'(x)$ or more generally: any differential polynomial in $f,f',f'',\dots$
  \item $f(a(x))$ where $a(x)$ is an algebraic function, i.e.,
    $p(x,a)=0$ with $p\in C[x,y]$.
  \end{enumerate}
\end{thm}

We use the following operator notation: Let $D_x$ denote the differentiation
w.r.t.~$x$, i.e., $D_x(f(x))=f'(x)$, $D_x^2(f(x))=f''(x)$, etc., and also
$D_x^0(f(x))=f(x)$.
Let $C(x)\langle D_x\rangle$ denote the polynomial ring in $D_x$ with
coefficients in~$C(x)$. This ring is non-commutative as $D_x\cdot x = x\cdot D_x+1$
(Leibniz rule), or more generally: $D_x\cdot r(x)=r(x)\cdot D_x+r'(x)$ for any
$r\in C(x)$. Using this notation, $f(x)$ is D-finite if and only if
\[
  \exists L\in C(x)\langle D_x\rangle\setminus\{0\}\colon
  L(f(x))=0.
\]
Some of the closure properties stated in Theorem~\ref{thm:closureD} can
be conveniently expressed and proven using the operator viewpoint.
Let $L,M\in C(x)\langle D_x\rangle\setminus\{0\}$ annihilate $f,g$,
respectively, i.e., $L(f(x))=0$ and $M(g(x))=0$. Then:
\begin{itemize}
\item $L\cdot D_x$ annihilates $\int f(x)\,dx$.
\item The least common left multiple of $L$ and~$M$, that is the lowest-order
  operator that equals $U\cdot L$ and $V\cdot M$ for certain
  $U,V\in C(x)\langle D_x\rangle$, annihilates $f+g$ (actually:
  any linear combination $c_1f+c_2g$ for arbitrary constants $c_1,c_2$).
\item If $h$ satisfies $L(h)=g$, then $h$ is also D-finite,
  because $(M\cdot L)(h)=0$.
\end{itemize}

The following theorem is actually a corollary of
Theorem~\ref{thm:closureD}(iv) with $f(x)=x$,
but we find it enlightening to state and prove it separately.
\begin{thm}
  Let $f(x)$ be an algebraic function. Then $f$ is D-finite.
\end{thm}
\begin{proof}
  Let $m\in C[x,y]$ be the minimal polynomial of~$f$, that is
  $m(x,f(x))=0$ and $m$ is irreducible. Then we have
  \[
    a_d(x)(f(x))^d + a_{d-1}(x)(f(x))^{d-1} + \dots + a_1(x)f(x) + a_0(x) = 0.
  \]
  Hence
  \[
    f' = \frac{-(a'_df^d + \dots +a'_1f+a'_0)}{d\cdot a_d f^{d-1} + \dots + a_1}
    =: \frac{q(x,f)}{r(x,f)}.
  \]
  Note that $\gcd(m,r)=1$, hence by the extended Euclidean algorithm there exist
  $u,v\in C(x)[f]$ with $u\cdot m+v\cdot r=1$. We get that $v(x,f)\cdot r(x,f)=1$
  modulo $m(x,f)$, hence
  \[
    f' = \frac{q(x,f)\cdot v(x,f)}{r(x,f)\cdot v(x,f)} = q(x,f)\cdot v(x,f) =
    c_{1,d-1} f^{d-1} + \dots + c_{1,1}f + c_{1,0}
  \]
  modulo~$m$ for some $c_{1,j}\in C(x)$.
  We see that the derivative $f'$ can be expressed as a $C(x)$-linear
  combination of the powers $f^0,f^1,\dots,f^{d-1}$. Now we differentiate
  again and perform a similar rewriting
  \[
    f'' = c_{2,d-1}f^{d-1}+\dots+c_{2,1}f+c_{2,0},
  \]
  which shows that also $f''$ can be expressed as such a $C(x)$-linear
  combination. Summarizing, we learn that all derivatives of~$f$ live
  in a $C(x)$-vector space of dimension at most~$d$, and therefore
  $f,f',f'',\dots,f^{(d)}$ will be $C(x)$-linearly dependent.
  It follows that $f$ satisfies an LODE with polynomial coefficients
  of order at most~$d$.
\end{proof}

We now turn our attention to univariate sequences, which satisfy
analogous properties as univariate D-finite functions.
\begin{defi}
  A sequence $(a_n)_{n\in\set N}$ is called P-recursive (or holonomic,
  or D-finite) if it satisfies a nontrivial linear ordinary
  recurrence equation (LORE) with polynomial coefficients, i.e.,
  \[
    p_r(n)\cdot a_{n+r} + \dots + p_1(n)\cdot a_{n+1} + p_0(n)\cdot a_n =0
  \]
  with $p_0,\dots,p_r\in C[n]$ (not all zero).
\end{defi}

\begin{ex}
  Examples for P-recursive sequences are the Fibonacci numbers,
  polynomial sequences,
  hypergeometric terms, harmonic numbers, orthogonal polynomials.
  Not P-recursive is for example the sequence of prime numbers.
\end{ex}

P-recursive sequences share many nice features: they are described by
finitely many initial conditions, hence by a finite amount of data,
which is desirable for computer implementations. They constitute a
quite rich class of sequences, and last but not least, they satisfy
nice closure properties, as we will demonstrate below.

\begin{thm}\label{thm:closureS}
  If $a_n$ and $b_n$ are P-recursive then also the following are P-recursive:
  \renewcommand{\labelenumi}{(\roman{enumi})}
  \begin{enumerate}
  \item $a_n+b_n$
  \item $a_n \cdot b_n$
  \item $a_{cn+d}$ for integers $c,d\in\set Z$
  \item $\sum_n a_n$, the indefinite sum, i.e., $c_n$ such that $c_{n+1}-c_n=a_n$
  \end{enumerate}
\end{thm}

Again, we employ operator notation:
$S_n$ denotes the forward shift operator, i.e.,
$S_n(a_n)=a_{n+1}$ and $S_n^2(a_n)=a_{n+2}$, etc.
We denote by $C(n)\langle S_n\rangle$ the ring of all linear recurrence
operators with coefficients in $C(n)$. We observe the following commutation
rule: $S_n\cdot n=(n+1)\cdot S_n$, or more generally: $S_n\cdot r(n)=r(n+1)\cdot S_n$
for any $r(n)\in C(n)$. The properties of linear recurrence operators are
analogous to the ones stated for linear differential operators.

D-finite functions and P-recursive sequences are closely related to each other,
as the next theorem shows. In fact, this is also the reason why, by abuse of
wording, we call P-recursive sequences sometimes also ``D-finite sequences''.

\begin{thm} A power series $f(x)$ is D-finite if and only if the sequence
  of its Taylor coefficients is P-recursive. Stated differently:
  A sequence is P-recursive if and only if its generating function is D-finite.
\end{thm}
\begin{proof}
  Let $f(x)=\sum_{n=0}^{\infty} a_n x^n$. Then
  $f'(x) = \sum_{n=1}^{\infty} n\cdot a_n\cdot x^{n-1}$
  and, more generally for any $i\in\set N$, we get
  \[
    f^{(i)}(x)=\sum_{n=i}^{\infty} (n-i+1)_i a_n x^{n-i} =
    \sum_{n=0}^{\infty} (n+1)_i a_{n+i} x^n.
    \]
  Now we assume that $f$ satisfies the LODE
  \[
    \sum_{i=0}^r \sum_{j=0}^d p_{i,j} \, x^j f^{(i)}(x) = 0,
  \]
  which can equivalently be written as
  \[
    \sum_{i=0}^r \sum_{j=0}^d \sum_{n=0}^{\infty}
    p_{i,j} \, (n+1)_i \, a_{n+i} \, x^{n+j} = 0
  \]
  or
  \[
    \sum_{i=0}^r \sum_{j=0}^d \sum_{n=j}^{\infty}
    p_{i,j} \, (n-j+1)_i \, a_{n-j+i} \, x^n = 0.
  \]
  For $n\geq d$ we obtain the desired LORE:
  \begin{equation}\label{eq:de2re}
    \sum_{i=0}^r \sum_{j=0}^d p_{i,j} \, (n-j+1)_i \, a_{n-j+i} = 0.
  \end{equation}
  Conversely, if $a_n$ is a P-recursive sequence, then write its recurrence
  in the form~\eqref{eq:de2re} and perform the above calculations backwards,
  to conclude that its generating function will be D-finite.
\end{proof}

\section{D-Finite Functions in Several Variables}\label{sec:DfiniteMult}

D-finite functions in several variables are systematically studied by Lipshitz~\cite{Lipshitz1988, Lipshitz1989}. We will recall some basic closure properties
of these functions. Let $F = K(x_1, \ldots, x_n)$ be the field of rational functions in variables $x_1, \ldots, x_n$ and $\delta_{x_i}$ denote 
the usual partial derivation $\frac{\partial}{\partial x_i}$ with respect to $x_i$. Let
$\mathcal{D} := F \langle D_{x_1}, \ldots, D_{x_n} \rangle$ be the ring of linear differential operators over $F$ in which we have the commutation rules: 
\[D_{x_i} \cdot f = f \cdot D_{x_i} + \delta_i(f)\quad \text{for any $f \in F$ and $i\in \{1, \ldots, n\}$};\]
\[D_{x_i}\cdot D_{x_j} = D_{x_j}\cdot D_{x_i} \quad \text{for any $i, j$ with $1 \leq i \leq j \leq n$}.\]
Let $M = K[[x_1,\ldots, x_n ]]$ be the ring of formal power series in variables $x_1, \ldots, x_n$ and~$L \in \mathcal{D}$. Write
\[
    L = \sum_{i_1,\ldots,i_n} \ell_{i_1,\ldots,i_n} D_{x_1}^{i_1}\cdots D_{x_n}^{i_n}, 
\]
and define the action of $L $ on $f\in M$ by
\[
   L \cdot f = \sum_{i_1,\ldots,i_n}  \ell_{i_1,\ldots,i_n} \frac{\partial^{i_1}}{\partial x_1^{i_1}} \cdots  \frac{\partial^{i_n}}{\partial x_n^{i_n}}(f).
\]
This turns $M$ into a~$\mathcal{D}$-module. For $f \in K[[x_1,\ldots,x_n]]$, define the annihilating ideal $I_f$ of $f$ in~$\mathcal{D}$
as 
\[I_f = \left\lbrace L \in \mathcal{D} \ \bigg| \  L \cdot f = 0 \right\rbrace.\]
Note that $I_f$ is also a subspace of~$\mathcal{D}$ as vector spaces over $F$. In general, the dimension of the quotient $\mathcal{D}/I_f$
is infinite. 

\begin{defi}~A power series $f \in K[[x_1,\ldots,x_n]]$ is said to be \emph{D-finite} over $F$ if 
    \[
    \dim_F\big( \mathcal{D}/I_f \big) < + \infty. 
    \]
    A sequence $T \colon \mathbb{N}^n \rightarrow K$ is said to be \emph{D-finite} if the generating function
    \[
    f(x_1, \ldots, x_n) = \sum_{i_1,\ldots,i_n} T(i_1,\ldots, i_n) x_1^{i_1}\cdots x_n^{i_n}
    \]
    is D-finite over~$F$.
\end{defi}
\begin{lem}
    A series $f \in K[[x_1,\ldots, x_n]]$ is D-finite if and only if  
    \[
    I_f \cap F\<D_{x_i}> \neq \{0\}~\text{for each}~i \in \{1,\ldots,n \}.
    \]
\end{lem}
\begin{proof}
    For the necessity, if $f$ is D-finite then~$ d = \dim_F(\mathcal{D} \cdot f) < +\infty$.
    Thus for any~$i = 1,2, \ldots, n$, the elements
    \[
        f, \,  D_{x_i}\cdot f,\,  D_{x_i}^2 \cdot f,\quad  \cdots, \quad D_{x_i}^d \cdot f
    \]
    are linearly dependent over~$F$, which implies that~$I_f \cap F\<D_{x_i}> \neq \{0\}$. 
    
    For the sufficiency, if~$I_f \cap F\<D_{x_i}> \neq \{0\}$ for each $i$ then there exists~$L_i \in I_f \setminus \{0\}$ such that~$L_i = \sum_{j=0}^{d_i} \ell_{i,j} D_{x_i}^j$.
    Hence we have that $D_{x_1}^{i_1}\cdots D_{x_n}^{i_n} \cdot f $ can be rewritten into a $F$-linear combination of terms
    \[
        D_{x_1}^{j_1}\cdots D_{x_n}^{j_n} \cdot f~\text{ with}~0 \leq j_s < d_i .
    \]
    It follows that~$\dim_F(\mathcal{D}\cdot f) < + \infty$.
\end{proof}
The following theorem summarizes some closure properties of D-finite functions. 
\begin{thm}
    Let $f,g \in K[[x_1,\ldots,x_n]]$ be D-finite over~$F$. Then
    \begin{enumerate}
        \item $f + g$ is D-finite.
        \item $L \cdot f$ is D-finite for any~$L \in \mathcal{D}$.
        \item $f \cdot g$ is D-finite.
        \item If $\alpha_1,\ldots,\alpha_n \in K[[y_1,\ldots,y_m]]$ are algebraic over $K(y_1,\ldots,y_m)$, and $f(\alpha_1,\ldots,\alpha_n)$ 
        is well-defined, then $f(\alpha_1,\ldots, \alpha_n)$ is D-finite over~$K(y_1,\ldots,y_n)$.
    \end{enumerate}
  \end{thm}
  \begin{proof}
        \begin{enumerate}
            \item Since $\dim_F\bigl(\mathcal{D}\cdot f\bigr) < +\infty$ and $\dim_F\bigl(\mathcal{D}\cdot g\bigr) < +\infty$, we have that 
            \[
                \dim_F(\mathcal{D}\cdot f + \mathcal{D}\cdot g) < + \infty.
            \]
            By observing the fact that $\mathcal{D}\cdot(f+g) \subseteq \mathcal{D}\cdot f + \mathcal{D}\cdot g$, we obtain that 
            \[
                \dim_F\bigl( \mathcal{D}\cdot (f+g)\bigr) < + \infty.
            \]
            \item Since $L \cdot f \in \mathcal{D}\cdot f$ and $P \cdot (L \cdot f) \in \mathcal{D}\cdot f$ for any~$P \in \mathcal{D}$, we have that~$\mathcal{D}\cdot (L \cdot f) \subseteq \mathcal{D}\cdot f $.
            Hence 
            \[
                \dim_F\bigl( \mathcal{D}\cdot (L \cdot f)\bigr) \leq \dim_F\bigl( \mathcal{D}\cdot f \bigr) < +\infty.
            \]
            \item   
            Since $I_f \cap F\<D_{x_i}> \neq \{0\} $ and $I_g \cap F\<D_{x_i}> \neq \{0\} $, by the closure property in the univariate case (Theorem~\ref{thm:closureD} in Section~\ref{sec:DfiniteUniv}), we get~$I_{fg} \cap F\<D_{x_i}> \neq \{ 0\}$, which implies that $f\cdot g$ is D-finite.
            \item We leave the proof of this property as an exercise.
        \end{enumerate}
    \end{proof}
    \begin{defi}
        Let $\mathcal{A} := K[x_1,\ldots,x_n]\langle D_{x_1},\ldots,D_{x_n}\rangle$ and $f$ be an element of an $\mathcal{A}$-module~$M$. 
         A left ideal $J$ of~$\mathcal{A}$ is said to be \emph{holonomic} if for every subset $U \subseteq \{ x_1,\ldots,x_n, D_{x_1},\ldots,D_{x_n}\}$ with $|U| =n+1$, we have~$J \cap K[U] \neq \{0\}$. The element $f$ is said to be holonomic if the annihilating  ideal $J_f = \{L\in \mathcal{A} \mid L\cdot f =0\}$ is holonomic. 
    \end{defi}
    The following theorem says that the D-finiteness is equivalent to the holonomicity in the differential setting.
    \begin{thm}\label{thm:holiffdf}
        Let $f$ be an element of a $\mathcal{D}$-module which can also be viewed as an $\mathcal{A}$-module. 
        Then $f$ is holonomic if and only if $f$ is D-finite. 
    \end{thm}
    \begin{proof}
        For the necessity, we assume that $f$ is holonomic. Then for every $i = 1,\ldots,n$, let 
        \[
            U_i = \{x_1,\ldots,x_n, D_{x_i} \}, \quad |U_i|=n+1.
        \]
         we have that $J_f \cap K[U_i] \neq \{0\}$, which implies that 
        \[
            I_f \cap F\<D_{x_i}> \neq \{0\},\quad \text{for every } i=1,\ldots,n.
        \]
        Hence $f$ is D-finite over $F$.
        
   For the sufficiency, we assume that $f$ is D-finite. Then $r = \dim_F(\mathcal{D}\cdot f) < +\infty$.
        Let~$\{b_1=f,\ldots,b_r\}$ be a basis of the vector space of~$\mathcal{D}\cdot f$. Then for any $g \in \mathcal{D}\cdot f$, 
        \[
            g = g_1 b_1 + \cdots + g_r b_r = \vec{g}\cdot \vec{b}^T \quad \text{with} \quad g_i \in F. 
        \]
        We write for $i=1,\ldots,n$,
        \[
            D_{x_i}\cdot g = A_i \cdot \vec{g} \cdot \vec{b}^T + \delta_i(\vec{g})\cdot \vec{b}^T,
        \]
        where $A_i \in F^{r \times r}$.
        Let $q \in K[x_1,\ldots,x_n]$ be the common denominator of all entries of~$A_i \ (i=1,\ldots,n)$, and $d \geq 1 $ be such that the total degree of $q$ as well as the entries of $qA_i \ (i=1,\ldots,n)$ are less than~$d$. Note that 
        \[
            f = (1,0,\ldots,0)
            \begin{pmatrix}
                b_1\\
                \vdots\\
                b_r
            \end{pmatrix}.
        \]
        For every $k \in \bN$, we have that if $i_1+\cdots + i_n+j_1 + \cdots + j_n \leq k $ then 
        { 
        \begin{equation*}
            x_1^{i_1}\cdots x_n^{i_n}D_{x_1}^{j_1}\cdots D_{x_n}^{j_n} \cdot f \in V_k,
        \end{equation*}
        }
        where 
        \[V_k := 
            \Bigg\lbrace \frac{1}{q^k}(p_1b_1 + \cdots + p_rb_r)\ \bigg| \  p_i \text{ are polynomials of total degree of } kd\Bigg\rbrace.\]
        Let $W_{U,k}$ be the vector space over $K$ generated by monomials of variables in $U \subseteq \{x_1,\ldots,x_n,D_{x_1},\ldots,D_{x_n} \}$ with~$|U|=n+1$.
        Let 
        \begin{align*}
            \phi_k \colon W_{U,k} &\longrightarrow V_k \\
                            L     & \longmapsto L \cdot f
        \end{align*} be a $K$-linear map. Then we have that 
        \[
            \dim_K(W_{U,k}) = \binom{n+1+k}{k} = \binom{n+1+k}{n+1} \sim O(k^{n+1})
        \]
        and 
        \[
            \dim_{K}(V_k) \leq r \cdot \binom{n+kd}{kd} = r \cdot \binom{n+kd}{n} \sim O(k^n).
        \]
        Then for large enough~$k$, we have that $\ker(\phi_k)$ is nontrivial, which implies that 
        \[
            J_f \cap K[U] \neq \{0\}.
        \]
        Then $f$ is holonomic. 
    \end{proof}
    
    The following corollary will be used to guarantee the existence of telescopers for D-finite functions in two variables.
    \begin{cor}\label{COR:ct}
        Let $f(x,y)$ be D-finite over~$K(x,y)$. Then there exist operators $P \in K(x)\langle D_x\rangle \setminus\{0\}$ and $Q \in K(x,y)
        \langle D_x,D_y\rangle$ such that 
        \[P \cdot f = D_y\bigl(Q\cdot f\bigr).\]
    \end{cor}
    \begin{proof}
        Since $f$ is D-finite, there exists  $L \in K[x]\<D_x, D_y>$ such that $L\cdot f =0$. Now we write 
        \[
            L = D_y^m \bigl( \bar{P}(x, D_x) + D_y \bar{Q} \bigr).
        \]
        By Wegschaider's trick, $y^m D_y^m = D_y R + m!$ for some~$R \in K[y]$. Then we have 
        \[
            0 = y^m L \cdot f = \bigl( D_y R + m! \bigr) (\bar{P}(x, D_x) + D_y \bar{Q})\cdot f = \bigl( P(x, D_x) + D_y Q \bigr) \cdot f.
        \]
        This completes the proof.
    \end{proof}

\section{Advanced Closure Properties}\label{sec:AdvClosureP}

In the previous section, we have seen that a function $f(x,y)$ is called \emph{D-finite}
if the left ideal $I$ consisting of all operators $L\in C(x,y)\<D_x,D_y>$ mapping $f$ to
zero satisfies
\[
  \dim_{C(x,y)}C(x,y)\<D_x,D_y>/I<\infty.
\]
Moreover, the function is called \emph{holonomic} if the left ideal $J$ consisting of
all operators $L\in C[x,y]\<D_x,D_y>$ mapping $f$ to zero satisfies
\[
J\cap C[U]\neq\{0\}
\]
for every $U\subseteq\{x,y,D_x,D_y\}$ with $|U|>2$. An ideal $J$ is called \emph{holonomic}
if it has this property. 

The notions extend to functions in more variables, and to shift and other operators
besides derivations. 

D-finiteness and holonomy are preserved by addition, multiplication, and a few other
operations. This means, for example, that when $f$ and $g$ are D-finite or holonomic,
then so are $f+g$ and~$fg$. These \emph{closure properties} are based on linear algebra.
In the present section, we will discuss more advanced closure properties based on
creative telescoping.

Recall that creative telescoping refers to the search for annihilating operators of the
form $P-D_xQ$ or $P-\Delta_xQ$ where $P$ (the telescoper) must be nonzero and free of $x$
while $Q$ (the certificate) may or may not involve $y$ and may or may not be nonzero.
So far we have discussed creative telescoping algorithms for rational functions and (proper)
hypergeometric terms. By the following theorem, the concept of creative telescoping
applies in a more general setting.

\begin{thm}\label{thm:holoct}
 If $J\subseteq C[x,y]\<\partial_x,\partial_y>$ is holonomic, then there exist
 $P\in C[x]\<\partial_x>\setminus\{0\}$ and $Q\in C[x,y]\<\partial_x,\partial_y>$
 such that $P-\partial_yQ\in J$.
 Moreover, given a basis of $J$, such $P$ and $Q$ can be computed.
\end{thm}

Rather than entering into why this is true, we shall discuss the implications of this
theorem. Recall that telescopers were introduced in order to deal with definite sums
and integrals: if things go smoothly, every telescoper for a summand/integrand is an
annihilating operator for the sum/integral. For things to go smoothly means the following.

\begin{defi}
 A sum/integral over a holonomic function $f$ is said to have \emph{natural boundaries}
 if there is a telescoper/certificate pair $(P,Q)$ such that $Q\cdot f$ evaluates to
 zero at the boundaries of the summation/integration range.
\end{defi}

For example, note that for every operator $Q\in C[n,k]\<S_n,S_k>$, we have $[Q\cdot\binom nk]_{k=-\infty}^{\infty}=0$,
because $\binom nk=0$ whenever $k<0$ or~$k>n$. With this terminology, we can say that
the class of holonomic functions with natural boundaries is closed under definite
summation/integration. What about non-natural boundaries?

\begin{ex}
  Consider the sum $S(n)=\sum_{k=0}^n\binom{2n}k$. For the summand, we have the creative
  telescoping relation
  \[
    (S_n-4)\cdot\binom{2n}k=\Delta_k\frac{k(2k-6n-5)}{2(2n+1)(n+1)}\binom{2n+2}k.
  \]
  Summing this equation over $k=0,\dots,n$ gives
  \[
    \sum_{k=0}^n\binom{2(n+1)}k - 4\sum_{k=0}^n\binom{2n}k
    =\Bigl[\frac{k(2k-6n-5)}{2(2n+1)(n+1)}\binom{2n+2}k\Bigr]_{k=0}^{n+1}.
  \]
  The right-hand side evaluates to $-\frac{4n+3}{4n+2}\binom{2n+2}{n+1}$, and the first sum
  on the left is recognized as $S(n+1)-\binom{2(n+1)}{n+1}$. The second sum is $S(n)$, so
  if we move $\binom{2(n+1)}{n+1}$ to the right-hand side and simplify a bit, we obtain
  the recurrence
  \[
    S(n+1) - 4 S(n) = -\frac1{4n+2}\binom{2n+2}{n+1}
  \]
  for the sum. We see that non-natural boundaries lead to inhomogeneous recurrences. In the
  present example, the right-hand side is annihilated by $(n+2)S_n-(4n+2)$, so if we apply
  this operator to the equation above, we obtain a homogeneous recurrence for the sum:
  \[
   (n+2)S(n+2)-(8n+10)S(n+1)+(16n+8)S(n)=0.
  \]
  In particular, we see that the sum is holonomic.
\end{ex}

In general, a sum/integral with non-natural boundaries leads to an inhomogeneous recurrence
or differential equation where the right-hand side involves evaluations of the summand/integrand:
$P\cdot f=[Q\cdot f]_\Omega$. If the right-hand side is annihilated by an operator~$L$, then
the operator $LP$ annihilates the definite sum/integral. The question is therefore whether
holonomy is preserved under evaluation.

To answer this question, observe that Thm.~\ref{thm:holoct} remains true if we exchange the
roles of $y$ and~$\partial_y$. In other words, in a holonomic ideal $J\subseteq C[x,y]\<\partial_x,\partial_y>$
we also always have operators $P\in C[x]\<\partial_x>\setminus\{0\}$ and $Q\in C[x,y]\<\partial_x,\partial_y>$
such that $P-yQ\in J$. This implies that if $f(x,y)$ is holonomic, so is $f(x,0)$, for if
we set $y=0$ in the equation $P\cdot f(x,y)=yQ\cdot f(x,y)$, then we obtain $P\cdot f(x,0)=0$.
In short, holonomy is preserved under evaluation.

By the closure of holonomy under evaluation, we can considerably relax the requirement of having
natural boundaries. Integration ranges can be any semi-algebraic set, and summation ranges can
be any rational polygonal sets. More precise statements can be found in Sect.~5.3 of~\cite{Kauers2023}.

\medskip
For proving closure properties, the concept of holonomy is more handy than the concept of D-finiteness,
because the latter may suffer from trouble with singularities, and, in the summation case, with the
lack of an existence guarantee for telescopers. On the other hand, algorithms based on holonomy tend
to be much more expensive than algorithms based on D-finiteness. It is therefore desirable to combine
the best of the two worlds. According to Theorem~\ref{thm:holiffdf} there is a tight
connection, at least in the differential case:
\begin{itemize}
\item If an ideal $I\subseteq C(x,y)\<D_x,D_y>$ is D-finite, then the ideal $I\cap C[x,y]\<D_x,D_y>$ is
  holonomic.
\item If an ideal $J\subseteq C[x,y]\<D_x,D_y>$ is holonomic, then the ideal $\<J>$ it generates in
  $C(x,y)\<D_x,D_y>$ is D-finite. 
\end{itemize}
In particular, telescoper/certificate pairs always exist for D-finite ideals in the differential case.

\def\res{\operatorname{res}}
For a formal infinite series $f(x,y)=\sum_{n,k\in\set Z}a_{n,k}x^ny^k$, define the \emph{residue}
$\res_y f(x,y)$ as the series $\sum_{n\in\set Z} a_{n,-1}x^n$. For every series $g(x,y)$ we then
have $\res_y D_y\cdot g(x,y)=0$, because differentiation cannot produce terms with exponent~$-1$.
Therefore, if $P$ and $Q$ are such that $(P-D_yQ)\cdot f=0$ and $P$ is free of~$y$, then
$P\cdot\res_y f=0$. Note that in this formal setting, there is no trouble with singularities or
cumbersome boundary conditions.

However, there is a little algebraic issue: what do we mean by $Q\cdot f$ if $Q$ is an element of
$C(x,y)\<D_x,D_y>$ and $f$ is a bilateral infinite series? More precisely, while it is clear that
the bilateral infinite series form a $C[x,y]\<D_x,D_y>$-module, what shall it mean to multiply such a
series with a rational function? In order to make this precise, we need an interpretation of
rational functions as infinite series, and a way to multiply infinite series.

Without any restriction on the support of the series, multiplication is not well defined. That's
why in the univariate case, formal Laurent series are defined as infinite series of the form
$\sum_{n=n_0}^\infty a_nx^n$, i.e., series having a minimal exponent. These series form a
field denoted by $C((x))$, the quotient field of the formal power series ring $C[[x]]$.

The construction of formal Laurent series in the case of several variables is a bit more subtle.
We do not need to go into the details here, they are explained in~\cite{aparicio12}.
All that matters is that there is a way to construct fields $C((x,y))$ whose elements are infinite
series with certain restrictions on their supports, and that these fields have a natural
$C(x,y)\<D_x,D_y>$-module structure. For elements of such fields, it is therefore meaningful to
talk about D-finiteness. In particular, we can say that for every D-finite element of a field $C((x,y))$,
the residue (w.r.t. $y$, say) is a D-finite element of~$C((x))$. 

\begin{ex}
  In an earlier example, we have found a telescoper $P$ for the rational function
  $f=\frac1{xy^3+y+1}$. If $\tilde f$ is any expansion of $f$ as a bivariate infinite series,
  then $P\cdot\res_y\tilde f=0$.
\end{ex}

\def\diag{\operatorname{diag}}
The fact that D-finiteness is preserved under taking residues implies a couple of other useful
closure properties. For example, the \emph{diagonal} of a series $a(x,y)=\sum_{n,k}a_{n,k}x^ny^k$ is
defined as
\[
  \diag a(x,y) := \sum_n a_{n,n}x^n.
\]
Since the diagonal can be expressed as a residue via
\[
  \diag a(x,y) = \res_y y^{-1} a(y,x/y),
\]
it is clear that D-finiteness is also preserved under taking diagonals.

Next, the \emph{Hadamard product} of infinite series $a(x,y)=\sum_{n,k}a_{n,k}x^ny^k$ and $b(x,y)=\sum_{n,k}b_{n,k}x^ny^k$
is defined by the termwise product of the coefficients:
\[
  a(x,y)\odot b(x,y) := \sum_{n,k} a_{n,k}b_{n,k}x^ny^k.
\]
Since the Hadamard product can be expressed by residues via
\[
  a(x,y)\odot b(x,y) = \res_{x'}\res_{y'} \frac1{x'y'} a(x',y')b(x/x',y/y'),
\]
where $x'$ and $y'$ are two new variables, it is clear that D-finiteness is also preserved under taking Hadamard products.

Finally, the \emph{positive part} of an infinite series $a(x,y)=\sum_{n,k}a_{n,k}x^ny^k$ (which may involve some terms
with negative exponents) is defined as
\[
  [x^>y^>]a(x,y) := \sum_{n,k>0}a_{n,k}x^ny^k.
\]
Since the positive part can be expressed as a Hadamard product via
\[
  [x^>y^>]a(x,y) = \frac x{1-x}\frac y{1-y} \odot a(x,y)
\]
and D-finiteness is preserved under Hadamard products, it is clear that D-finiteness is also preserved under taking positive parts.

These relationships are worked out in more detail in~\cite{bostan16b}, where they have been used to compute annihilating operators
arising in the context of lattice walk counting.

\medskip
The relation between holonomy and D-finiteness is more complicated in the discrete case.
In this case, not every D-finite function admits a telescoper, and even if it does, we may have trouble with singularities.
\begin{ex}
\begin{enumerate}
\item For $f(n,k)=\binom nk^2$, we have the telescoping relation
 \[
   \biggl((n+1)S_n - 2(2n+1)-\Delta_k\frac{k^2(2k-3n-3)}{(n-1-k)^2}\biggr)\cdot f(n,k)=0.
 \]
 In view of the factor $n-1-k$ in the denominator of the certificate, it is dangerous to sum this relation
 over all~$k$.
\item $\frac1{n^2+k^2}$ is an example for a hypergeometric term which does not have a telescoper.
\end{enumerate}
\end{ex}

In order to escape from these problems, it has been proposed in~\cite{bostan17a} to use residues of infinite series
also for solving summation problems.
The idea is to first use some rewriting rules in order to translate a given expression involving sums and binomial
coefficients into a multivariate rational function.
During this translation, definite sums are made indefinite by introducing new auxiliary variables as needed.
In the end, the auxiliary variables are identified with the original variables by a residue computation.

\begin{ex}
  Consider the sum $S(n)=\sum_{k=0}^n\binom nk$.
  Its innermost expression is the binomial coefficient, whose generating function is
  \[
    \sum_{n,k=0}^\infty\binom nkx^ny^k=\frac1{1-(1+y)x}.
  \]
  For applying the sum, replace the upper summation bound by a new variable~$m$, unrelated to the upper parameter $n$ in the
  binomial coefficient. The generating function is then
  \[
    \sum_{n,m=0}^\infty\biggl(\sum_{k=0}^m\binom nk\biggr)x^nz^m=\frac1{1-z}\frac1{1-(1+z)x}.
  \]
  To finally identify $n$ and $m$, we take the diagonal of this series and find
  \[
    \diag\frac1{1-z}\frac1{1-(1+z)x}=\frac1{1-2x}.
  \]
  As the latter is the generating function of~$2^n$, we have found the expected closed form for~$S(n)$.
\end{ex}

This approach applies to a certain class of summation expressions that the authors of~\cite{bostan17a} call
\emph{binomial sums.} This class does not contain every (proper) hypergeometric term.
On the other hand, unlike the creative telescoping algorithms discussed in earlier sections,
it also covers nested sums, it does not suffer from trouble with singularities, and there
is a guarantee that every binomial sum is D-finite.


\section{Abramov--van Hoeij's Algorithm and Chyzak's Algorithm}\label{sec:Chyzak}


The problems of symbolic summation and creative telescoping for hypergeometric terms (the first-order case) have been studied in Sections~\ref{sec:Gosper} and \ref{sec:Zeilberger}. 
In this section, we shall explain algorithms for symbolic integration and creative telescoping for D-finite functions (the high-order case), namely Abramov--van Hoeij's algorithm~\cite{AbramovHoeij1997,AbramovHoeij1999} and Chyzak's algorithm~\cite{Chyzak2000}, respectively.

The integration problem on D-finite functions is as follows. 
\begin{pro}\label{PB:dfinite}
    Given a D-finite function $f(y)$ of order~$n$, decide whether there exists another D-finite function $g(y)$ of the same order~$n$ such that~$f = D_y(g)$.
\end{pro}

Let $\mathcal{D}_y$ be the ring $C(y)\langle D_y\rangle$ of linear differential operators in $y$ over $C(y)$
and $f(y)$ be a D-finite function over $C(y)$. Since $\mathcal{D}_y$ is a left Euclidean domain, the annihilating ideal $I_f := \{P\in \mathcal{D}_y \mid P(f)=0 \}$
is uniquely generated by one monic operator $L$, which is called the \emph{minimal annihilator} of $f$. For any operator~$L = \sum_{i=0}^r a_i D_y^i \in \mathcal{D}_y$,
we call $L^{*} = \sum_{i=0}^r (-1)^i D_y^i a_i$ the \emph{adjoint operator} of~$L$. From the definition, we have $(L^{*})^{*} =L$ and the following Lagrange's identity: 
For functions $u,v$ in some $\mathcal{D}_y$-module, we have 
        \[
            uL(v)-vL^{*}(u) = D_y\bigl(M(u,v)\bigr),
        \]
        where $M$ is a polynomial of successive derivatives of $u$ and $v$ with order at most~$\ord(L)$. 
        This identity is a high-order extension of the Leibniz rule:
        \[
            u D_y(v)+ v D_y(u) = D_y(u v) \quad \text{with  $L =D_y$ and $L^{*}=-D_y$}.
        \]
    In particular, if $L^{*}(r)=1$ and~$L(f)=0$, then Lagrange's identity implies that $f = D_y(T\cdot f)$ for some $T \in \mathcal{D}_y$ with~$\ord(T) < \ord(L)$.
   
Problem~\ref{PB:dfinite} can be solved by the following algorithm. 
\begin{alg}[Abramov--van Hoeij's Algorithm]
    \qquad
    
    \noindent INPUT: A D-finite function $f$ defined by the minimal operator $L \in \mathcal{D}_y$ with~$\ord(L) = n$.

    \noindent OUTPUT: $g= T(f)$ with $T \in \mathcal{D}_y$ such that $f = D_y(g)$, otherwise return \underline{No}.
    \begin{enumerate}
        \item Compute the adjoint operator~$L^{*}$ of $L$.
        \item Decide whether there exists $r\in C(y)$  such that $L^{*}(r)=1$. If no such $r$ exists, return \underline{No}, 
        otherwise return  $T(f)$, where $T$ satisfies $r L + D_y T =1$, computed by the extended left Euclidean algorithm in $\mathcal{D}_y$.
    \end{enumerate}
\end{alg}

The next theorem proves the correctness of the above algorithm. 
\begin{thm}
    Let $f(y)$ be a D-finite function with minimal annihilator $L\in \mathcal{D}_y$ of order~$n$. Then the following conditions are equivalent:
    \begin{itemize}
        \item[(1)] $f= D_y(g)$ for some D-finite function $g$ of the same order~$n$.
        \item[(2)] $f = D_y(T(f))$ for some $T \in \mathcal{D}_y$ of order~$\leq n-1$.
        \item[(3)] There exists $r\in C(y)$  such that $L^{*}(r)=1$.
    \end{itemize}
\end{thm}
\begin{proof}
    $(1) \Rightarrow (2)$: Let $P$ be the minimal operator of order~$n$ for~$g$, i.e., $P(g)=0$. Since $L(f) = L(D_y(g)) = 0$, 
    we have~$P \mid L D_y$. Since $\mathcal{D}_y$ is a left Euclidean domain,  $P = \tilde{P}D_y + r$ for some~$r \in C(y)$ and~$\tilde{P}\in \mathcal{D}_y$ with~$\ord(\tilde{P}) < \ord(P)$. 
    If~$r = 0$, then~$\tilde{P} D_y(g) =\tilde{P}(f) = 0$, which contradicts with the minimality of $L$. Then $r \neq 0$. Since $P(g)=0$, we have
    \[P(g) = \tilde{P}D_y(g) + rg=0. \]
    This implies that $g = T(f)$ with $T=-\frac{1}{r} \tilde{P}$.
    \medskip
    
    $(2) \Rightarrow (3)$: If $f = D_y(T(f))$ for some $T \in \mathcal{D}_y$ of order~$\leq n-1$, then $L  \mid 1-D_y T$, which implies that there exists $r \in C(y)$ 
    such that~$r L = 1 - D_y T$, i.e., $1 = r L + D_y T $. Taking the adjoint operator on both sides of this equality yields 
    \[   1 = L^{*} \cdot r + T^{*}(-D_y).
    \]
    Acting the operator on~$1$ and noting that $D_y(1)=0$ yields $L^{*}(r) =1$.
         \medskip 
    
    $(3)\Rightarrow (1)$: Assume that there exists $r\in C(y)$  such that $L^{*}(r)=1$. In Lagrange's identity
    \[
        u L(v) -v L^{*}(u) = D_y\bigl( M(u,v)\bigr),
    \]
    we can take $v = f$ and~$u=r$. Then 
    \[
        r L(f) -f L^{*}(r) = -f = D_y(T \cdot f).
    \]
    It suffices to take~$g = -T(f)$. It is clear that $g$ satisfies an operator of order at most~$n$ because~$\dim_{C(y)} \mathcal{D}_y \cdot f =n$. 
    If~$\ord(g) < n$, then $f=D_y(g)$ will also have order~$<n$, which leads to a contradiction with the minimality of $L$.  Thus, $g$ is also a D-finite function of 
    the order~$n$.
\end{proof}

The problem of creative telescoping for bivariate D-finite functions is as follows.
\begin{pro}\label{PB:ctdfinite}
    Given a bivariate D-finite function $f(x, y)$, find a nonzero operator $P\in C(x)\langle D_x\rangle$ and $Q\in C(x, y)\langle D_x, D_y\rangle$
    such that~$P(f) = D_y(Q(f))$. Such an operator $P$ if exists is called a \emph{telescoper} for $f$.
\end{pro}
The existence of telescopers for bivariate D-finite functions is guaranteed by Theorem~\ref{thm:holoct}.
Chyzak's algorithm~\cite{chyzak00} can compute the minimal telescoper for a given bivariate D-finite function.  Let $f(x,y)$ be a D-finite function over~$C(x,y)$. Then 
\[
    \dim_{C(x,y)} \biggl( C(x,y)\langle D_x, D_y\rangle / I_f \biggr) < +\infty.
\]
Since the quotient module $C(x,y)\langle D_x, D_y\rangle/I_f$ is isomorphic to the module 
$M := C(x,y)\langle D_x, D_y\rangle\cdot f$, which is also 
a finitely dimensional vector space over $C(x, y)$.  
We now describe Chyzak's algorithm as follows (See~\cite[Section 5.4]{Kauers2023} for more discussions).

\begin{alg}[Chyzak's Algorithm]
\qquad

\noindent INPUT: A D-finite function $f\in M$ with $\dim_{C(x,y)} M < \infty$.

\noindent OUTPUT: A nonzero operator $P\in C(x)\langle D_x \rangle$ and $Q\in M$ such that~$P(f)=D_y(Q)$.
\begin{enumerate}
    \item Let $e_1, \ldots, e_d$ be a $C(x, y)$-vector space basis of $M$ and let $A = (a_{i, j})\in C(x, y)^{d\times d}$
    be such that $D_y(e_j) = \sum_{i=1}^{d} a_{i, j}e_i$ for $j=1, \ldots, d$.
    \item For~$r = 0,1,\ldots$
        \begin{enumerate}
            \item[2.1] Set $P = \sum_{k=0}^r c_k D_x^k$ and $Q = \sum_{i=1}^{d} q_i e_i$. Write $D_x^k(f) = \sum_{i=1}^{d} b_{k, i} e_i$.
            \item[2.2] Solve the coupled linear system (using Barkatou's algorithm~\cite{barkatou2000} for instance)
            \[
A\begin{pmatrix}
q_1\\
\vdots\\
q_d
\end{pmatrix}+\begin{pmatrix}
D_y(q_1)\\
\vdots\\
D_y(q_d)
\end{pmatrix}=c_0
\begin{pmatrix}
b_{0, 1}\\
\vdots\\
b_{0, d}
\end{pmatrix}
+\cdots + c_s
\begin{pmatrix}
b_{r, 1}\\
\vdots\\
b_{r, d}
\end{pmatrix}
\]
for $q_1, \ldots, q_d\in C(x, y)$ and $c_0, \ldots, c_r\in C(x)$.
            \item[2.3] If there is a solution with $(c_0, \ldots, c_r) \neq 0$, then \\
               \,\,\,  Return $(c_0 + \cdots + c_rD_x^r, q_1e_1 + \cdots + q_de_d)$.
        \end{enumerate}
\end{enumerate}
\end{alg}


\section{Examples and Applications}\label{sec:Examples}

In this section, we put the algorithms discussed previously into practice,
by applying them to concrete problems, some of which are artificial, but most
of which come from real-world applications. For demonstration purposes, we
employ the third author's Mathematica package
\texttt{HolonomicFunctions}~\cite{koutschan10c}.

The first step in a creative telescoping computation is usually the conversion
of some input function, which is given as a mathematical expression, into a
holonomic representation, i.e., a set of recurrence and/or differential
operators that annihilate the given expression---desirably a so-called (left)
Gr\"obner basis~\cite{Buchberger65,CoxLittleOShea92}.
This step can be executed automatically
by invoking the algorithms for closure properties (see Section~\ref{sec:DfiniteUniv}).
For atomic expressions, such as individual special functions, the defining equations
have to be looked up. All these steps are combined in the command
\texttt{Annihilator}.

\begin{ex}
  We look at the four expressions
  \[
    \operatorname{erf}\bigl(\sqrt{x+1}\bigr)^{\!2} + \exp\bigl(\sqrt{x+1}\bigr)^{\!2}
  \]
  \[
    \Bigl(\bigl(\sinh(x)\bigr)^{\!2} + \bigl(\sin(x)\bigr)^{\!-2}\Bigr) \cdot
    \Bigl(\bigl(\cosh(x)\bigr)^{\!2} + \bigl(\cos(x)\bigr)^{\!-2}\Bigr)
  \]
  \[
    \frac{\log\bigl(\sqrt{1-x^2}\bigr)}{\exp\bigl(\sqrt{1-x^2}\bigr)}
  \]
  \[
    \arctan\bigl(\mathrm{e}^x\bigr)
  \]
  and ask which of them are D-finite. With the knowledge that we already
  acquired in the course of this lecture, we can determine that the first
  and third expression have no reason not to be D-finite, while the
  second expression involves the reciprocals of the sine and cosine
  functions (not D-finite since they have infinitely many singularities),
  and the fourth expression substitutes $\mathrm{e}^x$, which is not an
  algebraic function. Indeed, calling
\begin{verbatim}
Annihilator[Erf[Sqrt[x + 1]]^2 + Exp[Sqrt[x + 1]]^2, Der[x]]
\end{verbatim}
  we obtain a differential operator that annihilates this expression:
  \begin{align*}
    \bigl\{ & (32 x^4+80 x^3+84 x^2+56 x+20) D_{x}^5
      +(96 x^4+336 x^3+348 x^2+252 x +144) D_{x}^4 \\
    & {} +(64 x^4+336 x^3+312 x^2+154 x+249) D_{x}^3
      +(-32 x^2-88 x +52) D_{x}^2 \\
    & {} +(-64 x^3-64 x^2+24 x-84) D_{x}\bigr\}
  \end{align*}
  proving that it is D-finite. Similarly
\begin{verbatim}
Annihilator[Log[Sqrt[1 - x^2]]/Exp[Sqrt[1 - x^2]], Der[x]]
\end{verbatim}
  yields
  \begin{align*}
    \bigl\{ & (4 x^{11}-17 x^9+27 x^7-19 x^5+5 x^3) D_{x}^4
      +(8 x^{10}-4 x^8-46 x^6+72 x^4 -30 x^2) D_{x}^3 \\
    & {} +(8 x^{11}-30 x^9+10 x^7+96 x^5-159x^3+75 x) D_{x}^2 \\
    & {} +(8 x^{10}-12 x^8+14 x^6-96 x^4+159 x^2-75) D_{x}
      +(4 x^{11}-13 x^9+7 x^7)\bigr\}.
  \end{align*}
  However, if we try to feed the last expression into the command
\begin{verbatim}
Annihilator[ArcTan[Exp[x]], Der[x]]
\end{verbatim}
  we obtain the following error messages, suggesting that the expression
  is most likely not D-finite:
  \begin{quote}
  \textbf{DFiniteSubstitute::algsubs:} The substitutions for continuous variables
  $\{\mathrm{e}^x\}$ are supposed to be algebraic expressions.
  Not all of them are recognized to be algebraic.
  The result might not generate a $\partial$-finite ideal. \\
  \textbf{Annihilator::nondf:} The expression (w.r.t. Der[x]) is not recognized to be
  $\partial$-finite. The result might not generate a zero-dimensional ideal.
  \end{quote}
\end{ex}

In the introduction (see Section~\ref{sec:intro}) we have already alluded
to an application in numerical analysis, where for the efficient implementation
of the simulation of electromagnetic waves, it was necessary to find certain
difference-differential relations of the basis functions
\[
  \varphi_{i,j}(x,y) :=
  (1-x)^i P_j^{(2 i+1,0)}(2 x-1) P_i\big(\textstyle\frac{2 y}{1-x}-1\big).
\]
As before, the first step is to derive a holonomic system satisfied by
$\varphi_{i,j}(x,y)$, which is done by the command
\begin{verbatim}
annphi = Annihilator[
  (1 - x)^i * JacobiP[j, 2i + 1, 0, 2x - 1] *
  LegendreP[i, 2y/(1 - x) - 1], {S[i], S[j], Der[x], Der[y]}]
\end{verbatim}
The output is a bit unhandy (about 0.5 MB in size) and therefore not printed
here, but the computation is rather quick (less than a second). Since the
\texttt{HolonomicFunctions} package returns a set of generators of the
annihilator ideal that form a left Gr\"obner basis, one can exploit the
nice properties of such bases for further processing. For example,
from the support of the operators, obtained by \texttt{Support[annphi]},
\begin{align*}
  \bigl\{ & \{S_{j}, D_{x}, D_{y}, 1\}, \{D_{y}^2, D_{y}, 1\},
  \{D_{x} D_{y}, S_{i}, D_{x}, D_{y}, 1\}, \{D_{x}^2,
  S_{i}, D_{x}, D_{y}, 1\},  \\ &\{S_{i} D_{y}, S_{i}, D_{x},
  D_{y}, 1\}, \{S_{i} D_{x}, S_{i}, D_{x}, D_{y}, 1\},
  \{S_{i}^2, S_{i}, D_{x}, D_{y}, 1\}\bigr\},
\end{align*}
one can deduce that this annihilator has holonomic rank~$4$, by reading off
the monomials $\{1,S_i,D_x,D_y\}$ that are not divisible by any of the
leading monomials and therefore form a vector space basis of the solution space.
We make an ansatz with undetermined coefficients and represent it
in that basis (this corresponds to reduction modulo the Gr\"obner basis);
the ansatz operator is in the ideal if and only if all the
coefficients of the remainder equal zero, which leads to a linear system
of equations. This strategy is implemented in the command \texttt{FindRelation},
which we use in the form
\begin{verbatim}
FindRelation[annphi, Eliminate -> {x, y},
  Pattern -> {_, _, 0 | 1, 0}]
\end{verbatim}
asking for an operator whose coefficients are free of~$x$ and~$y$ and
which may have arbitrary degrees in $S_i$ and $S_j$, but at most
degree~$1$ in~$D_x$ and no~$D_y$. The computation takes just a few seconds
and returns an operator that translates into the following recurrence
relation (note that this relation was crucial for the simulations and
therefore even entered a patent~\cite{SchoeberlKoutschanPaule15}):

\begin{align*}
  & (2i+j+3)(2i+2j+7)\tfrac{\mathrm{d}}{\mathrm{d} x}\varphi_{i,j+1}(x,y)
    +2(2i+1)(i+j+3)\tfrac{\mathrm{d}}{\mathrm{d} x}\varphi_{i,j+2}(x,y) \\
  & {}-(j+3)(2i+2j+5)\tfrac{\mathrm{d}}{\mathrm{d} x}\varphi_{i,j+3}(x,y)
    +(j+1)(2i+2j+7)\tfrac{\mathrm{d}}{\mathrm{d} x}\varphi_{i+1,j}(x,y) \\
  & {}-2(2i+3)(i+j+3)\tfrac{\mathrm{d}}{\mathrm{d} x}\varphi_{i+1,j+1}(x,y)
    -(2i+j+5)(2i+2j+5)\tfrac{\mathrm{d}}{\mathrm{d} x}\varphi_{i+1,j+2}(x,y) \\
  & {}+2(i+j+3)(2i+2j+5)(2i+2j+7)\varphi_{i,j+2}(x,y) \\
  & {}+2(i+j+3)(2i+2j+5)(2i+2j+7)\varphi_{i+1,j+1}(x,y) = 0.
\end{align*}

As we mentioned already before, one of the standard tables of identities
involving special functions is the book by Gradshteyn and
Ryzhik~\cite{GradshteynRyzhik14}. Opening it at a random page, we find
the identity
\[
  \int_{-1}^1 \bigl(1-x^2\bigr)^{\nu-\frac12} \mathrm{e}^{\mathrm{i}ax}
  C_n^\nu(x)\,dx =
  \frac{\pi \, 2^{1-\nu} \, \mathrm{i}^n \, \Gamma(2\nu+n)}{n!\,\Gamma(\nu)}
  a^{-\nu} J_{\nu+n}(a)
\]
that involves the Gamma function~$\Gamma(x)$, the Gegenbauer
polynomials~$C_n^{(\alpha)}(x)$, and the Bessel function~$J_\nu(x)$.
Let us see how creative telescoping can assist in proving this identity.
We start with the left-hand side, where creative telescoping delivers
a holonomic system for the integral. By typing the command
\begin{verbatim}
CreativeTelescoping[(1 - x^2)^(nu - 1/2) * Exp[I a x] *
  GegenbauerC[n, nu, x], Der[x], {S[n], Der[a]}]
\end{verbatim}
we receive the following output:
\begin{align*}
  \bigl\{ & \{(a n+a) S_{n}+(i a n+2 i a \nu) D_{a}+(-i n^2-2 i n \nu),\\
  & \quad a^2 D_{a}^2+(2 a \nu +a) D_{a}+(a^2-n^2-2 n \nu)\}, \\
  & \{i (n+1) S_{n}-i (n x+2 \nu x), (n+1) S_{n}-i (a x^2-a-i n x-2 i \nu x)\}
    \bigr\}
\end{align*}
(note that the output consists of a list of telescopers~$P$ and a list
of corresponding certificates~$Q$.)

Similarly, we aim at computing a holonomic system for the right-hand side,
which is much easier since no integral or summation is involved. We
can obtain it by pure application of closure properties, as demonstrated
above. The command
\begin{verbatim}
Annihilator[
  Pi * 2^(1 - nu) * I^n * Gamma[2 nu + n] / n! / Gamma[nu] *
  a^(-nu) * BesselJ[nu + n, a], {S[n], Der[a]}]
\end{verbatim}
almost instantaneously delivers
\begin{align*}
  \bigl\{ & (a n+a) S_{n}+(i a n+2 i a \nu) D_{a}+(-i n^2-2 i n \nu), \\
  & a^2 D_{a}^2+(2 a \nu +a) D_{a}+(a^2-n^2-2 n \nu)\bigr\}
\end{align*}
which turns out to be the same set of operators that was found for
the left-hand side by creative telescoping. Hence, the identity is
established by comparing a suitable amount of initial values (this
has to be done by hand and is not shown here explicitly).

In order to emphasize the versatility of the approach, we give a
collection of identities that can be proven by creative telescoping,
in a similar fashion as demonstrated above:
\[
  \sum_{k=0}^n\binom{n}{k}^{\!2}\binom{k+n}{k}^{\!2}=\sum_{k=0}^n\binom{n}{k}\binom{k+n}{k}\sum_{j=0}^k\binom{k}{j}^{\!3}
\]
\[
  \int_0^\infty \frac{1}{\left(x^4+2 a x^2+1\right)^{m+1}} \, dx=\frac{\pi
  P_m^{\left(m+\frac{1}{2},-m-\frac{1}{2}\right)}(a)}{2^{m+\frac {3}{2}}(a+1)^{m+\frac{1}{2}}}
\]
\[
  e^{-x} x^{a/2} n!\, L_n^a(x)=\int_0^{\infty } e^{-t} t^{\frac{a}{2}+n} J_a\big(2 \sqrt{t x}\big) dt
\]
\[
  \int_{-\infty}^\infty\sum_{m=0}^\infty\sum_{n=0}^\infty\frac{H_m(x)H_n(x)r^ms^ne^{-x^2}}{m!\,n!} \, dx
  =\sqrt{\pi } e^{2 r s}
\]

Our next application, the \emph{holonomic ansatz}~\cite{Zeilberger07},
is a computer-algebra-based approach to find and/or prove the evaluation
of a symbolic determinant~$\det(A_n)$, where the dimension of the square
matrix $A_n:=(a_{i,j})_{0\leq i,j<n}$ is given by a symbolic parameter~$n$.
The method is applicable
to non-singular matrices whose entries~$a_{i,j}$ are holonomic sequences
in the index variables~$i$ and~$j$. Moreover, the entries~$a_{i,j}$
must not depend on~$n$, i.e., $A_{n-1}$ is an upper-left submatrix of~$A_n$.

The holonomic ansatz works as follows: define the quantity
\begin{equation}\label{cnj}
  c_{n,j} := (-1)^{n-1+j} \frac{M_{n-1,j}}{M_{n-1,n-1}}
\end{equation}
where $M_{i,j}$ denotes the $(i,j)$-minor of the matrix~$A_n$.
In other words, $c_{n,j}$ is the $(n-1,j)$-cofactor
of $A_n$ divided by $\det(A_{n-1})$. Using Laplace expansion with respect to
the last row, one can write
\begin{equation}\label{H3}
  \sum_{j=0}^{n-1} a_{n-1,j} c_{n,j} = \frac{\det(A_n)}{\det(A_{n-1})}.
\end{equation}
Under the assumptions that (i) the bivariate sequence $c_{n,j}$ is holonomic
and that (ii) its holonomic definition is known, the symbolic sum on the
left-hand side of~\eqref{H3} can be tackled with creative
telescoping, yielding a linear recurrence in~$n$ for the sum.
If an evaluation~$b_n$ for the determinant of~$A_n$ is conjectured,
then one can prove it by verifying that $b_n/b_{n-1}$ satisfies the
obtained recurrence and by comparing a sufficient number of initial values.

What can be said about the two assumptions? There is no general theorem that
implies that $c_{n,j}$ is always holonomic, and in fact, there are many
examples where it is not. If (i) is not satisfied, i.e., if $c_{n,j}$ is not
holonomic, then the method fails (not necessarily; in some situations one may
succeed to overcome the problem by applying a mild
reformulation; see~\cite{KoutschanNeumuellerRadu16}). Concerning~(ii): by a holonomic
definition we mean a set of linear recurrence equations whose coefficients
are polynomials in the sequence indices~$n$ and~$j$, together with finitely
many initial values, such that the entire bivariate sequence
$(c_{n,j})_{1\leq n,\,0\leq j<n}$ can be produced by unrolling the recurrences
and by using the initial values. The question now is how the original
definition~\eqref{cnj} can be converted into a holonomic definition.

Clearly, \eqref{cnj} allows one to compute the values of $c_{n,j}$ for
concrete integers~$n$ and~$j$ in a certain, finite range. From these data,
candidate recurrences can be constructed by the method of guessing (i.e.,
employing an ansatz with undetermined coefficients; cf.~\cite{Kauers09}). It
remains to prove that these recurrences, constructed from finite, and
therefore incomplete data, are correct, i.e., are valid for all $n\geq1$ and
$0\leq j<n$. For this purpose, we show that $c_{n,j}$ is the unique solution
of a certain system of linear equations, and then we prove that the sequence defined by
the guessed recurrences (and appropriate initial conditions) also satisfies the
same system. By uniqueness, it follows that the two sequences agree,
i.e., that the guessed recurrences define the desired sequence~$c_{n,j}$.

Suppose that the last row of $A_n$ is replaced by its $i$-th row;
the resulting matrix is clearly singular, turning~\eqref{H3} into
\begin{equation}\label{H2}
  \sum_{j=0}^{n-1} a_{i,j} c_{n,j} = 0\qquad (0\leq i<n-1).
\end{equation}
For each $n\in\set N$ the above equation~\eqref{H2} represents a system
of $n-1$ linear equations in the $n$ ``unknowns'' $c_{n,0},\dots,c_{n,n-1}$,
whose coefficient
matrix $(a_{i,j})_{0\leq i<n-1,0\leq j<n}$ has full rank because
$\det(A_{n-1})\neq0$ (if the latter is not known a priori, it can be argued by
induction on~$n$). Hence the homogeneous system~\eqref{H2} has a
one-dimensional kernel. The solution is made unique by normalizing with respect to its
last component, that is, by imposing a condition that is obvious
from~\eqref{cnj}, namely
\begin{equation}\label{H1}
  c_{n,n-1} = 1.
\end{equation}
Hence, \eqref{H2} and~\eqref{H1} together define $c_{n,j}$ uniquely. On the
other hand, given a holonomic definition of $c_{n,j}$, creative telescoping
and holonomic closure properties can be applied to prove~\eqref{H1}
and~\eqref{H2}, respectively. If these proofs succeed, then it follows that
the guessed recurrences are correct.

The holonomic ansatz has already been applied in many different contexts
\cite{KKZ2011, KoutschanThanatipanonda19, DuKoutschanThanatipanondaWong22};
here we want to showcase its most
recent application~\cite{KoutschanKrattenthalerSchlosser25} to a
real-world problem: in his study of the twenty-vertex model, Di
Francesco~\cite[Conj.~8.1 + Thm.~8.2]{DiFran21} came up with the
following conjectured determinant evaluation:
\begin{equation}\label{eq:DiFran1}
  \det_{0\leq i,j<n} \left(2^i \binom{i+2j+1}{2j+1} - \binom{i-1}{2j+1}\right) =
  2\prod_{i=1}^n \frac{2^{i-1} \, (4i-2)!}{(n+2i-1)!}.
\end{equation}
For computing a sufficient number of values for~$c_{n,j}$, it is more
efficient to employ their definition via~\eqref{H1} and~\eqref{H2},
rather than computing determinants in the spirit of~\eqref{cnj}. Then
we invoke the Mathematica package \texttt{Guess.m}~\cite{Kauers09},
which delivers three recurrence relations for the quantities~$c_{n,j}$,
whose shape suggests that $c_{n,j}$ indeed is a holonomic sequence.
The recurrence operators are too big to be displayed here (they would
require approximately one page), so we give only their supports instead:
\begin{equation}\label{eq:supp}
  \{S_j^2, S_n, S_j, 1\}, \quad
  \{S_n S_j, S_n, S_j, 1\}, \quad
  \{S_n^2, S_n, S_j, 1\}.
\end{equation}
These three operators generate the annihilator ideal~$I$ of~$c_{n,j}$
(and have the desirable property of forming a left Gr\"obner basis);
in the code, we denote them by \texttt{annc}.

We want to show that the guessed recurrences (represented by~$I$)
produce the correct values of $c_{n,j}$ for all~$j$ with
$0\leq j<n$. For this
purpose, we introduce another sequence $\tilde{c}_{n,j}$ that is defined
via~$I$, and we show that it actually agrees with the
sequence~$c_{n,j}$. The latter will be done by verifying that~\eqref{H1}
and~\eqref{H2} hold when $c_{n,j}$ is replaced by~$\tilde{c}_{n,j}$.

From the leading monomials $S_j^2,\, S_nS_j,\, S_n^2$ in~\eqref{eq:supp} one
can deduce that the holonomic rank of~$I$ is three, since there are
exactly three irreducible monomials that are not divisible by any of the
leading monomials: $1,\,S_j,\,S_n$. Stated differently, one needs to specify
the initial values $\tilde{c}_{1,0},\,\tilde{c}_{1,1},\,\tilde{c}_{2,0}$ in
order to fix a particular solution of the annihilator~$I$. Hence,
we define $\tilde{c}_{n,j}$ to be the unique solution of~$I$
whose three initial values agree with~$c_{n,j}$.

From this definition of $\tilde{c}_{n,j}$ one can derive algorithmically a
(univariate) recurrence for the almost-diagonal sequence
$\tilde{c}_{n,n-1}$ by the following command
\begin{verbatim}
DFiniteSubstitute[annc, {j -> n - 1}]
\end{verbatim}
This recurrence has order~$3$, which is equal to the
holonomic rank of~$I$, as expected.  The corresponding operator
has the right factor $S_n-1$, and more precisely, it can be written in the form
\begin{multline*}
  \bigl(9 (n+4) (2 n+5) (3 n+2) (3 n+4) (3 n+5) (3 n+7) p_1(n) S_n^2 \\
  + 12 (3 n+2) (3 n+4) (4 n+3) (4 n+5) p_2(n) S_n \\
  - 16 n (2 n+1) (4 n-1) (4 n+1) (4 n+3) (4 n+5) p_1(n+1)\bigr) \cdot (S_n - 1),
\end{multline*}
where $p_1(n)$ and $p_2(n)$ are irreducible polynomials of degree~$9$
and~$11$, respectively. It follows that any constant sequence is a solution
of this recurrence.  Together with the initial conditions
$\tilde{c}_{1,0}=\tilde{c}_{2,1}=\tilde{c}_{3,2}=1$, which are easy to
check, this proves that $\tilde{c}_{n,n-1}=1$ holds for all $n\geq1$.

The proof of the summation identity~\eqref{H2} is achieved by the method
of creative telescoping. For reasons of efficiency,
we split the sum in \eqref{H2} into two sums as follows:
\[
  \sum_{j=0}^{n-1} a_{i,j} \tilde{c}_{n,j} =
  \sum_{j=0}^{n-1} 2^i \binom{i+2j+1}{2j+1} \tilde{c}_{n,j}
  - \sum_{j=0}^{n-1} \binom{i-1}{2j+1} \tilde{c}_{n,j}.
\]
For the first of the two sums, we obtain an annihilator ideal by typing
\begin{verbatim}
annci = OreGroebnerBasis[Append[annc, S[i] - 1], 
  OreAlgebra[S[n], S[j], S[i]]];
annSmnd1 = DFiniteTimes[Annihilator[
  2^i * Binomial[i + 2 j + 1, 2 j + 1],
  {S[n], S[j], S[i]}], annci];
id2fct1 = FindCreativeTelescoping[annSmnd1, S[j] - 1];
\end{verbatim}
and analogously for the second sum (each of these computations takes
about 20 minutes). The ideal is generated by
four operators whose supports are as follows:
\begin{align*}
  & \{S_i^3, S_n^2, S_nS_i , S_i^2, S_n, S_i, 1\}, \quad
    \{S_i^2 S_n, S_n^2, S_nS_i, S_i^2, S_n, S_i, 1\}, \\
  & \{S_i S_n^2, S_n^2, S_nS_i, S_i^2, S_n, S_i, 1\}, \quad
    \{S_n^3, S_n^2, S_nS_i, S_i^2, S_n, S_i, 1\}.
\end{align*}
Actually, the two sums are annihilated by the very same operators, hence
these operators constitute an annihilator for the left-hand side of~\eqref{H2}.
The leading terms of the operators have the form:
\begin{align*}
  & 12 (i-1) i (i+1) (3 n+1) (3 n+4) (4 n-1) (4 n+1) (i-n+3) (i-n+4) q_1 S_i^3, \\
  & {-9} i (3 n-1) (3 n+1) (3 n+4) q_2 S_nS_i^2, \\
  & {-18} (i-1) i (n+1) (2 n+3) (3 n-1) (3 n+1)^2 (3 n+2) (3 n+4) (i+2 n+5) q_3 S_n^2S_i, \\
  & {-54} (n+1) (n+2) (2 n+3) (2 n+5) (3 n-1) (3 n+1)^2 (i-2 n-6) (i-2 n-5) q_4 S_n^3,
\end{align*}
where $q_1,q_2,q_3,q_4$ are (not necessarily irreducible) polynomials in~$n$ and~$i$.
It remains to check a finite set of initial values. The shape of this set is
determined by the support displayed above, by the condition $i<n-1$, and
by the zeros of the leading coefficients of the operators. More precisely
we have to verify that $\sum_{j=0}^{n-1} a_{i,j} \tilde{c}_{n,j} = 0$ for
\begin{align*}
  (i,n) \in \{&(0, 2), (0, 3), (0, 4), (0, 5), (0, 6), (1, 3), (1, 4), (1, 5), (2, 4), \\
              & (1, 6), (1, 7), 
                (2, 5), (2, 6), 
                (2, 7), (2, 8), 
                (3, 5), (4, 6)\} 
\end{align*}
(where the points in the first line are determined by the support, and the
second line is determined by the zeros of the leading coefficients).
This verification is successful, and hence it follows that $\tilde{c}_{n,j}=c_{n,j}$
for all~$j$ with $0\leq j<n$, which allows us to use $I$ as a holonomic
definition for~$c_{n,j}$.

In order to derive a recurrence for the left-hand side of~\eqref{H3} we
split the sum into two sums, as before:
\[
  \sum_{j=0}^{n-1} a_{n-1,j} c_{n,j} =
  \sum_{j=0}^{n-1} 2^{n-1} \binom{n+2j}{2j+1} c_{n,j}
  - \sum_{j=0}^{n-1} \binom{n-2}{2j+1} c_{n,j}.
\]
Then we compute, for each of the two
sums, a recurrence by creative telescoping:
\begin{verbatim}
annSmnd1 = DFiniteTimes[Annihilator[2^(n - 1) *
  Binomial[n + 2 j, 2 j + 1], {S[n], S[j]}], annc];
id3fct1 = FindCreativeTelescoping[annSmnd1, S[j] - 1];
annSmnd2 = DFiniteTimes[Annihilator[
  Binomial[n - 2, 2 j + 1], {S[n], S[j]}], annc];
id3fct2 = FindCreativeTelescoping[annSmnd2, S[j] - 1];
rec = DFinitePlus[id3fct1[[1]], id3fct2[[1]]];
\end{verbatim}
In both cases, the output is a recurrence \texttt{rec}
of order~$6$ with polynomial coefficients of degree~$52$.
Actually one finds that both sums satisfy the same
order-$6$ recurrence, and hence so does their sum.  One now has to verify
that $b_n/b_{n-1}$ satisfies this order-$6$ recurrence, where $b_n$ denotes
the right-hand side of \eqref{eq:DiFran1}. We have
\[
  \frac{b_n}{b_{n-1}} =
  \frac{(4n-2)!}{(3n-1)! \, \bigl(\frac{n+1}{2}\bigr)_{n-1}}.
\]
Note that this expression is hypergeometric in $n/2$ and hence satisfies a
second-order recurrence whose operator has support $\{S_n^2,1\}$:
\begin{verbatim}
annqb = Annihilator[(4 n - 2)! / (3 n - 1)! /
  Pochhammer[(n + 1)/2, n - 1], S[n]];
Factor[annqb]
\end{verbatim}
yields
\begin{align*}
  \bigl\{ & 27 (3 n-1) (3 n+1)^2 (3 n+2) (3 n+4) (3 n+5) S_{n}^2-{} \\
  & 256 (2 n+1) (2 n+3) (4 n-1) (4 n+1) (4 n+3) (4 n+5)\bigr\}
\end{align*}
Right-dividing the operator \texttt{rec} by the second-order operator \texttt{annqb}
can be performed by
\begin{verbatim}
OreReduce[rec, annqb]
\end{verbatim}
which returns~$0$, hence \texttt{rec} annihilates $b_n/b_{n-1}$.
It now suffices to verify
\[
  \frac{\det_{0\leq i,j\le n-1}(a_{i,j})}{\det_{0\leq i,j\le n-2}(a_{i,j})} = \frac{b_n}{b_{n-1}}
\]
for $n=2,\dots,7$. On both sides, one calculates the values $4$, $15$, $832/15$,
$204$, $9728/13$, $16445/6$, respectively.  By virtue of the recurrence \texttt{rec},
the asserted identity~\eqref{eq:DiFran1} holds for all integers $n\geq1$.

\section{Conclusion}

Creative telescoping is a key technique in symbolic summation and integration that has been the subject
of intensive research during the past 35 years. The purpose of this introductory article, like the purpose
of the introductory course at the RTCA special semester on which it is based, was not to summarize the
current state of the art but to only explain the basic principles and the main results on the matter.
We have discussed the summation case for hypergeometric terms and the integration case for rational
functions in some detail and then turned to the more general concepts of D-finiteness and holonomy,
where the method of creative telescoping applies as well and finds many additional applications.

We hope that the reader got a sense what a telescoper is, why it is useful, and how it can be computed.

In this concluding section, we will briefly comment on more advanced topics related to
creative telescoping that we did not cover in the RTCA course.
One such aspect concerns the use of differential or difference fields.

A \emph{derivation} on a field $K$ is a map $D\colon K\to K$ that satisfies
\[
  D(a+b)=D(a)+D(b)
  \quad\text{and}\quad
  D(ab)=D(a)b+aD(b)
\]
for all $a,b\in K$. A field $K$ together with a derivation is called a \emph{differential field.}

If $E$ is a differential field and $K$ is a subfield of~$E$, then an element $e$ of $E$ is called
a \emph{primitive} (over~$K$) if $D(e)=u$ for some $u\in K$, and it is called
\emph{hyperexponential} (over~$K$) if $D(e)/e=u$ for some $u\in K$. A differential field
$K=C(t_1,\dots,t_d)$ is called \emph{liouvillian} if every $t_i$ is primitive or hyperexponential
over $C(t_1,\dots,t_{i-1})$ and for all $r\in K$ we have $D(r)=0$ if and only if $r\in C$.

\begin{ex}
  Consider the expression $x-\log(1+\exp(x))$. We want to construct a differential field which
  contains an element that behaves like this expression upon differentiation.
  To this end, we take $K=C(t_1,t_2,t_3)$ and define $D(c)=0$ for all $c$ as well as
  \begin{alignat*}3
    D(t_1) &= 1 &\quad&\text{so that $t_1$ behaves like $x$}\\
    D(t_2) &= t_2 &&\text{so that $t_2$ behaves like $\exp(x)$, and}\\
    D(t_3) &= \frac{t_2}{1 + t_2} &&\text{so that $t_3$ behaves like $\log(1+\exp(x))$}.
  \end{alignat*}
  Then $t_1-t_3$ behaves like $x-\log(1+\exp(x))$.

  Note that $D$ is completely specified by its values on $C$ and the generators $t_1,t_2,t_3$.
  Note also that while it is clear by the definition of $C$ that $D(c)=0$
  for all~$c$, it is not obvious that there is no element in $K\setminus C$ whose
  derivative is zero. However, it can be checked that this is the case. 
\end{ex}

The celebrated Risch algorithm for indefinite integration~\cite{Risch1969,Risch1970,BronsteinBook} solves the integration
problem in liouvillian fields: given a liouvillian field $K$ and an element $f\in K$, it constructs a
liouvillian field $E$ containing $K$ and an element $g\in E$ such that $D(g)=f$, or it proves
that no such $E$ exists. For example, it can find
\[
  \int\frac{dx}{1+\exp(x)}=x-\log(1+\exp(x))
\]
and prove that
\[
  \int\frac{dx}{x+\exp(x)}\text{ is not elementary.}
\]
The Risch algorithm reduces the given integration problem to an integration problem in a smaller field,
which is then solved recursively.

Within the Risch algorithm, the following parameterized version of the integration problem has
to be solved:
\begin{itemize}
\item given $f_1,\dots,f_r\in K$
\item find $c_1,\dots,c_r\in C$ and $g\in K$ such that
\[
  c_1 f_1 + \cdots + c_r f_r = D(g).
\]
\end{itemize}
Although Risch's algorithm focusses on indefinite integration, this problem formulation looks
remarkably similar to the specification of creative telescoping. Indeed, as pointed out by
Raab~\cite{raab12,Raab2016}, we can construct a creative telescoping procedure based on Risch's algorithm
in a similar way as Zeilberger's algorithm is based on Gosper's algorithm for indefinite
summation. As a result, we obtain a method for handling definite integrals over liouvillian
functions.

There is also an analogous theory for the discrete case. A \emph{difference field} is a field
$K$ together with a field automorphism~$\sigma$, i.e., a map $\sigma\colon K\to K$ with
\[
\sigma(a+b)=\sigma(a)+\sigma(b)
\quad\text{and}\quad
\sigma(ab)=\sigma(a)\sigma(b)
\]
for all $a,b\in K$.
This automorphism plays the role of a shift.

If $E$ is a difference field and $K$ is a subfield of~$E$, then an element $e$ of $E$ is called
a \emph{sum} (over~$K$) if $\sigma(e)-e=u$ for some $u\in K$, and it is called
a \emph{product} (over~$K$) if $\sigma(e)/e=u$ for some $u\in K$. A difference field
$K=C(t_1,\dots,t_d)$ is called a \emph{$\Pi\Sigma$-field} if every $t_i$ is either a
sum or a product over $C(t_1,\dots,t_{i-1})$ and for all $r\in K$ we have $\sigma(r)=r$ if
and only if $r\in C$.

\begin{ex}
  Consider the expression $\sum_{k=1}^n\frac{1+2^k}{1 + \sum_{i=1}^k\frac1i}$. 
  We want to construct a difference field which contains an element that behaves
  like this expression upon shift.
  To this end, we take $K=C(t_1,t_2,t_3,t_4)$ and define $\sigma(c)=c$ for all $c$ as well as
  \begin{alignat*}3
    \sigma(t_1) &= t_1+1 &&\text{so that $t_1$ behaves like $n$}\\
    \sigma(t_2) &= 2t_2 &&\text{so that $t_2$ behaves like $2^n$}\\
    \sigma(t_3) &= t_3 + \frac1{t_1+1} &&\text{so that $t_3$ behaves like $\sum_{k=1}^n\frac1k$, and}\\
    \sigma(t_4) &= t_4 + \frac{1+2t_2}{1+t_3+\frac1{t_1+1}} &\quad&\text{so that $t_4$ behaves like the target expression}.
  \end{alignat*}
  Again, while it is clear by definition that $\sigma(c)=c$ for all $c\in C$, it is not
  obvious (but true) that there is not also some element $r$ in $K\setminus C$ with
  $\sigma(r)=r$.
\end{ex}

$\Pi\Sigma$-fields were introduced by Karr~\cite{Karr1981,Karr1985}, who used them to formulate a counterpart
of Risch's algorithm for the summation case.
Like Risch's algorithm, also Karr's algorithm includes a subroutine that can be used to do
creative telescoping.
This was first observed by Schneider~\cite{Schneider2004}, who has since extended the algorithmic theory
of $\Pi\Sigma$-fields in many directions. 

For many years, the line of research extending creative telescoping to functions
described by annihilating operators and the line of research extending creative telescoping
to functions described by differential or difference fields were developped side by side
with surprisingly little interaction. It would be interesting and challenging to combine
these two trends into a unified theory.

With the line of research that focussed on operator techniques, much of the research
during the past years has been devoted to reduction-based telescoping techniques.
In Sect.~\ref{sec:ct}, we explained the idea of this technique for rational functions
in the differential case. The technique was extended to hyperexponential functions~\cite{bostan13a},
to rational functions in more variables~\cite{bostan13b}, to algebraic functions~\cite{chen16a},
to hypergeometric terms (in the summation case)~\cite{chen15a,huang16},
to problems involving discrete as well as continuous variables~\cite{bostan16},
and finally to the
general D-finite setting~\cite{vdHoeven2017,vdHoeven2018,bostan18a,vanderHoeven21,brochet24,chen25,du25}.
In fact, there are two competing
approaches to reduction-based telescoping for D-finite functions, one based on
Lagrange's identity and one based on integral bases. For the time being, it is
unclear which of the two is superior. 

The extension of reduction-based creative telescoping to further function classes is
a subject of ongoing research. Some further research problems related to creative
telescoping were raised in a paper by Chen and Kauers~\cite{ChenKauers2017}. 


\bibliographystyle{spmpsci}
\bibliography{ct}
